\documentclass[11pt]{article}
\usepackage{amssymb, amsmath}
\usepackage{fullpage,subfigure}

\usepackage{hyperref}
\usepackage{graphicx,amsmath,amsfonts,amscd,amssymb,bm,cite,epsfig,epsf,url,color}
\usepackage{fullpage} 
\usepackage[small,bf]{caption}
\newtheorem{theorem}{Theorem}[section]
\newtheorem{lemma}[theorem]{Lemma}
\newtheorem{corollary}[theorem]{Corollary}
\newtheorem{proposition}[theorem]{Proposition}
\newtheorem{definition}[theorem]{Definition}

\newcommand{\R}{\mathbb{R}}
\newcommand{\C}{\mathbb{C}}

\newcommand{\Z}{\mathbb{Z}}
\newcommand{\<}{\langle}
\renewcommand{\>}{\rangle}

\newcommand{\goto}{\rightarrow}

\newcommand{\tr}{\operatorname{Tr}}

\newcommand{\Id}{\text{\em I}}

\numberwithin{equation}{section}

\def \endprf{\hfill {\vrule height6pt width6pt depth0pt}\medskip}
\newenvironment{proof}{\noindent {\bf Proof} }{\endprf\par}

\newcommand{\srf}{\text{SRF}}

\newcommand{\PWn}{\mathcal{P}_{W}}

\newcommand{\normInf}[1]{\left|\left| #1 \right|\right| _{\infty}}
\newcommand{\normTwo}[1]{\left|\left| #1 \right|\right| _{2}}
\newcommand{\normOne}[1]{\left|\left| #1 \right|\right| _{1}}

\newcommand{\normInfInf}[1]{\left|\left| #1 \right|\right| _{\infty}}

\newcommand{\normTV}[1]{\left|\left| #1 \right|\right| _{\text{TV}}}

\newcommand{\abs}[1]{\left| #1 \right|}
\newcommand{\keys}[1]{\left\{ #1 \right\}}
\newcommand{\sqbr}[1]{\left[ #1 \right]}
\newcommand{\brac}[1]{\left( #1 \right) }
\newcommand{\MAT}[1]{\begin{bmatrix} #1 \end{bmatrix}}

\newcommand{\PROD}[2]{\left \langle #1, #2\right \rangle}

\newcommand{\derTwo}[2]{\frac{\text{d}^2#2}{\text{d}#1^2}}

\newcommand{\Imag}[1]{\text{Im}\brac{ #1 }}
\newcommand{\Real}[1]{\text{Re}\brac{ #1 }}

\newcommand{\optvalue}{2}
\newcommand{\optvaluetimestwo}{4}
\newcommand{\optvaluetwoD}{2.38}
\newcommand{\minm}{128}
\newcommand{\minmTwoD}{512}

\newcommand{\tC}{0.1649}
\newcommand{\tp}{t_{+}}
\newcommand{\tm}{t_{-}}
\newcommand{\tA}{0.7559}
\newcommand{\tx}{0.4269}

\newcommand{\rOne}{0.2447}
\newcommand{\rTwo}{0.84}
\newcommand{\fc}{f_c}
\newcommand{\Deltamin}{\Delta_{\text{min}}}
\newcommand{\Deltaminth}{\Delta_{\text{\em min}}}

\newcommand{\Ktwo}{K^{\text{2D}}}

\author{Emmanuel J. Cand\`{e}s\thanks{Departments of Mathematics and
    of Statistics, Stanford University, Stanford CA} \,\, and Carlos
  Fernandez-Granda\thanks{Department of Electrical Engineering,
    Stanford University, Stanford CA}}

\title{Towards a Mathematical Theory of Super-Resolution
}
\date{March 2012; Revised June 2012}

\begin{document}

\maketitle
\vspace{-0.3in}

\begin{abstract}
  This paper develops a mathematical theory of
  super-resolution. Broadly speaking, super-resolution is the problem
  of recovering the fine details of an object---the high end of its
  spectrum---from coarse scale information only---from samples at the
  low end of the spectrum.  Suppose we have many point sources at
  unknown locations in $[0,1]$ and with unknown complex-valued
  amplitudes. We only observe Fourier samples of this object up until
  a frequency cut-off $f_c$. We show that one can super-resolve these
  point sources with infinite precision---i.e.~recover the exact
  locations and amplitudes---by solving a simple convex optimization
  problem, which can essentially be reformulated as a semidefinite
  program. This holds provided that the distance between sources is at
  least $2/f_c$. This result extends to higher dimensions and other
  models. In one dimension for instance, it is possible to recover a
  piecewise smooth function by resolving the discontinuity points with
  infinite precision as well. We also show that the theory and methods
  are robust to noise. In particular, in the discrete setting we
  develop some theoretical results explaining how the accuracy of the
  super-resolved signal is expected to degrade when both the noise
  level and the {\em super-resolution factor} vary.
\end{abstract}


\section{Introduction}
\label{sec:intro}

\subsection{Super-resolution}

Super-resolution is a word used in different contexts mainly to design
techniques for enhancing the resolution of a sensing system. Interest
in such techniques comes from the fact that there usually is a
physical limit on the highest possible resolution a sensing system can
achieve. To be concrete, the spatial resolution of an imaging device
may be measured by how closely lines can be resolved. For an optical
system, it is well known that resolution is fundamentally limited by
diffraction. In microscopy, this is called the Abbe diffraction limit
and is a fundamental obstacle to observing sub-wavelength structures.
This is the reason why resolving sub-wavelength features is a crucial
challenge in fields such as astronomy \cite{astronomy_puschmann},
medical imaging \cite{medical_greenspan}, and microscopy
\cite{microscopy_mccutchen}.  In electronic imaging, limitations stem
from the lens and the size of the sensors, e.~g.~pixel size. Here,
there is an inflexible limit to the effective resolution of a whole
system due to photon shot noise which degrades image quality when
pixels are made smaller.  Some other fields where it is desirable to
extrapolate fine scale details from low-resolution data---or resolve
sub-pixel details---include spectroscopy
\cite{spectroscopy_superresolution}, radar \cite{radar_odendaal},
non-optical medical imaging \cite{pet_kennedy} and geophysics
\cite{seismic_khaidukov}. For a survey of super-resolution techniques
in imaging, see \cite{imaging_park} and the references
therein.


This paper is about super-resolution, which loosely speaking is the
process whereby the fine scale structure of an object is retrieved
from coarse scale information only.\footnote{We discuss here
  computational super-resolution methods as opposed to instrumental
  techniques such as interferometry.} A useful mathematical model may
be of the following form: start with an object $x(t_1,t_2)$ of
interest, a function of two spatial variables, and the point-spread
function $h(t_1,t_2)$ of an optical instrument. This instrument acts
as a filter in the sense that we may observe samples from the
convolution product
\[
y(t_1,t_2) = (x * h)(t_1, t_2).
\]
In the frequency domain, this equation becomes 
\begin{equation}
\label{eq:low-pass}
\hat y(\omega_1, \omega_2) = \hat x(\omega_1, \omega_2) \hat
h(\omega_1, \omega_2), 
\end{equation}
where $\hat x$ is the Fourier transform of $x$, and $\hat h$ is the
{\em modulation transfer function} or simply {\em transfer function}
of the instrument. Now, common optical instruments act as low-pass filters
in the sense that their transfer function $\hat h$ vanishes for all
values of $\omega$ obeying $|\omega| \ge \Omega$ in which $\Omega$ is
a frequency cut-off; that is, 
\[
|\omega| := \sqrt{\omega_1^2 + \omega_2^2} > \Omega \quad \Rightarrow
\quad \hat h(\omega_1, \omega_2) = 0.
\]
In microscopy with coherent illumination, the bandwidth $\Omega$ is
given by $\Omega = 2\pi \text{NA}/\lambda$ where $\text{NA}$ is the
numerical aperture and $\lambda$ is the wavelength of the illumination
light. For reference, the transfer function in this case is simply the
indicator function of a disk and the point-spread function has
spherical symmetry and a radial profile proportional to the ratio
between the Bessel function of the first order and the radius. In
microscopy with incoherent light, the transfer function is the Airy
function and is proportional to the square of the coherent point-spread
function. Regardless, the frequency cut-off induces a physical
resolution limit which is roughly inversely proportional to $\Omega$
(in microscopy, the Rayleigh resolution distance is defined to be
$0.61 \times 2\pi/\Omega$).

\begin{figure}
\centering
\subfigure[]{
\includegraphics[width=5.5cm]{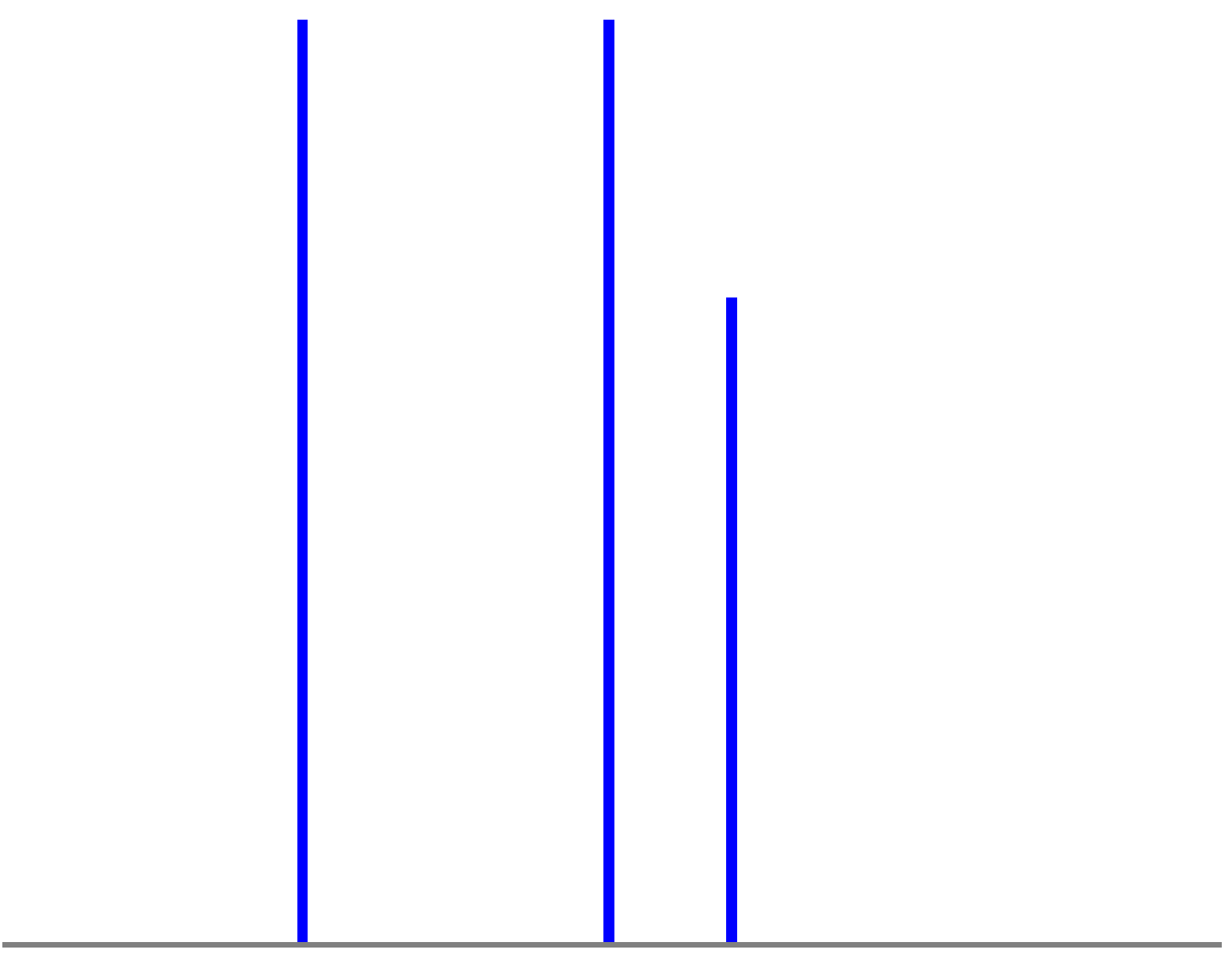}
\label{fig:example_a}
}
\subfigure[]{
\includegraphics[width=5.5cm,height=3cm]{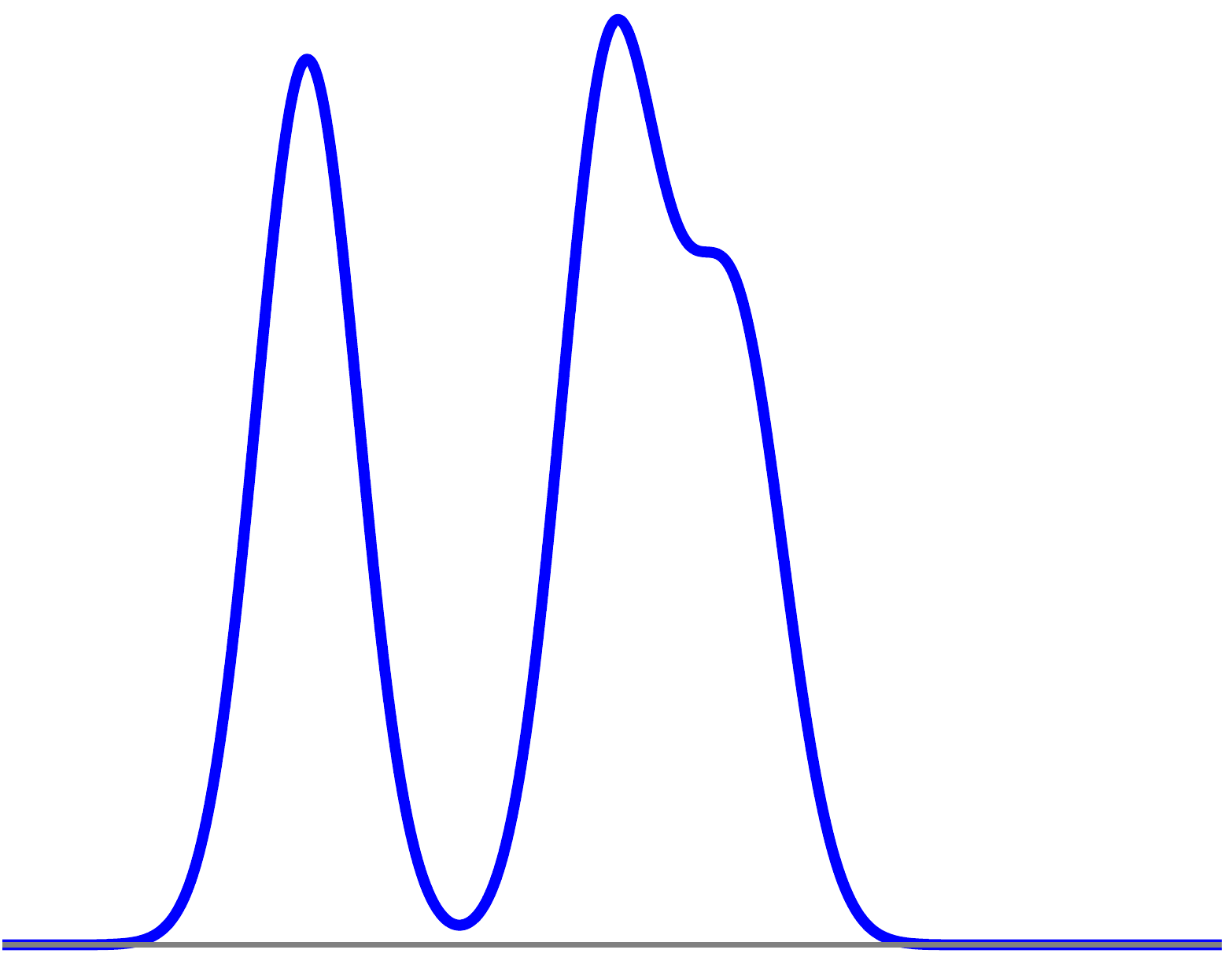}
}
\caption{Schematic illustration of the super-resolution problem in the
  spatial domain. (a) Highly resolved signal. (b) Low-pass version of
  signal in (a).  Super-resolution is about recovering the signal in
  (a) by deconvolving the data in (b).}
\label{fig:example}
\end{figure}
\begin{figure}
\centering
\subfigure[]{
\includegraphics[width=5.5cm]{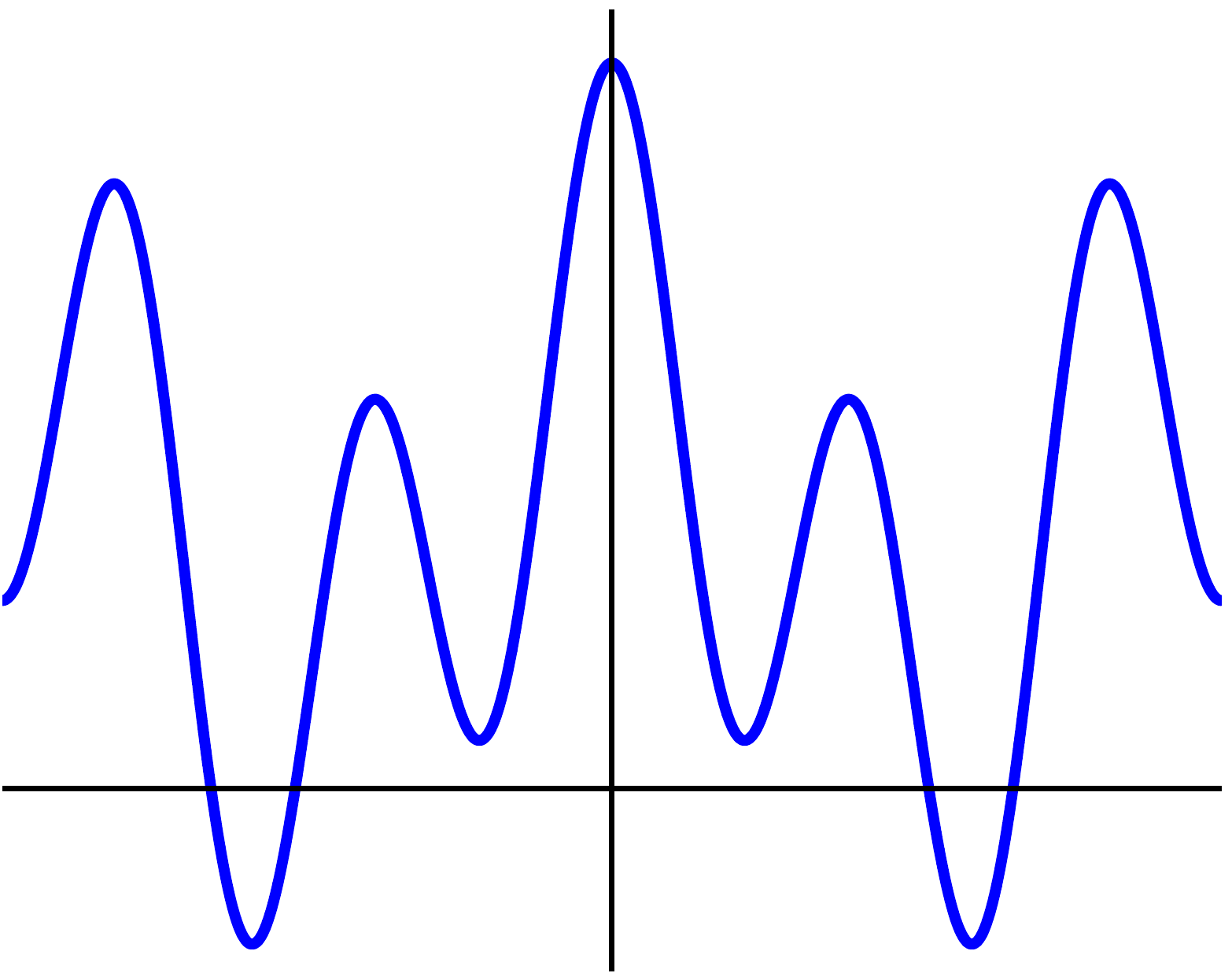}
}
\subfigure[]{
\includegraphics[width=5.5cm]{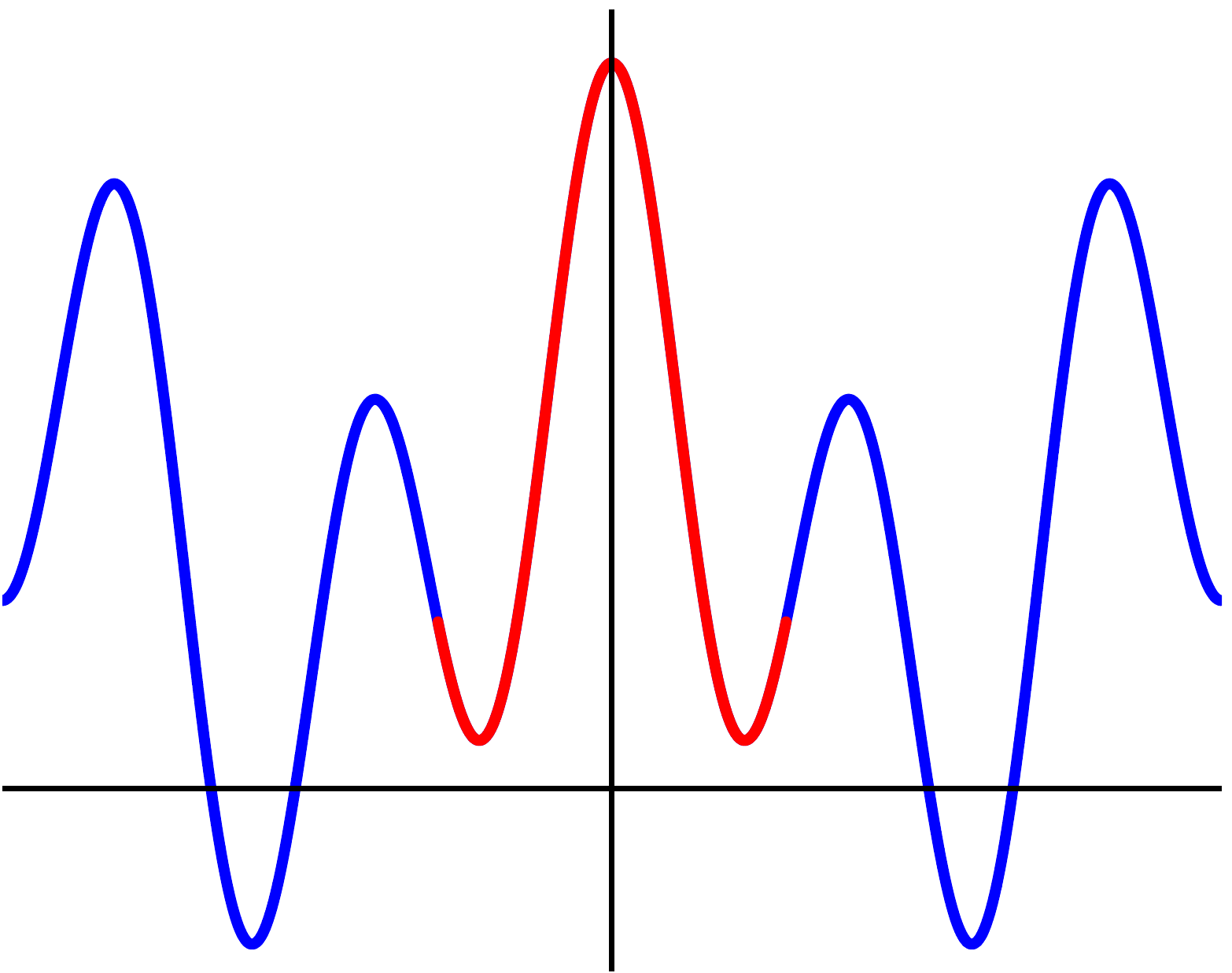}
}
\caption{Schematic illustration of the super-resolution problem in the
  frequency domain. (a) Real part of the Fourier transform in Figure
  \ref{fig:example_a}. (b) In red, portion of the observed
  spectrum. Super-resolution is about extrapolating the red fragment
  to recover the whole curve.}
\label{fig:example_spectrum}
\end{figure}
The daunting and ill-posed super-resolution problem then consists in
recovering the fine-scale or, equivalently, the high-frequency
features of $x$ even though they have been killed by the measurement
process. This is schematically represented in Figure
\ref{fig:example}, which shows a highly resolved signal together with
a low-resolution of the same signal obtained by convolution with a
point-spread function.  Super-resolution aims at recovering the fine
scale structure on the left from coarse scale features on the
right. Viewed in the frequency domain, super-resolution is of course
the problem of extrapolating the high-end and missing part of the
spectrum from the low-end part, as seen in Figure
\ref{fig:example_spectrum}. For reference, this is very different from
a typical compressed sensing problem \cite{candesFreq} in which we
wish to interpolate---and not extrapolate---the spectrum.

This paper develops a mathematical theory of super-resolution, which
is not developed at all despite the abundance of empirical work in
this area, see Section \ref{subsec:comparison} for some
references. Our theory takes multiple forms but a simple and appealing
incarnation is as follows. Suppose we wish to super-resolve a
spike-train signal as in Figure \ref{fig:example}, namely, a
superposition of pointwise events. Our main result is that we can
recover such a signal exactly from low-frequency samples by tractable
convex optimization in which one simply minimizes the continuous
analog to the $\ell_1$ norm for discrete signals subject to data
constraints. This holds as long as the spacing between the spikes is
on the order of the resolution limit. Furthermore, the theory shows
that this procedure is robust in the sense that it degrades smoothly
when the measurements are contaminated with noise. In fact, we shall
quantify quite precisely the error one can expect as a function of the
size of the input noise and of the resolution one wishes to achieve.\\


\subsection{Models and methods}

For concreteness, consider a continuous-time model in which the signal
of interest is a weighted superposition of spikes
\begin{equation}
  \label{eq:model}
  x = \sum_j a_j \delta_{t_j}, 
\end{equation}
where $\{t_j\}$ are locations in $[0,1]$ and $\delta_{\tau}$ is a
Dirac measure at $\tau$. The amplitudes $a_j$ may be complex
valued. Expressed differently, the signal $x$ is an atomic measure on
the unit interval putting complex mass at time points $t_1$, $t_2$,
$t_3$ and so on. The information we have available about $x$ is a
sample of the lower end of its spectrum in the form of the lowest
$2\fc + 1$ Fourier series coefficients ($\fc$ is an integer):
\begin{equation}
  \label{eq:fourier}
  y(k) = \int_0^1 e^{-i2\pi kt} x(\text{d}t)  = \sum_j a_j e^{-i2\pi k
    t_j}, \quad k \in \Z, \, \abs{k}\leq \fc. 
\end{equation}
For simplicity, we shall use matrix notations to relate the data $y$
and the object $x$ and will write \eqref{eq:fourier} as $y =
\mathcal{F}_{n} \, x$ where $\mathcal{F}_{n}$ is the linear map
collecting the lowest $n = 2\fc + 1$ frequency coefficients.  It is
important to bear in mind that we have chosen this model mainly for
ease of exposition. Our techniques can be adapted to settings where
the measurements are modeled differently, e.~g.~by sampling the
convolution of the signal with different low-pass kernels.  The
important element is that just as before, the frequency cut-off
induces a resolution limit inversely proportional to $f_c$; below we
set $\lambda_c = 1/\fc$ for convenience.

To recover $x$ from low-pass data we shall find, among all measures
fitting the observations, that with lowest total variation.  The total
variation of a complex measure (see Section \ref{sec:tv} in the
Appendix for a rigorous definition) can be interpreted as being the
continuous analog to the $\ell_1$ norm for discrete signals. In fact,
with $x$ as in \eqref{eq:model}, $\normTV{x}$ is equal to the $\ell_1$
norm of the amplitudes $\normOne{a} = \sum_j |a_j|$. Hence, we propose
solving the convex program
\begin{align}
\label{TVnormMin}
\min_{\tilde x}\normTV{\tilde x} \quad \text{subject to} \quad
\mathcal{F}_{n} \, \tilde x = y, 
\end{align}
where the minimization is carried out over the set of all finite
complex measures $\tilde x$ supported on $[0,1]$.  Our first result
shows that if the spikes or atoms are sufficiently separated, at least
$\optvalue \, \lambda_c$ apart, the solution to this convex program is
exact. 

\begin{definition}[Minimum separation] Let $\mathbb{T}$ be the circle
  obtained by identifying the endpoints on $[0,1]$ and $\mathbb{T}^d$
  the $d$-dimensional torus. For a family of points $T \subset
  \mathbb{T}^d$, the minimum separation is defined as the closest
  distance between any two elements from $T$,
  \begin{equation}
    \label{eq:min-distance-def}
    \Delta(T) = \inf_{(t, t') \in T \, : \, t \neq t'} \, \, |t - t'|, 
  \end{equation}
  where $|t - t'|$ is the $\ell_\infty$ distance (maximum deviation in
  any coordinate). To be clear, this is the wrap-around distance so
  that for $d = 1$, the distance between $t = 0$ and $t' = 3/4$ is
  equal to $1/4$.
\end{definition}
\begin{theorem}
  \label{theorem:noiseless} Let $T = \{t_j\}$ be the support of
  $x$. If the minimum distance obeys 
\begin{equation}
\label{eq:min-dist}
\Delta(T)  \geq {\optvalue}\, /{\fc} := \optvalue \, \lambda_c, 
\end{equation}
then $x$ is the unique solution to \eqref{TVnormMin}. This holds with
the proviso that $f_c \geq \minm$. If $x$ is known to be real-valued,
then the minimum gap can be lowered to $1.87 \, \lambda_c$.
\end{theorem}
We find this result particularly unexpected. The total-variation norm
makes no real assumption about the structure of the signal. Yet, not
knowing that there are any spikes, let alone how many there are,
total-variation minimization locates the position of those spikes with
{\em infinite precision}!  Even if we knew that \eqref{TVnormMin}
returned a spike train, there is no reason to expect that the
locations of the spikes would be infinitely accurate from coarse scale
information only. In fact, one would probably expect the fitted
locations to deviate at least a little from the truth. This is not
what happens.



The theorem does not depend on the amplitudes and applies to
situations where we have both very large and very small spikes. The
information we have about $x$ is equivalent to observing the
projection of $x$ onto its low-frequency components, i.e.~the
constraint in \eqref{TVnormMin} is the same as $\mathcal{P}_n \tilde x
= \mathcal{P}_n x$, where $\mathcal{P}_n = \mathcal{F}_n^*
\mathcal{F}_n$. As is well known, this projection is the convolution
with the Dirichlet kernel, which has slowly decaying side lobes.
Hence, the theorem asserts that total-variation minimization will pull
out the small spikes even though they may be completely buried in the
side lobes of the large ones as shown in Figure
\ref{fig:dirichlet}. To be sure, by looking at the curves in Figure
\ref{fig:dirichlet}, it seems a priori impossible to tell how many
spikes there are or roughly where they might be.
\begin{figure}
\centering
\includegraphics[width=12cm,height=4.5cm]{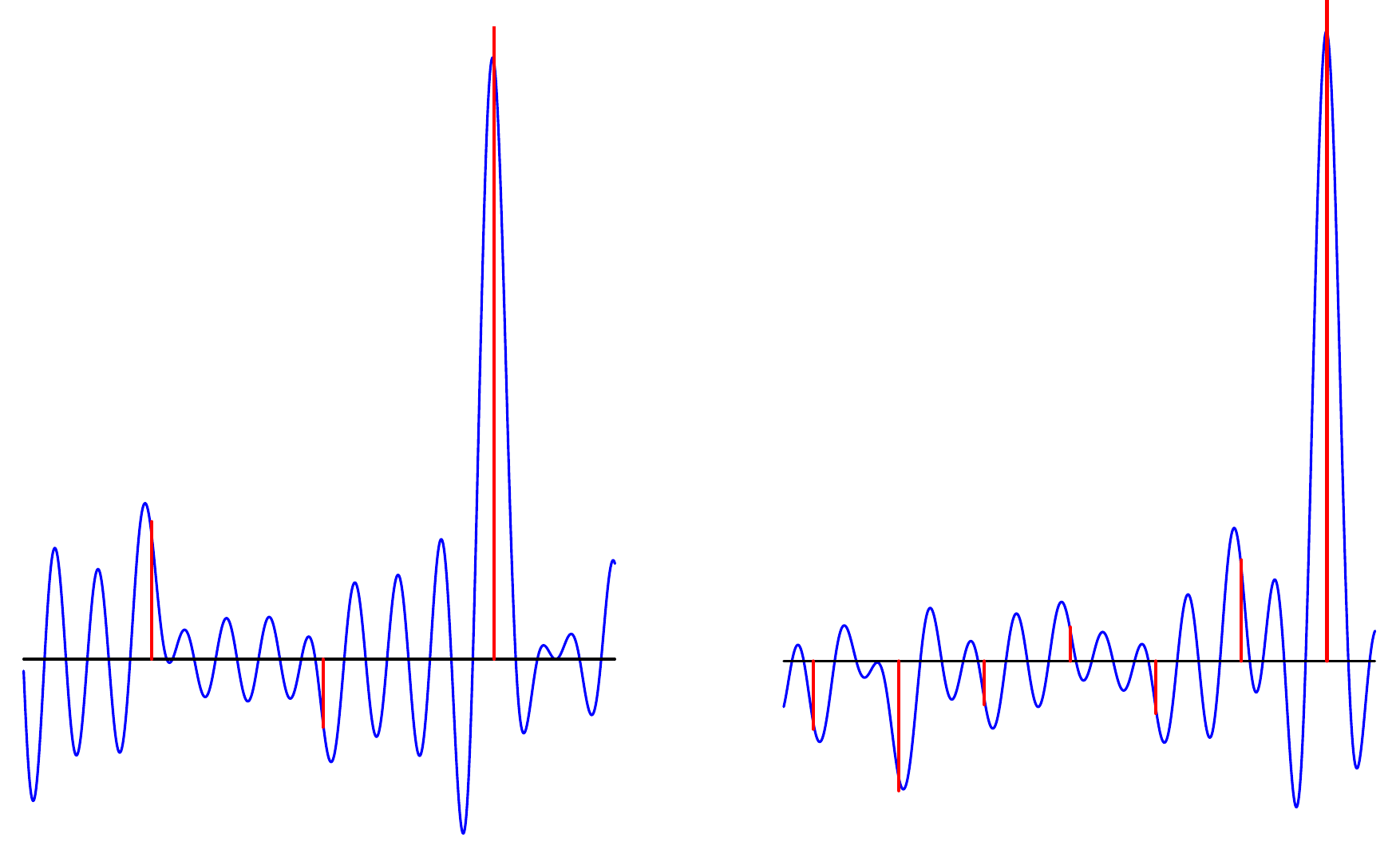}
\caption{Low-frequency projection of a signal corresponding to a
  signal with three (left) and seven (right) spikes. The blue curve
  represents the available data while the red ticks represent the
  unknown locations and relative amplitudes of the spikes. The
  locations of the spikes obey the condition of Theorem
  \ref{theorem:noiseless}.  Yet, by looking at the curves, it seems a
  priori impossible to tell how many spikes there are or roughly where
  they might be. }
\label{fig:dirichlet}
\end{figure}

An interesting aspect of this theorem is that it cannot really be
tested numerically. Indeed, one would need a numerical solver with
infinite precision to check that total-variation minimization truly
puts the spikes at exactly the right locations.  Although this is of
course not available, Section \ref{sec:sdp} shows how to solve the
minimum total-variation problem \eqref{TVnormMin} by using ideas from
semidefinite programming and allowing recovery of the support with
very high precision. Also, we demonstrate through numerical
simulations in Section \ref{sec:numerical} that Theorem
\ref{theorem:noiseless} is fairly tight in the sense that a necessary
condition is a separation of at least $\lambda_c = 1/f_c$.


Viewed differently, one can ask: how many spikes can be recovered from
$n = 2\fc + 1$ low-frequency samples? The answer given by Theorem
\ref{theorem:noiseless} is simple. At least $n/\optvaluetimestwo$
provided we have the minimum separation discussed above.  A classical
argument shows that any method whatsoever would at least need two
samples per unknown spike so that the number of spikes cannot exceed
half the number of samples, i.~e.~$n/2$. This is another way of
showing that the theorem is reasonably tight.

\subsection{Super-resolution in higher dimensions}

Our results extend to higher dimensions and reveal the same dependence
between the minimum separation and the measurement resolution as in
one dimension. For concreteness, we discuss the $2$-dimensional
setting and emphasize that the situation in $d$ dimensions is
similar. Here, we have a measure
\[
  x = \sum_j a_j \delta_{t_j}, 
\]
as before but in which the $t_j \in [0,1]^2$. We are given information
about $x$ in the form of low-frequency samples of the form 
\begin{equation}
  \label{eq:2DFT}
  y(k) = \int_{[0,1]^2} e^{-i2\pi \<k, t\>} x(\text{d}t) = \sum_j a_j
  e^{-i2\pi \<k, t_j\>}, \quad k = (k_1, k_2) \in \Z^2, \, \abs{k_1}, \abs{k_2}
  \leq \fc.
\end{equation}
This again introduces a physical resolution of about $\lambda_c =
1/\fc$. In this context, we may think of our problem as imaging point
sources in the 2D plane---such as idealized stars in the sky---with an
optical device with resolution about $\lambda_c$---such as a
diffraction limited telescope. Our next result states that it is
possible to locate the point sources without any error whatsoever if
they are separated by a distance of $\optvaluetwoD \, \lambda_c$
simply by minimizing the total variation.
\begin{theorem}
\label{theorem:2D}
Let $T = \{t_j\} \subset [0,1]^2$ identified with $\mathbb{T}^2$ be
the support of $x$ obeying the separation condition\footnote{Recall
  that distance is measured in $\ell_\infty$.}
\begin{equation}
  \Delta(T) \geq {\optvaluetwoD}\, /{\fc} = {\optvaluetwoD}\, \lambda_c.  \label{min_dist_condition_2D}
\end{equation}
Then if $x$ is real valued, it is the unique minimum total-variation
solution among all real objects obeying the data constraints
\eqref{eq:2DFT}. Hence, the recovery is exact. For complex measures,
the same statement holds but with a slightly different constant.
\end{theorem}
Whereas we have tried to optimize the constant in one
dimension\footnote{This is despite the fact that the authors have a
  proof---not presented here---of a version of Theorem
  \ref{theorem:noiseless} with a minimum separation at least equal to
  $1.85 \, \lambda_c$.}, we have not really attempted to do so here in order to keep the proof reasonably short and simple. Hence, this theorem is
subject to improvement.

Theorem \ref{theorem:2D} is proved for real-valued measures in Section
\ref{sec:proof_2D} of the Appendix.  However, the proof techniques can
be applied to higher dimensions and complex measures almost
directly. In details, suppose we observe the discrete Fourier
coefficients of a $d$-dimensional object at $k = (k_1, \ldots, k_d)
\in \Z^d$ corresponding to low frequencies $0 \le |k_1|, \ldots, |k_d|
\le \fc$. Then the minimum total-variation solution is exact provided
that the minimum distance obeys $\Delta(T) \ge c_d \, \lambda_c$,
where $c_d$ is some positive numerical constant depending only on the
dimension. Finally, as the proof makes clear, extensions to other
settings, in which one observes Fourier coefficients if and only if
the $\ell_2$ norm of $k$ is less or equal to a frequency cut-off, are
straightforward.

\subsection{Discrete super-resolution}
\label{sec:discrete-SR}

Our continuous theory immediately implies analogous results for finite
signals. Suppose we wish to recover a discrete signal $x \in
\mathbb{C}^N$ from low-frequency data. Just as before, we could
imagine collecting low-frequency samples of the form
\begin{equation}
\label{eq:discreteFT}
y_k = \sum_{t = 0}^{N-1} x_t e^{-i2\pi k t/N}, \quad |k| \le \fc;
\end{equation}
the connection with the previous sections is obvious since $x$ might
be interpreted as samples of a discrete signal on a grid $\{t/N\}$
with $t = 0,1, \ldots, N-1$. In fact, the continuous-time setting is
the limit of infinite resolution in which $N$ tends to infinity while
the number of samples remains constant ($\fc$ fixed). Instead, we can
choose to study the regime in which the ratio between the actual
resolution of the signal $1/N$ and the resolution of the data defined
as $1/\fc$ is constant. This gives the corollary below.
\begin{corollary}
\label{cor:discrete}
Let $T \subset \{0,1, \ldots, N-1\}$ be the support of $\{x_t\}_{t =
  0}^{N-1}$ obeying
\begin{equation}
  \min_{t, t' \in T : t \neq t'} \frac{1}{N} \abs{t-t'} \geq {\optvalue}\, \lambda_c = {\optvalue}\, /{\fc}. 
\end{equation} 
Then the solution to
\begin{equation}
\label{l1problem}
\min \normOne{\tilde x} \quad \text{subject to} \quad F_n \tilde x = y 
\end{equation}
in which $F_n$ is the partial Fourier matrix in \eqref{eq:discreteFT}
is exact.
\end{corollary}

\subsection{The super-resolution factor}
\label{sec:SRF}

In the discrete framework, we wish to resolve a signal on a fine grid
with spacing $1/N$. However, we only observe the lowest $n = 2 f_c +
1$ Fourier coefficients so that in principle, one can only hope to
recover the signal on a coarser grid with spacing only $1/n$ as shown
in Figure \ref{fig:Nyquistgrids}. Hence, the factor $N/n$, or
equivalently, the ratio between the spacings in the coarse and fine
grids, can be interpreted as a super-resolution factor (SRF). Below,
we set
\begin{equation}
  \label{eq:SRF}
  \srf = N/n \approx N/2f_c;
\end{equation}
when the SRF is equal to 5 as in the figure, we are looking for a
signal at a resolution 5 times higher than what is stricto senso
permissible. One can then recast Corollary \ref{cor:discrete} as
follows: if the nonzero components of $\{x_t\}_{t = 0}^{N-1}$ are
separated by at least $\optvaluetimestwo \times \srf$, perfect
super-resolution via $\ell_1$ minimization occurs.
\begin{figure}[ht]
\centering
\includegraphics[width=7cm]{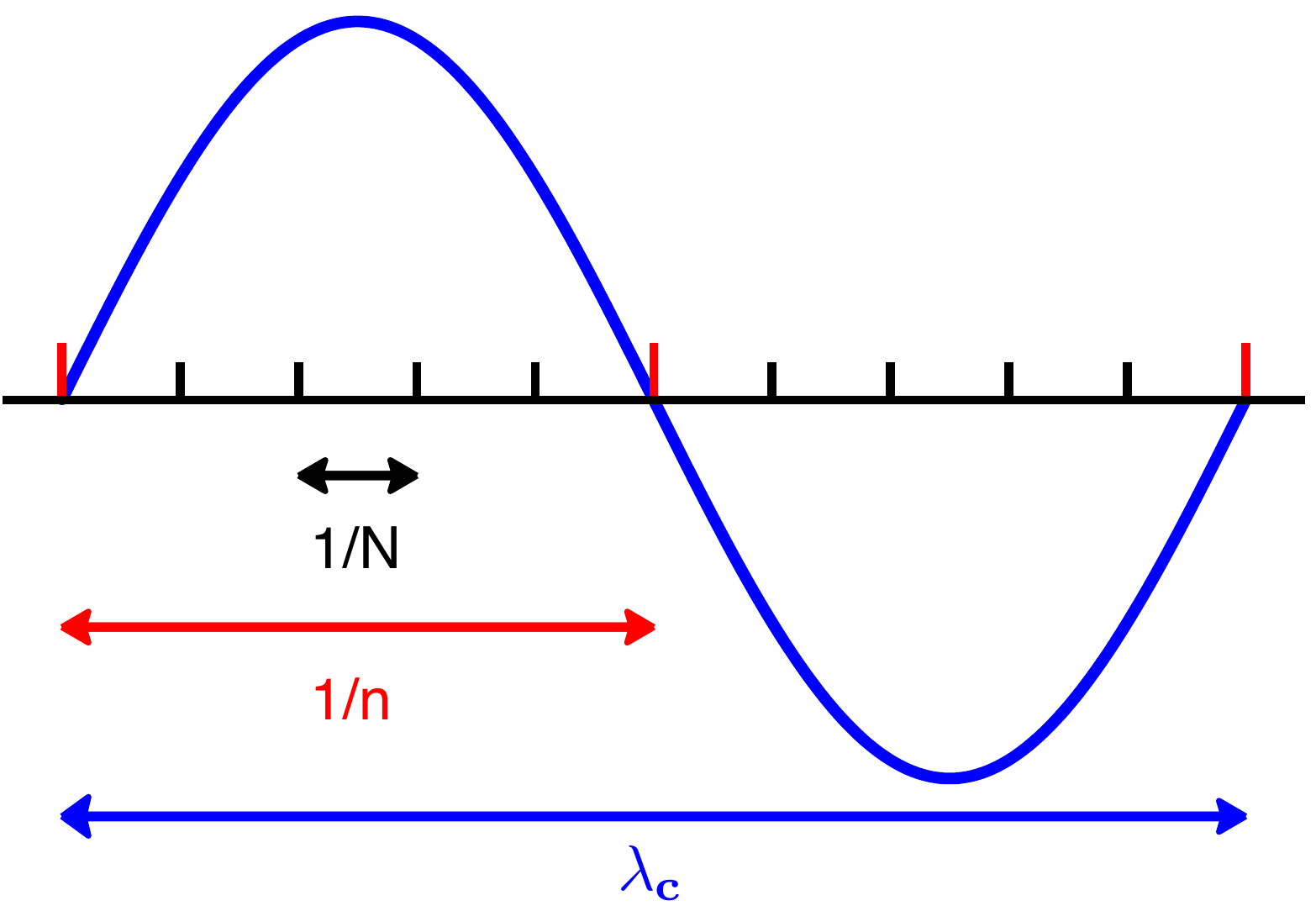}
\caption{Fine grid with spacing $1/N$.  We only observe frequencies
  between $-f_c$ and $f_c$, $f_c/N \approx \frac12 \srf^{-1}$, so that
  the highest frequency sine wave available has wavelength $1/f_c =
  \lambda_c$. These data only allow a Nyquist sampling rate of
  $\lambda_c/2 \approx 1/n$. In this sense, we can interpret the
  super-resolution factor $N/n$ as the ratio between these two
  resolutions.}
\label{fig:Nyquistgrids}
\end{figure}

The reason for introducing the SRF is that with inexact data, we
obviously cannot hope for infinite resolution. Indeed, noise will
ultimately limit the resolution one can ever hope to achieve and,
therefore, the question of interest is to study the accuracy one might
expect from a practical super-resolution procedure as a function of
both the noise level and the SRF.

\subsection{Stability}
\label{subsec:stability}

To discuss the robustness of super-resolution methods vis a vis noise,
we examine in this paper the discrete setting of Section
\ref{sec:discrete-SR}. In this setup, we could certainly imagine
studying a variety of deterministic and stochastic noise models, and a
variety of metrics in which to measure the size of the error. For
simplicity, we study a deterministic scenario in which the projection
of the noise onto the signal space has bounded $\ell_1$ norm but is
otherwise arbitrary and can be adversarial. The observations are
consequently of the form
\begin{equation}
\label{noisemodel_l1}
y = F_n x + w, \quad \frac{1}{N} \, \normOne{F_n^{\ast}w} \leq  \delta
\end{equation}
for some $\delta \ge 0$, where $F_n$ is as before. Letting $P_n$ be
the orthogonal projection of a signal onto the first $n$ Fourier
modes, $P_n = \frac{1}{N} F_n^* F_n$, we can view
\eqref{noisemodel_l1} as an input noise model since with $w = F_n z$,
we have
\[
y = F_n (x + z), \quad \normOne{z} \le \delta, \, z = P_n z. 
\]
Another way to write this model with arbitrary input noise $z \in
\C^N$ is 
\[
y = F_n(x + z), \quad \normOne{P_n z} \le \delta
\]
since the high-frequency part of $z$ is filtered out by the
measurement process. Finally, with $s = {N}^{-1} F_n^* y$,
\eqref{noisemodel_l1} is equivalent to
\begin{equation}
  \label{eq:model-time}
  s = P_n x + P_n z, \quad \normOne{P_n z} \le \delta.
\end{equation}
In words, we observe a low-pass version of the signal corrupted with
an additive low-pass error whose $\ell_1$ norm is at most $\delta$.
In the case where $n = N$, $P_n = I$, and our model becomes
\[
s = x + z, \quad \|z\|_1 \le \delta.
\]
In this case, one cannot hope for a reconstruction $\hat x$ with an
error in the $\ell_1$ norm less than the noise level $\delta$. We now
wish to understand how quickly the recovery error deteriorates as the
super-resolution factor increases.

We propose studying the relaxed version of the noiseless problem
\eqref{l1problem}
\begin{equation}
\label{l1problem_relaxed}
\min_{\tilde x} \,  \normOne{\tilde x} 
\quad \text{subject to} \quad \normOne{P_n \tilde x - s} \leq \delta. 
\end{equation}
We show that this recovers $x$ with a precision inversely proportional to
$\delta$ and to the square of the
super-resolution factor.
\begin{theorem}
\label{theorem:discrete_noisy}
Assume that $x$ obeys the separation condition
\eqref{eq:min-dist}. Then with the noise model
\eqref{eq:model-time}, the solution $\hat{x}$ to
\eqref{l1problem_relaxed} obeys
\begin{equation}
\label{eq:discrete-noisy}
\normOne{\hat x- x} \leq C_0 \, \text{\em SRF}^2 \, \delta,  
\end{equation}
for some positive constant $C_0$. 
\end{theorem}
This theorem, which shows the simple dependence upon the
super-resolution factor and the noise level, is proved in Section
\ref{sec:stability}. Clearly, plugging in $\delta = 0$ in
\eqref{eq:discrete-noisy} gives Corollary \ref{cor:discrete}.

Versions of Theorem \ref{theorem:discrete_noisy} hold in the
continuous setting as well, where the locations of the spikes are not
assumed to lie on a given fine grid but can take on a continuum of
values. The arguments are more involved than those needed to establish
\eqref{eq:discrete-noisy} and we leave a detailed study to a future
paper.

\subsection{Sparsity and stability} 
\label{sec:slepian_intro}

Researchers in the field know that super-resolution under sparsity
constraints alone is hopelessly ill posed. In fact, without a minimum
distance condition, the support of sparse signals can be very
clustered, and clustered signals can be nearly completely annihilated
by the low-pass sensing mechanism. The extreme ill-posedness can be
understood by means of the seminal work of Slepian
\cite{slepian_discrete} on discrete prolate spheroidal sequences. This
is surveyed in Section \ref{sec:slepian} but we give here a concrete
example to drive this point home.

To keep things simple, we consider the `analog version' of
\eqref{eq:model-time} in which we observe 
\[
s = \mathcal{P}_W (x + z);
\]
$\mathcal{P}_W(x)$ is computed by taking the discrete-time Fourier
transform $y(\omega) = \sum_{t \in \Z} x_t e^{-i2\pi \omega t}$, $\omega \in
[-1/2,1/2],$ and discarding all `analog' frequencies outside of the
band $[-W,W]$.  If we set
\[
2W = n/N = 1/\srf,
\]
we essentially have $\mathcal{P}_W = P_n$ where the equality is true
in the limit where $N \goto \infty$ (technically, $\mathcal{P}_W$ is
the convolution with the sinc kernel while $P_n$ uses the Dirichlet
kernel). Set a mild level of super-resolution to fix ideas, 
\[
\srf = 4. 
\]
Now the work of Slepian shows that there is a $k$-sparse signal
supported on $[0, \ldots, k-1]$ obeying
\begin{equation}
  \label{eq:eigenvalue_slepian}
\mathcal{P}_W x = \lambda \, x, \quad \lambda \approx 5.22 \, \sqrt{k+1}
  \, e^{-3.23 \brac{k+1}}.
\end{equation}
For $k = 48$,
\begin{equation}
\label{eq:ten_minus_60}
\lambda \le 7 \times 10^{-68}. 
\end{equation}
Even knowing the support ahead of time, how are we going to recover
such signals from noisy measurements? Even for a very mild
super-resolution factor of just $\srf = 1.05$ (we only seek to extend
the spectrum by 5\%), \eqref{eq:eigenvalue_slepian} becomes
\begin{equation}
  \label{eq:eigenvalue_slepian_1p05}
  \mathcal{P}_W x = \lambda \, x, \quad \lambda \approx 3.87 \, \sqrt{k+1}
  \, e^{-0.15 \brac{k+1}},
\end{equation}
which implies that there exists a unit-norm signal with at most $256$
consecutive nonzero entries such that $\normTwo{\mathcal{P}_W x} \leq
1.2 \times 10^{-15}$. Of course, as the super-resolution factor
increases, the ill-posedness gets worse. For large values of $\srf$,
there is $x$ obeying \eqref{eq:eigenvalue_slepian} with
\begin{equation}
\label{eq:slepian_largeSRF}
\log \lambda \approx - (0.4831 + 2\log(\srf))k.  
\end{equation}
It is important to emphasize that this is not a worst case
analysis. In fact, with $k = 48$ and $\srf = 4$, Slepian shows that
there is a large dimensional subspace of signals supported on $\C^k$
spanned by orthonormal eigenvectors with eigenvalues of magnitudes
nearly as small as \eqref{eq:ten_minus_60}. 

\subsection{Comparison with related work}
\label{subsec:comparison}

The use of $\ell_1$ minimization for the recovery of sparse spike
trains from noisy bandlimited measurements has a long history and was
proposed in the 1980s by researchers in seismic prospecting
\cite{claerbout,levy,santosa}. For finite signals and under the rather
restrictive assumption that the signal is real valued and nonnegative,
\cite{fuchs_positive} and \cite{donoho_positive} prove that $k$ spikes
can be recovered from $2k+1$ Fourier coefficients by this
method. 
The work \cite{supportPursuit} extends this result to the continuous
setting by using total-variation minimization. In contrast, our
results require a minimum distance between spikes but allow for
arbitrary complex amplitudes, which is crucial in applications. The
only theoretical guarantee we are aware of concerning the recovery of
spike trains with general amplitudes is very recent and due to Kahane
\cite{kahane_superresolution}. Kahane offers variations on compressive
sensing results in \cite{candesFreq} and studies the reconstruction of
a function with lacunary Fourier series coefficients from its values
in a small contiguous interval, a setting that is equivalent to that
of Corollary \ref{cor:discrete} when the size $N$ of the fine grid
tends to infinity. With our notation, whereas we require a minimum
distance equal to $\optvaluetimestwo \times \srf$, this work shows
that a minimum distance of $10 \times \srf \sqrt{\log \srf}$ is
sufficient for exact recovery.  Although the log factor might seem
unimportant at first glance, it in fact precludes extending Kahane's
result to the continuous setting of Theorem
\ref{theorem:noiseless}. Indeed, by letting the resolution factor tend
to infinity so as to approach the continuous setting, the spacing between
consecutive spikes would need to tend to infinity as well.

As to results regarding the robustness of super-resolution in the
presence of noise, Donoho \cite{donohoSuperres} studies the modulus of
continuity of the recovery of a signed measure on a discrete lattice
from its spectrum on the interval $\sqbr{-f_c,f_c}$, a setting which
is also equivalent to that of Corollary \ref{cor:discrete} when the
size $N$ of the fine grid tends to infinity.  More precisely, if the
support of the measure is constrained to contain at most $\ell$
elements in any interval of length $ 2/(\ell \, f_c)$, then the
modulus of continuity is of order $O\brac{\srf^{2\ell+1}}$ as $\srf$
grows to infinity (note that for $\ell = 1$ the constraint reduces to
a minimum distance condition between spikes, which is comparable to
the separation condition \eqref{eq:min-dist}). This means that if the
$\ell_2$ norm of the difference between the measurements generated by
two signals satisfying the support constraints is known to be at most
$\delta$, then the $\ell_2$ norm of the difference between the signals
may be of order $O\brac{\srf^{2\ell+1}\, \delta}$. This result
suggests that, in principle, the super-resolution of spread-out
signals is not hopelessly ill-conditioned. Having said this, it does
not propose any practical recovery algorithm (a brute-force search for
sparse measures obeying the low-frequency constraints would be 
computationally intractable).

Finally, we would like to mention an alternative approach to the
super-resolution of pointwise events from coarse scale
data. Leveraging ideas related to error correction codes and spectral
estimation, \cite{fri} shows that it is possible to recover trains of
Dirac distributions from low-pass measurements at their \textit{rate
  of innovation} (in essence, the density of spikes per unit of
time). This problem, however, is extraordinarily ill posed without a
minimum separation assumption as explained in Sections
\ref{sec:slepian_intro} and \ref{sec:slepian}.  Moreover, the proposed
reconstruction algorithm in \cite{fri} needs to know the number of
events ahead of time, and relies on polynomial root finding. As a
result, it is highly unstable in the presence of noise as discussed in
\cite{tan_stochastic_fri}, and in the presence of approximate
sparsity. Algebraic techniques have also been applied to the location
of singularities in the reconstruction of piecewise polynomial
functions from a finite number of Fourier coefficients (see
\cite{banerjee_algebraic,batenkov_algebraic,eckhoff_algebraic} and
references therein). The theoretical analysis of these methods proves
their accuracy up to a certain limit related to the number of
measurements. Corollary \ref{cor:piecewise} takes a different
approach, guaranteeing perfect localization if there is a minimum
separation between the singularities.

\subsection{Connections to sparse recovery literature}
\label{subsec:sparse_recovery}

Theorem \ref{theorem:noiseless} and Corollary \ref{cor:discrete} can
be interpreted in the framework of sparse signal recovery. For
instance, by swapping time and frequency, Corollary \ref{cor:discrete}
asserts that one can recover a sparse superposition of tones with
arbitrary frequencies from $n$ time samples of the form 
\[
y_t = \sum_{j = 0}^{N-1} x_j e^{-i2\pi t \omega_j}, \quad t = 0, 1, \ldots, n-1
\]
where the frequencies are of the form $\omega_j = j/N$. Since the
spacing between consecutive frequencies is not $1/n$ but $1/N$, we may
have a massively oversampled discrete Fourier transform, where the
oversampling ratio is equal to the super-resolution factor.
In this context, a sufficient condition for perfectly super-resolving
these tones is a minimum separation of $\optvaluetimestwo/n$. In
addition, Theorem \ref{theorem:noiseless} extends this to continuum
dictionaries where tones $\omega_j$ can take on arbitrary real values.


In the literature, there are several conditions that guarantee perfect
signal recovery by $\ell_1$ minimization.  The results obtained from
their application to our problem are, however, very weak. 
\begin{itemize}
\item The matrix with normalized columns $f_j = \{e^{-i2\pi t
    \omega_j}/\sqrt{n}\}_{t = 0}^{n-1}$ does not obey the restricted
  isometry property \cite{candes_decoding} since a submatrix composed
  of a very small number of contiguous columns is already very close
  to singular, see \cite{slepian_discrete} and Section \ref{sec:slepian} for related claims.  For
  example, with $N=512$ and a modest $\srf$ equal to 4, the smallest
  singular value of submatrices formed by eight consecutive columns is
  $3.32 \; 10^{-5}$.
\item Applying the discrete uncertainty principle proved in
  \cite{uncertainty_principles}, we obtain that recovery by $\ell_1$
  minimization succeeds as long as
\[
2\abs{T} (N - n) < N. 
\]
If $n < N/2$, i.e.~$\srf > 2$, this says that $|T|$ must be zero. In
other words, to recover one spike, we would need at least half of the
Fourier samples.
\item Other guarantees based on the coherence of the dictionary yield
  similar results. A popular condition \cite{just_relax} requires that
\begin{equation}
\label{eq:inc}
|T| < \frac12 (M^{-1} + 1), 
\end{equation}
where $M$ is the coherence of the system defined as $\max_{i \neq j}
|\<f_i, f_j\>|$. When $N = 1024$ and $\srf = 4$, $M \approx 0.9003$ so
that this becomes $|T| \le 1.055$, and we can only hope to recover
one spike.
\end{itemize}

There are slightly improved versions of \eqref{eq:inc}. In
\cite{dossal_thesis}, Dossal studies the deconvolution of spikes by
$\ell_1$ minimization. This work introduces the weak exact recovery
condition (WERC) defined as 
\[
\text{WERC}\brac{T}= \frac{\beta\brac{T}}{1-\alpha\brac{T}}, 
\]
where 
\[
\alpha(T) = \sup_{i \in T} \sum_{j\in T/\keys{i}} \abs{\PROD{f_i}{f_j}},
\quad \beta(T) = \sup_{i \in T^c} \sum_{j\in T}
\abs{\PROD{f_i}{f_j}}.
\]
The condition $\text{WERC}\brac{T}<1$ guarantees exact recovery.
Considering three spikes and using Taylor expansions to bound the sine
function, the minimum distance needed to ensure that
$\text{WERC}\brac{T}<1$ may be lower bounded by
$24\srf^3/\pi^3-2\srf$. This is achieved by considering three spikes
at $\omega \in \{0, \pm \Delta\}$, where $\Delta = (k+1/2)/n$ for some
integer $k$; we omit the details.  If $N=20,000$ and the number of
measurements is $1,000$, this allows for the recovery of at most $3$
spikes, whereas Corollary \ref{cor:discrete} implies that it is
possible to reconstruct at least $n/4 = 250$. Furthermore, the cubic
dependence on the super-resolution factor means that if we fix the
number of measurements and let $N \rightarrow \infty$, which is
equivalent to the continuous setting of Theorem
\ref{theorem:noiseless}, the separation needed becomes infinite and we
cannot guarantee the recovery of even two spikes.

Finally, we would also like to mention some very recent work on sparse
recovery in highly coherent frames by modified greedy compressed
sensing algorithms \cite{spectral_cs,grids_fannjiang}. Interestingly,
these approaches explicitly enforce conditions on the recovered
signals that are similar in spirit to our minimum distance
condition. As opposed to $\ell_1$-norm minimization, such greedy
techniques may be severely affected by large dynamic ranges (see
\cite{grids_fannjiang}) because of the phenomenon illustrated in
Figure \ref{fig:dirichlet}. Understanding under what conditions their
performance may be comparable to that of convex programming methods is
an interesting research direction.

\subsection{Extensions}

Our results and techniques can be extended to super-resolve many other
types of signals. We just outline such a possible extension. 
Suppose $x : [0,1] \rightarrow \C$ is a periodic piecewise smooth
function with period 1, defined by
\[
x(t) =\sum_{t_j\in T} \mathbf{1}_{\brac{t_{j-1},t_j}} p_j(t); 
\]
on each time interval $(t_{j-1}, t_j)$, $x$ is polynomial of degree
$\ell$. For $\ell = 0$, we have a piecewise constant function, for
$\ell = 1$, a piecewise linear function and so on.  Also suppose $x$
is globally $\ell-1$ times continuously differentiable (as for
splines). We observe
\[
y_k = \int_{[0,1]} x(t) \, e^{-i2\pi kt} \text{d}t \; , \; \abs{k}\leq
\fc.
\] 
The $(\ell+1)$th derivative of $x$ (in the sense of distributions)
denoted by $x^{(\ell+1)}$ is an atomic measure supported on $T$ and
equal to
\[
x^{(\ell+1)} = \sum_j a_j \delta_{t_j}, \quad a_j =
p^{(\ell)}_{j+1}(t_j) - p^{(\ell)}_{j}(t_j). 
\]
Hence, we can imagine recovering $x^{(\ell+1)}$ by solving 
\begin{equation}
  \label{eq:minTVext}
  \min \|\tilde x^{(\ell+1)}\|_{\text{TV}} \quad \text{subject to} \quad F_n \tilde x = y.   
\end{equation}
Standard Fourier analysis gives that the $k$th Fourier coefficient of
this measure is given by 
\begin{equation}
\label{eq:newdata}
y^{(\ell+1)}_k = (i 2 \pi k)^{\ell+1} \, y_k, \quad k \neq 0.
\end{equation}
Hence, we observe the Fourier coefficients of $x^{(\ell+1)}$ except
that corresponding to $k = 0$, which must vanish since the periodicity
implies $\int_{0}^1 x^{(\ell+1)}(\text{d}t) = 0 = \int_{0}^1
x^{(j)}(t) \text{d}t$, $1 \le j \le \ell$. Hence, it follows from
Theorem \ref{theorem:noiseless} that \eqref{eq:minTVext} recovers
$x^{(\ell+1)}$ exactly as long as the discontinuity points are at
least $2\lambda_c$ apart.
Because $x$ is $\ell-1$ times continuously differentiable and
periodic, $x^{(\ell+1)}$ determines $x$ up to a shift in function
value, equal to its mean. However, we can read the mean value of $x$
off $y_0 = \int_0^1 x(t) \text{d}t$ and, therefore,
\eqref{eq:minTVext} achieves perfect recovery.
\begin{corollary}
\label{cor:piecewise}
If $T = \{t_j\}$ obeys \eqref{eq:min-dist}, $x$ is determined exactly
from $y$ by solving \eqref{eq:minTVext}.
\end{corollary}
Extensions to non-periodic functions, other types of discontinuities
and smoothness assumptions are straightforward.

\subsection{Organization of the paper}

The remainder of the paper is organized as follows. We prove our main
noiseless result in Section \ref{sec:dual_pol}. There, we introduce
our techniques which involve the construction of an interpolating
low-frequency polynomial. Section \ref{sec:stability} proves our
stability result and argues that sparsity constraints cannot be
sufficient to guarantee stable super-resolution. Section \ref{sec:sdp}
shows that \eqref{TVnormMin} can be cast as a finite semidefinite
program.
Numerical simulations providing a lower bound for the minimum distance
that guarantees exact recovery are presented in Section
\ref{sec:numerical}. We conclude the paper with a short discussion in
Section \ref{sec:discussion}.

\section{Noiseless Recovery}
\label{sec:dual_pol}

\newcommand{\Dmin}{\Delta_{\min}}

This section proves the noiseless recovery result, namely, Theorem
\ref{theorem:noiseless}. Here and below, we write $\Delta = \Delta(T)
\ge \Deltamin = \optvalue \, \lambda_c$. Also, we identify the
interval $[0,1)$ with the circle $\mathbb{T}$.

\subsection{Dual polynomials}

In the discrete setting, the compressed sensing literature has made
clear that the existence of a certain \textit{dual certificate}
guarantees that the $\ell_1$ solution is exact \cite{candesFreq}. In
the continuous setting, a sufficient condition for the success of the
total-variation solution is this: for any $v \in \C^{|T|}$ with $|v_j|
= 1$, there exists a low-frequency trigonometric polynomial
\begin{equation}
\label{eq:trig}
q(t) = \sum_{k = -\fc}^{\fc} c_k e^{i2\pi k t} 
\end{equation}
obeying the following properties 
\begin{equation}
  \label{eq:cond_q}
  \begin{cases} q(t_j) = v_j, & t_j \in T,\\
    |q(t)| < 1, & t \in \mathbb{T} \setminus T. 
\end{cases}
\end{equation}
This result follows from elementary measure theory and is included in
Section \ref{sec:tv} of the Appendix for completeness. Constructing a
bounded low-frequency polynomial interpolating the sign pattern of
certain signals becomes increasingly difficult if the minimum distance
separating the spikes is too small. This is illustrated in Figure
\ref{fig:dual_pol}, where we show that if spikes are very near, it
would become in general impossible to find an interpolating
low-frequency polynomial obeying \eqref{eq:cond_q}.
\begin{figure}
\centering
\subfigure[]{
\includegraphics[scale=0.4]{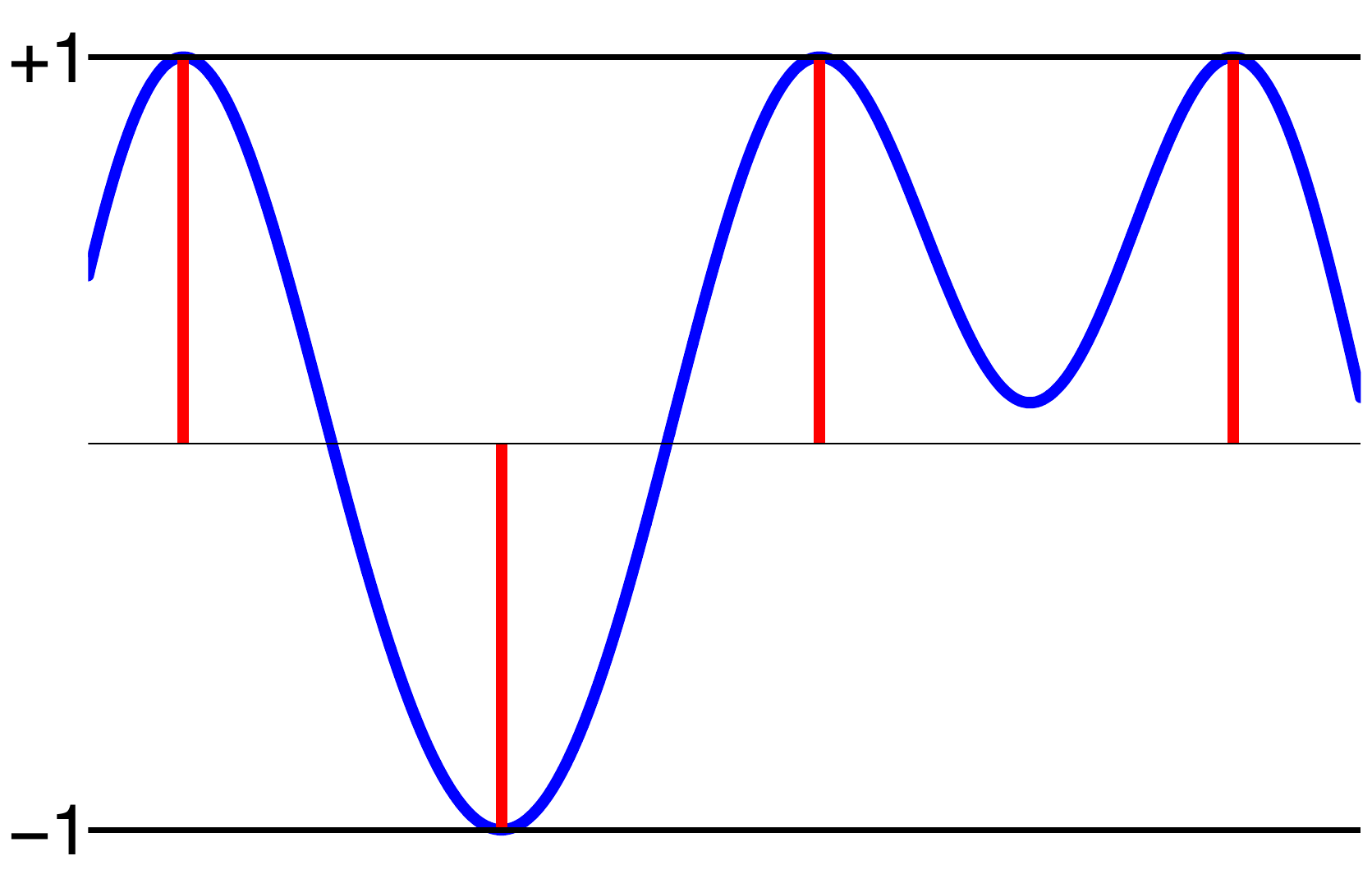}
\label{fig:dual_pol_a}
}
\subfigure[]{
\includegraphics[scale=0.4]{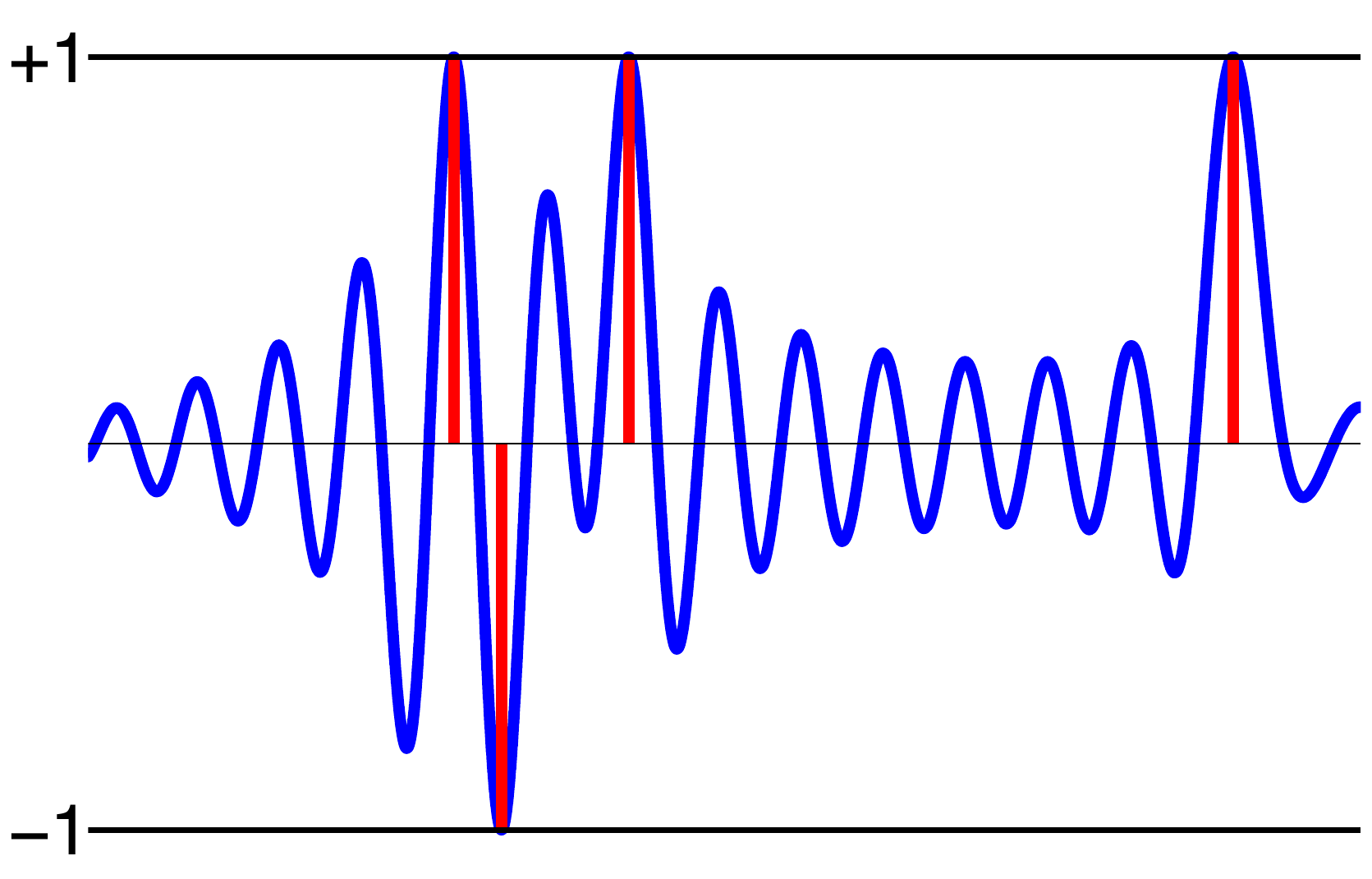}
\label{fig:dual_pol_b}
}
\caption{(a) Low-frequency polynomial interpolating a sign pattern in
  which the support is well separated, and obeying the off-support
  condition \eqref{eq:cond_q}. In (b), we see that if the spikes
  become very near, we would need a rapidly (high-frequency)
  interpolating polynomial in order to achieve \eqref{eq:cond_q}. This
  is the reason why there must be a minimum separation between
  consecutive spikes. }
\label{fig:dual_pol}
\end{figure}

\subsection{Proof of Theorem \ref{theorem:noiseless}}

Theorem \ref{theorem:noiseless} is a direct consequence of the
proposition below, which establishes the existence of a valid dual
polynomial provided the elements in the support are sufficiently
spaced.
\begin{proposition}
\label{prop:continuous_dualcert}
Let $v \in \C^{|T|}$ be an arbitrary vector obeying $|v_j| = 1$.  Then
under the hypotheses of Theorem \ref{theorem:noiseless}, there exists
a low-frequency trigonometric polynomial \eqref{eq:trig} obeying
\eqref{eq:cond_q}.
\end{proposition}

The remainder of this section proves this proposition. Our method
consists in interpolating $v$ on $T$ with a low-frequency kernel and
correcting the interpolation to ensure that the derivative of the dual
polynomial is zero on $T$. The kernel we employ is 
\begin{equation} 
  K(t) =
   \left[\frac{\sin \brac{\brac{\frac{\fc}{2}+1} \pi t}}{\brac{\frac{\fc}{2}+1}\sin \brac{\pi t}}\right]^4, \quad 0 < t < 1,  
\label{def:kernel}
\end{equation}
and $K(0) = 1$. If $f_c$ is even, $K(t)$ is the square of the Fej\'er
kernel which is a trigonometric polynomial with frequencies obeying
$|k| \le f_c/2$. As a consequence, $K$ is of the form
\eqref{eq:trig}. The careful reader might remark that the choice of
the interpolation kernel seems somewhat arbitrary. In fact, one could
also use the Fej\'er kernel or any other power of the Fej\'er kernel
using almost identical proof techniques. We have found that the second
power nicely balances the trade-off between localization in time and
in frequency, and thus yields a good constant.

To construct the dual polynomial, we interpolate $v$ with both $K$ and
its derivative $K'$, 
\begin{equation}
  q(t)  = \sum_{t_j \in T} \alpha_j K(t-t_j) + \beta_j K'(t-t_j), 
  \label{def:dual_cert}
\end{equation}
where $\alpha, \beta \in \C^{\abs{T}}$ are coefficient sequences. The
polynomial $q$ is as in \eqref{eq:trig} and in order to obey $q(t_k) =
v_k$, we impose
\begin{equation}
  \label{eq:interp1}
q(t_k) = \sum_{t_j \in T} \alpha_j K\brac{t_k-t_j} + \beta_j
K'\brac{t_k-t_j} = v_k, \quad \forall t_k \in T,
\end{equation}
whereas in order to obey $|q(t)|<1$ for $t \in T^c$, we impose
$q'(t_k) = 0$,
\begin{equation}
  \label{eq:interp2}
q'(t_k) = \sum_{t_j \in T} \alpha_j K'\brac{t_k-t_j} + \beta_j
K^{\prime\prime}\brac{t_k-t_j} = 0, \quad \forall t_k \in T.
\end{equation}
As we will see, this implies that the magnitude of $q$ reaches a local
maximum at those points, which in turn can be used to show that
\eqref{eq:cond_q} holds.

The proof of Proposition \ref{prop:continuous_dualcert} consists of
three lemmas, which are the object of the following section. The
first one establishes that if the support is spread out, it is
possible to interpolate any sign pattern exactly.
\begin{lemma}
\label{lemma:fejersq_coeffs}
Under the hypotheses of Proposition \ref{prop:continuous_dualcert},
there exist coefficient vectors $\alpha$ and $\beta$ obeying 
\begin{equation}
 \label{bound_alpha}
\begin{aligned}
  \normInf{\alpha}  & \leq \alpha^{\infty} := 1 + 8.824 \; 10^{-3},\\
  \normInf{\beta}& \leq \beta^{\infty} := 3.294 \; 10^{-2} \,
  \lambda_c,
\end{aligned}
\end{equation}
such that \eqref{eq:interp1}--\eqref{eq:interp2} hold. Further, if
$v_1=1$,
\begin{equation}
 \label{bound_Re_alpha1}
\begin{aligned}
  \operatorname{Re}{\alpha_1} & \geq 1-8.824 \; 10^{-3},\\
  \abs{\operatorname{Im}{\alpha_1}} & \leq 8.824 \; 10^{-3}. 
\end{aligned}
\end{equation} 
\end{lemma}
To complete the proof, Lemmas \ref{lemma:concavity} and
\ref{lemma:boundq} show that $|q\brac{t}| < 1$.
\begin{lemma}
\label{lemma:concavity}
Fix $\tau \in T$. Under the hypotheses of Proposition
\ref{prop:continuous_dualcert}, $|q(t)| < 1$ for $\abs{t - \tau} \in
(0, \tC \, \lambda_c]$.
\end{lemma}
\begin{lemma}
\label{lemma:boundq}
Fix $\tau \in T$. Then under the hypotheses of Proposition
\ref{prop:continuous_dualcert}, $|q(t)| < 1$ for $\abs{t - \tau} \in
[\tC \, \lambda_c, \Delta/2]$. This can be extended as follows:
letting $\tau_+$ be the closest spike to the right, i.~e.~$\tau_+ =
\min \{t \in T : t > \tau\}$. Then $\abs{q(t)} < 1$ for all $t$
obeying $0 < t - \tau \leq (\tau_+ - \tau)/2$, and likewise for the left
side.
\end{lemma}
Finally, we record a useful lemma to derive stability results. 
\begin{lemma}
\label{lemma:qbound_strict}
If $\Delta\brac{T} \geq 2.5 \, \lambda_c$, then for any $\tau \in T$,
\begin{equation}
  \abs{q\brac{t}} \leq 1-0.3353 \fc^2 \brac{t-\tau}^2, \quad \text{for all } t: \, \abs{t-\tau}\leq \tC \, \lambda_c. \label{qbound_strict}
\end{equation}
Further, for $\min_{\tau \in T} \, \abs{t-\tau}> \tC \, \lambda_c$,
$\abs{q\brac{t}}$ is upper bounded by the right-hand side above
evaluated at $\tC \, \lambda_c$. 
\end{lemma}
Section \ref{subsec:real} describes how the proof can be adapted to
obtain a slightly smaller bound on the minimum distance for
real-valued signals.

\subsection{Proofs of Lemmas}
\label{subsec:proofs_lemmas}
The proofs of the three lemmas above make repeated use of the fact
that the interpolation kernel and its derivatives decay rapidly away
from the origin. The intermediate result below proved in Section
\ref{proof:fejersq_bounds} of the Appendix quantifies this.
\begin{lemma}
\label{lemma:fejersq_bounds}
For $\ell \in \{0,1,2,3\}$, let $K^{(\ell)}$ be the $\ell$th
  derivative of $K$ ($K = K^{(0)}$).  For $\frac12 f_c^{-1} = \frac12 \lambda_c \leq t
  \leq \frac12$, we have
\[
\abs{K^{(\ell)}(t)} \le B_\ell(t)= \begin{cases}\tilde{B}_\ell(t)= \frac{\pi^\ell H_\ell(t) }{(f_c+2)^{4-\ell} \, t^4} , \qquad &  \frac12 \lambda_c \leq t \leq \sqrt{2}/\pi, \\
                         \frac{\pi^\ell H_\ell^\infty }{(f_c+2)^{4-\ell} \, t^4}, \qquad &  \sqrt{2}/\pi \leq t < \frac12,
                        \end{cases}
\]
where $H_0^\infty =  1$, $H_1^\infty  =  4$, $H_2^\infty  =   18$, $ H_3^\infty = 77$,
\begin{align*}
  H_0(t) & =  a^4(t),\\
  H_1(t) & = a^4(t)\brac{2+2b(t)},\\
  H_2(t) & =    a^4(t)\brac{4+7b(t)+6b^2(t)},\\
  H_3(t) & = a^4(t)\brac{8+24b(t)+30b^2(t) + 15 b^3(t)},
\end{align*}
and
\begin{align*}
a(t) & =  \frac{2}{\pi \brac{1-\frac{\pi^2 t^2}{6}}},
\qquad 
  b(t)  = \frac{1}{f_c} \, \frac{a(t)}{t}. 
\end{align*} 

For each $\ell$, the bound on the magnitude of $B_\ell(t)$ is nonincreasing in $t$ and
$\tilde{B}_\ell(\Delta-t) + \tilde{B}_\ell(\Delta+t)$ is increasing in $t$ for $0 \leq t<\Delta/2$ if $0 \leq
\Delta + t \leq \sqrt{2}/\pi$.
\end{lemma}
This lemma is used to control quantities of the form $\sum_{t_i \in
  T\setminus \{\tau\}} \abs{K\brac{t-t_i}}$ ($\tau \in T$) as shown
below.
\begin{lemma}
\label{lemma:bound_sum_kernel}
Suppose $0 \in T$. Then for all $t \in [0,\Delta/2]$,

\[
\sum_{t_i \in T\setminus \{0\}} \abs{K^{\brac{\ell}}\brac{t-t_i}} \leq
F_\ell\brac{\Delta,t} = F_\ell^+(\Delta,t) + F_\ell^-(\Delta,t) + F_\ell^{\infty}(\Deltaminth),
\]
where 
\begin{align*}
  F_\ell^+(\Delta,t) & = \max \keys{\max_{\Delta \leq \tp \leq 3 \Deltaminth} \,
    \abs{K^{\brac{\ell}}\brac{t-\tp}}, B_\ell\brac{3 \Deltaminth-t}} +
  \sum_{j = 2}^{20} \tilde{B}_\ell(j \Deltaminth - t),\\
  F_\ell^-(\Delta,t) & = \max \keys{\max_{ \Delta \leq \tm \leq 3
      \Deltaminth} \, \abs{K^{\brac{\ell}}\brac{\tm}}, B_\ell\brac{3
      \Deltaminth}} + \sum_{j = 2}^{20} \tilde{B}_\ell(j \Deltaminth + t), \\
  F_\ell^{\infty}(\Deltaminth) & = \frac{\kappa  \, \pi^\ell  H_\ell^\infty }{(\fc+2)^{4-\ell} \Deltaminth
  ^4 }, \qquad \kappa =  \frac{\pi^4}{45}-2\sum_{j=1}^{19} \frac{1}{j^4} \leq 8.98 \, 10^{-5}.
\end{align*}
Moreover, $F_\ell\brac{\Delta,t}$ is nonincreasing in $\Delta$ for all
$t$, and $F_\ell\brac{\Deltaminth,t}$ is nondecreasing in $t$. 
\end{lemma}
\begin{proof}
  We consider the sum over positive $t_i \in T$ first and denote by
  $\tp$ the positive element in $T$ closest to $0$. We have
\begin{equation}
\label{eq:Tpos}
\sum_{t_i \in T: \, 0 < t_i \le 1/2} \abs{K^{\brac{\ell}}\brac{t-t_i}} =
\abs{K^{\brac{\ell}}\brac{t-\tp}} + \sum_{t_i \in T\setminus\{t_+\}: \,
  0 < t_i \le 1/2} \abs{K^{\brac{\ell}}\brac{t-t_i}}.
\end{equation}
Let us assume $t_+ < 2 \Deltamin $ (if $t_+ > 2 \Deltamin$ the argument is very similar). Note that the assumption that $\fc \geq 128$ implies $21 \Deltamin < 0.33 < \sqrt{2}/\pi $. By Lemma~\ref{lemma:fejersq_bounds} and the minimum separation condition, this means that the second term in the right-hand side is at most 
\begin{equation}
\sum_{j = 2}^{20} \tilde{B}_\ell(j \Deltamin - t) + \frac{\pi^\ell}{(\fc+2)^{4-\ell} } \sum_{j=21}^{\infty} \frac{H_\ell^\infty}{\brac{j \Deltamin \pm t}^4}, \label{sumB_formula}
\end{equation}
which can be upper bounded using the fact that
\[
\sum_{j=21}^{\infty} \frac{H_\ell^\infty}{\brac{j \Deltamin \pm t}^4}
\leq \sum_{j=20}^{\infty} \frac{H_\ell^\infty}{\brac{j \Deltamin}^4} =
\frac{H_\ell^\infty}{ \Deltamin ^4}\brac{\sum_{j=1}^{\infty} \frac{1}{j^4}
  - \sum_{j=1}^{19} \frac{1}{j^4}} = \frac{H_\ell^\infty}{ \Deltamin
  ^4}\brac{ \frac{\pi^4}{90}-\sum_{j=1}^{19} \frac{1}{j^4}} = \frac{\kappa H_\ell^\infty}{2 \Deltamin
  ^4};
\]
the first inequality holds because $t < \Deltamin $ and the last because the Riemann
zeta function is equal to $\pi^4/90$ at 4. Also, 
\[
\abs{K^{\brac{\ell}}\brac{t-\tp}} \le \begin{cases} \max_{\Delta \leq
    \tp \leq 3 \Deltamin} \, \abs{K^{\brac{\ell}}\brac{t-\tp}}, &
  t_+ \le 3\Deltamin, \\
  B_\ell(3\Deltamin - t), & t_+ > 3\Deltamin. \end{cases}
\]
Hence, the quantity in \eqref{eq:Tpos} is bounded by
$F_\ell^+(\Delta,t)+F_\ell^{\infty}(\Deltamin)/2$. A similar argument shows that the sum over
negative $t_i \in T$ is bounded by $F_\ell^-(\Delta,t)+F_\ell^{\infty}(\Deltamin)/2$.

To verify the claim about the monotonicity w.r.t.~$\Delta$, observe
that both terms
\[
\max \keys{\max_{\Delta \leq \tp \leq 3 \Deltamin} \,
  \abs{K^{\brac{\ell}}\brac{t-\tp}}, B_\ell\brac{3 \Deltamin - t}}
\text{ and } \max \keys{\max_{ \Delta \leq \tm \leq 3 \Deltamin} \,
  \abs{K^{\brac{\ell}}\brac{\tm}}, B_\ell\brac{3 \Deltamin}}
\]
are nonincreasing in $\Delta$.  

Fix $\Delta = \Deltamin$ now.  Since $\tilde{B}_\ell\brac{j \Delta -t}
+\tilde{B}_\ell\brac{j \Delta+t}$ is increasing in $t$ for $j\leq 20$ (recall that $21 \Deltamin < \sqrt{2}/\pi $), we only need to check that the first term in the expression for $F_\ell^+$ is
nondecreasing in $t$. To see this, rewrite this term (with
$\Delta = \Deltamin$) as
\[
\max \keys{\max_{\Deltamin - t \leq u \leq 3 \Deltamin - t} \,
  \abs{K^{\brac{\ell}}\brac{u}}, B_\ell\brac{3 \Deltamin - t}}.
\]
Now set $t' > t$. Then by Lemma \ref{lemma:fejersq_bounds}, 
\[
B_\ell\brac{3 \Deltamin-t'} \ge \begin{cases} B_\ell\brac{3 \Deltamin-t}, & \\
\abs{K(u)}, & u \ge 3\Deltamin - t'. \end{cases}
\] 
Also, we can verify that  
\[
\max_{\Deltamin - t' \leq u \leq 3 \Deltamin - t'} \,
\abs{K^{\brac{\ell}}\brac{u}} \ge \max_{\Deltamin - t \leq u \leq 3
  \Deltamin - t'} \, \abs{K^{\brac{\ell}}\brac{u}}.
\]
This concludes the proof.
\end{proof}

In the proof of Lemmas \ref{lemma:concavity} and \ref{lemma:boundq},
it is necessary to find a numerical upper bound on
$F_{\ell}(\Deltamin,t)$ at $t \in \{0, \tC \, \lambda_c,0.4269  \, \lambda_c, \tA \,
\lambda_c\}$ (for the last two, we only need bounds for $\ell=0,1$). For a fixed $t$, it is easy to find the maximum of
$\abs{K^{\brac{\ell}}\brac{t-\tp}}$ where $\tp$ ranges over
$[\Deltamin, 3\Deltamin]$ since we have expressions for the smooth
functions $K^{(\ell)}$ (see Section \ref{proof:fejersq_bounds} of the
Appendix). For reference, these functions are plotted in Figure
\ref{fig:kernel}. The necessary upper bounds are gathered in Table \ref{table:F_tC}.

\begin{table}[htbp ]
\begin{center}
\begin{tabular}{| c | c | c | c | c |}
	\hline
	$t/\lambda_c$ & $F_{0}\brac{1.98 \lambda_c,t}$ & $F_{1}\brac{1.98 \lambda_c,t}$ & $F_{2}\brac{1.98 \lambda_c,t}$ & $F_{3}\brac{1.98 \lambda_c,t}$ \\
	\hline
	$0$ & $ 6.253\; 10^{-3}$ & $7.639\; 10^{-2} \fc$ & $1.053\, \fc^2$ & $ 8.078\, \fc^3$ \\
	\hline
	$\tC$ & $ 6.279\; 10^{-3}$ & $7.659\; 10^{-2} \fc$ & $1.055\, \fc^2$ & $ 18.56\, \fc^3$ \\
	\hline
	$0.4269$ & $ 8.029\; 10^{-3}$ & $0.3042 \fc$ &  &   \\
	\hline
	$0.7559$ & $ 5.565\; 10^{-2}$ & $1.918 \fc$ &  &  \\
	\hline
\end{tabular}
\vspace{10pt}
\caption{Numerical upper bounds on $F_{\ell}(1.98\lambda_c, t)$.}
  \label{table:F_tC}
\end{center}
\end{table}
\begin{figure}
\centering
\subfigure[]{
\includegraphics[scale=0.4]{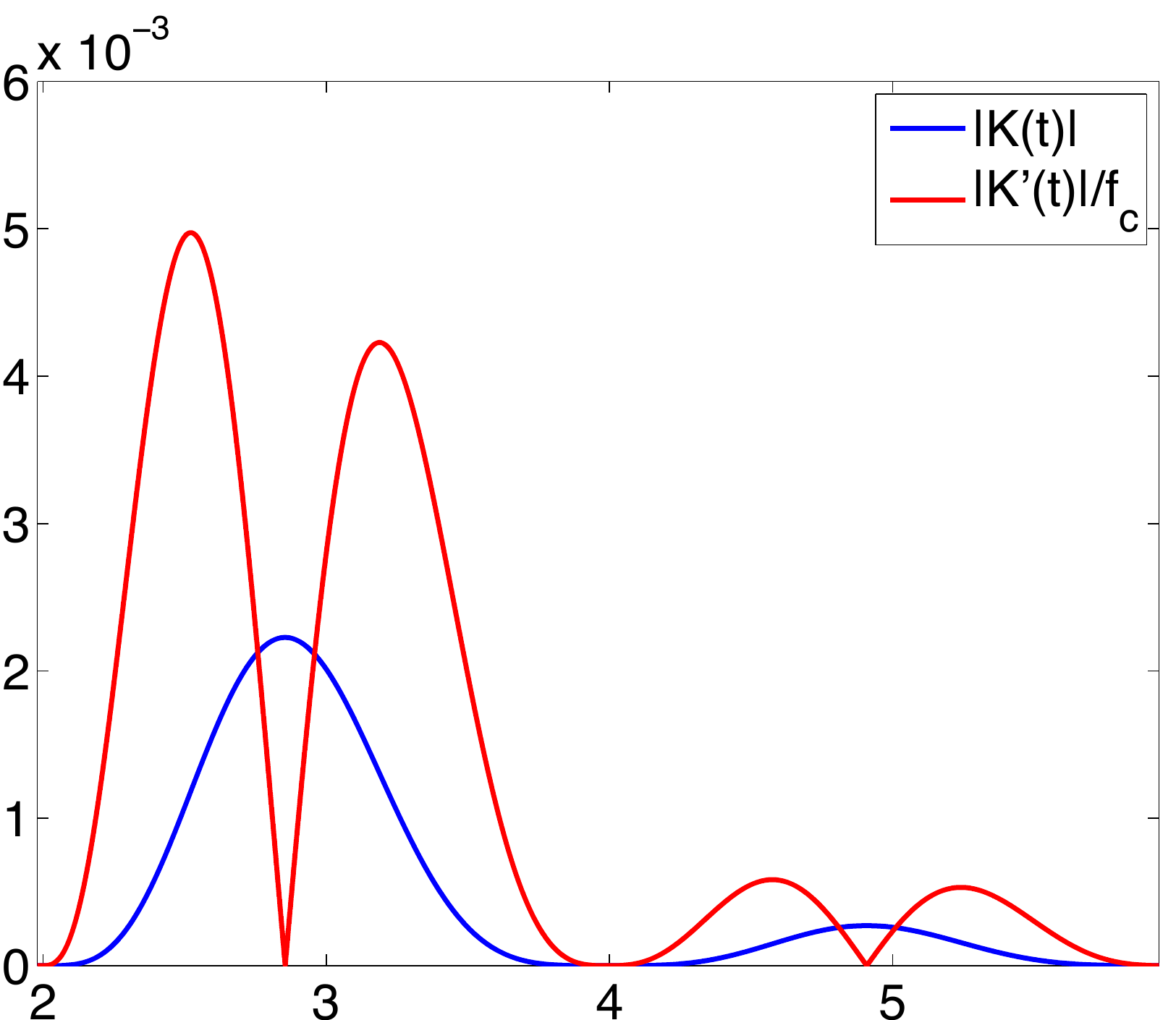}
\label{fig:kernel_a}
}
\subfigure[]{
\includegraphics[scale=0.4]{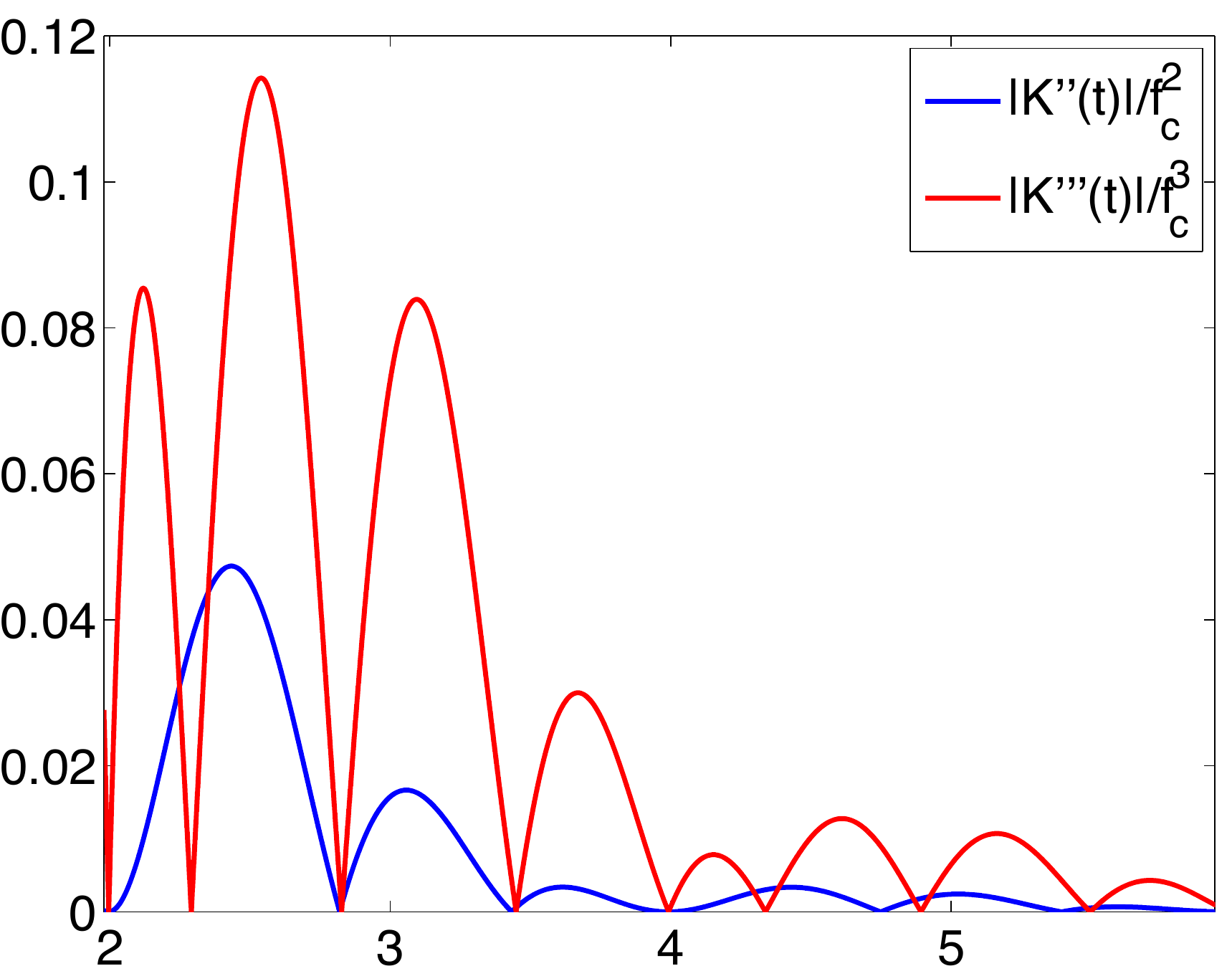}
\label{fig:kernel_b}
}
\caption{$\abs{K^{(\ell)}\brac{t}}$ for $t \in [\Deltamin,
  3\Deltamin]$. The scaling of the $x$-axis is in units of
  $\lambda_c$.}
\label{fig:kernel}
\end{figure}

Finally, a last fact we shall use is that $K\brac{0}=1$ is the global
maximum of $K$ and $\abs{K''\brac{0}}=\abs{-\pi^2\fc \brac{\fc+4}/3}$ the
global maximum of $\abs{K''}$.

\subsubsection{Proof of Lemma \ref{lemma:fejersq_coeffs}}
\label{subsec:fejersq_coeffs}

Set
\[
\brac{D_0}_{jk} = K\brac{t_j - t_k}, \quad \brac{D_1}_{jk} =
K'\brac{t_j - t_k}, \quad \brac{D_2}_{jk} = K''\brac{t_j - t_k},
\]
where $j$ and $k$ range from $1$ to $\abs{T}$. With this,
\eqref{eq:interp1} and \eqref{eq:interp2} become
\[
\begin{bmatrix} D_0 & D_1\\ D_1 & D_2 \end{bmatrix} \begin{bmatrix} \alpha \\
  \beta \end{bmatrix} =\begin{bmatrix} v\\0 \end{bmatrix}.
\]
A standard linear algebra result asserts that this system is
invertible if and only if $D_2$ and its Schur complement $D_0 -
D_1D_2^{-1}D_1$ are both invertible. To prove that this is the case we
can use the fact that a symmetric matrix $M$ is invertible if
\begin{equation}
 \normInfInf{\Id- M} < 1, \label{inf_invertibility}  
\end{equation}
where $\|A \|_\infty$ is the usual infinity norm of a matrix defined as
$\|A\|_\infty = \max_{\|x\|_\infty = 1} \|Ax\|_\infty = \max_i \sum_j
|a_{ij}|$. This follows from $M^{-1} = (I - H)^{-1} = \sum_{k \ge 0}
H^k$, $H = I - M$, where the series is convergent since $
\normInfInf{H} < 1$. In particular,
\begin{equation}
  \normInfInf{M^{-1}}  \leq \frac{1}{1-\normInfInf{\Id-M}}.
\label{infinfnorminv}
\end{equation}

We also make use of the inequalities below, which follow from Lemma
\ref{lemma:bound_sum_kernel}
\begin{align}
\normInfInf{\Id- D_0} & \le
  F_{0}\brac{\Deltamin,0}\le  6.253 \; 10^{-3},\label{Id_F_bound}\\
\normInfInf{D_1}   & \le
  F_{1}\brac{\Deltamin,0} \leq 7.639 \; 10^{-2}\; \fc, \label{D1_bound}\\
\normInfInf{\abs{K''\brac{0}}\Id-D_2} & \le
  F_{2}\brac{\Deltamin,0} \leq  1.053\; \fc^2
\label{Id_D2_bound}\text{.}
\end{align}

Note that $D_2$ is symmetric because the second derivative of the
interpolation kernel is symmetric. The bound \eqref{Id_D2_bound} and
the identity $K''\brac{0}=-\pi^2 \fc\brac{\fc+4}/3$ give
\[
\normInfInf{\Id-\frac{D_2}{\abs{K''\brac{0}}}}  < 1, 
\]
which implies the invertibility of $D_2$.  The bound
\eqref{infinfnorminv} then gives
\begin{equation}
\normInfInf{D_2^{-1}}  \leq \frac{1}{\abs{K''\brac{0}} - \normInfInf{\abs{K''\brac{0}}\Id-D_2}} \leq \frac{0.4275}{\fc^2}. \label{D2_inv_bound}
\end{equation}
Combining this with \eqref{Id_F_bound} and \eqref{D1_bound} yields 
\begin{align}
\normInfInf{\Id- \brac{D_0-D_1D_2^{-1}D_1}} & \leq \normInfInf{\Id- D_0} + \normInfInf{D_1}^2\normInfInf{D_2^{-1}}\notag\\
& \leq 8.747  \; 10^{-3} <1 \label{invschur_bound}\text{.}
\end{align}
Note that the Schur complement of $D_2$ is symmetric because $D_0$ and
$D_2$ are both symmetric whereas $D_1^T=-D_1$ since the derivative of
the interpolation kernel is odd. This shows that the Schur complement
of $D_2$ is invertible and, therefore, the coefficient vectors
$\alpha$ and $\beta$ are well defined.

There just remains to bound the interpolation coefficients, which can
be expressed as
\[
  \begin{bmatrix} \alpha \\
    \beta \end{bmatrix} =\begin{bmatrix} \Id \\
    -D_2^{-1}D_1 \end{bmatrix} C^{-1} v, \quad C :=
  D_0-D_1D_2^{-1}D_1,
\]
where $C$ is the Schur complement.  The relationships
\eqref{infinfnorminv} and \eqref{invschur_bound} immediately give a
bound on the magnitude of the entries of $\alpha$
\[
\normInf{\alpha} = \normInf{C^{-1} v} \leq \normInfInf{C^{-1}} \leq 1+
8.824 \; 10^{-3}.
\]
Similarly, \eqref{D1_bound}, \eqref{D2_inv_bound} and
\eqref{invschur_bound} allow to bound the entries of $\beta$:
\begin{align*}
  \normInf{\beta} & \leq \normInfInf{D_2^{-1}D_1 C^{-1}}\\
  & \leq \normInfInf{D_2^{-1}} \normInfInf{D_1}
  \normInfInf{C^{-1}} \leq 3.294 \; 10^{-2} \,
  \lambda_c.
\end{align*}
Finally, with $v_1=1$, we can use \eqref{invschur_bound} to show that
$\alpha_1$ is almost equal to 1. Indeed,  
\[
\alpha_1 = 1 - \gamma_1, \quad \gamma_1 = [(\Id - C^{-1})v]_1,  
\]
$\abs{\gamma_1} \le \normInfInf{\Id - C^{-1}}$, and
\[
\normInfInf{\Id - C^{-1}} = \normInfInf{C^{-1} (\Id - C)} \le
\normInfInf{C^{-1}} \, \normInfInf{\Id - C} \leq 8.824 \; 10^{-3}.
\]
This concludes the proof. 

\subsubsection{Proof of Lemma \ref{lemma:concavity}}
\label{subsec:concavity}

We assume without loss of generality that $\tau = 0$ and $q(0) =
1$. By symmetry, it suffices to show the claim for $t \in (0,\tC \,
\lambda_c]$. Since $q'(0) = 0$, local strict concavity would imply
that $\abs{q\brac{t}}<1$ near the origin. We begin by showing that the
second derivative of $\abs{q}$ is strictly negative in the interval
$\brac{0,\tC \, \lambda_c}$. This derivative is equal to
\[
\derTwo{t}{\abs{q}}\brac{t}= -\frac{\brac{q_R\brac{t}q_R'\brac{t}+
    q_I\brac{t}q_I'\brac{t}}^2}{\abs{q\brac{t}}^3}+
\frac{\abs{q'\brac{t}}^2+q_R\brac{t}q_R''\brac{t}+ q_I\brac{t}
  q_I''\brac{t}}{\abs{q\brac{t}}} \text{,}
\]
where $q_R$ is the real part of $q$ and $q_I$ the imaginary part. As a
result, it is sufficient to show that
\begin{equation}
  q_R\brac{t}q_R''\brac{t} + \abs{q'\brac{t}}^2+
  \abs{q_I\brac{t}}\abs{ q_I''\brac{t}}<0\text{,} \label{ineq_condition}
\end{equation}
as long as $\abs{q\brac{t}}$ is bounded away from zero.  In order to
bound the different terms in \eqref{ineq_condition}, we use the series
expansions of the interpolation kernel and its derivatives around the
origin to obtain the inequalities, which hold for all $t \in
[-1/2,1/2]$,
\begin{align}
  K\brac{t} & \geq  1 - \frac{\pi^2}{6} \fc \brac{\fc+4} t^2,\label{boundK_near}\\
  \abs{K'\brac{t}} & \leq  \frac{\pi^2}{3} \fc\brac{\fc+4} \, t, \label{boundKp_near}\\
  K''\brac{t} & \leq   -\frac{\pi^2}{3} \fc\brac{\fc+4} + \frac{ \pi^4}{6} \brac{\fc+2}^4 \, t^2, \label{boundKp2_near_low}\\
  \abs{K''\brac{t}} & \leq  \frac{\pi^2}{3} \fc\brac{\fc+4},\label{boundKp2_near}\\
  \abs{K'''\brac{t}} & \leq \frac{\pi^4}{3} \brac{\fc+2}^4 \, t.\label{boundKp3_near}
\end{align}
The lower bounds are decreasing in $t$, while the upper bounds are
increasing in $t$, so we can evaluate them at $\tC \, \lambda_c$ to
establish that for all $t \in [0, \tC \, \lambda_c]$,
\begin{equation}
  \label{boundK}
\begin{alignedat}{3}
K\brac{t} & \geq
0.9539, 
& \qquad K''\brac{t}  & \leq -2.923 \, \fc^2, & \qquad & \\
\abs{K'\brac{t}} & \leq 0.5595 \, \fc, & \qquad \abs{K''\brac{t}}
& \leq 3.393 \, \fc^2, & \qquad \abs{K'''\brac{t}} & \leq 5.697 \,
\fc^3. 
\end{alignedat}
\end{equation}
We combine this with Lemmas \ref{lemma:bound_sum_kernel} and
\ref{lemma:fejersq_coeffs} to control the different terms in
\eqref{ineq_condition} and begin with $q_R\brac{t}$. Here,
\begin{align}
  q_R\brac{t} & = \sum_{t_j \in T} \Real{\alpha_j} K\brac{t-t_j} + \Real{\beta_j} K'\brac{t-t_j}\notag\\
  & \geq \Real{\alpha_1}K\brac{t} -\normInf{\alpha}\sum_{t_j \in T\setminus\keys{0}}\abs{K\brac{t-t_j}} - \normInf{\beta}\sum_{t_j \in T} \abs{K'\brac{t-t_j}}\notag\\
  & \geq \Real{\alpha_1}K\brac{t}
  -\alpha^\infty F_{0}\brac{\Delta,t} -
  \beta^\infty \brac{\abs{K'\brac{t}} + F_{1}\brac{\Delta,t}} \notag\\
& \geq
  \Real{\alpha_1}K\brac{t} -\alpha^\infty F_{0}\brac{\Deltamin,t} -
  \beta^\infty \brac{\abs{K'\brac{t}} + F_{1}\brac{\Deltamin,t}}
  \geq 0.9182\text{.}\label{bound_qR}
\end{align}
The third inequality follows from the monotonicity of $F_\ell$ in
$\Delta$, and the last from \eqref{boundK} together with the
monotonicity of $F_{1}\brac{\Deltamin,t}$ in $t$, see Lemma
\ref{lemma:bound_sum_kernel}, so that we can plug in $t = \tC \,
\lambda_c$. Observe that this shows that $q$ is bounded away from zero
since $\abs{q\brac{t}} \geq q_R\brac{t} \geq 0.9198$.  Very similar
computations yield
\begin{align*}
\abs{q_I\brac{t}} & = \abs{\sum_{t_j \in T} \Imag{\alpha_j} K\brac{t-t_j} + \Imag{\beta_j} K'\brac{t-t_j}}\\
& \leq \abs{\Imag{\alpha_1}} + \normInf{\alpha}\sum_{t_j \in T\setminus\keys{0}}\abs{K\brac{t-t_j}} + \normInf{\beta}\sum_{t_j \in T} \abs{K'\brac{t-t_j}}\\
& \leq \abs{\Imag{\alpha_1}} + \alpha^\infty F_{0}\brac{\Deltamin,t} + \beta^\infty \brac{\abs{K'\brac{t}}+F_{1}\brac{\Deltamin,t}} \leq 3.611 \; 10^{-2}\end{align*}
and 
\begin{align}
  q_R''\brac{t} & = \sum_{t_j \in T}\Real{\alpha_j} K''\brac{t-t_j} + \sum_{t_j \in T}\Real{\beta_j} K'''\brac{t-t_j}\notag\\
  & \leq \Real{\alpha_1}K''\brac{t}+ \normInf{\alpha}\sum_{t_j \in T\setminus\keys{0}}\abs{K''\brac{t-t_j}} + \normInf{\beta}\sum_{t_j \in T} \abs{K'''\brac{t-t_j}}\notag\\
& \leq \Real{\alpha_1}K''\brac{t} + \alpha^\infty F_{2}\brac{\Deltamin,t} + \beta^\infty \brac{\abs{K'''\brac{t}}+ F_{3}\brac{\Deltamin,t}} \leq -1.034 \, \fc^2. \label{bound_qR2p}
\end{align}
Similarly, 
\begin{align*}
  \abs{q_I''\brac{t}} & = \abs{\sum_{t_j \in T}\Imag{\alpha_j} K''\brac{t-t_j} + \sum_{t_j \in T}\Imag{\beta_j} K'''\brac{t-t_j}}\\
  & \leq \Imag{\alpha_1}\abs{K''\brac{t}}+ \normInf{\alpha}\sum_{t_j \in T\setminus\keys{0}}\abs{K''\brac{t-t_j}} + \normInf{\beta}\sum_{t_j \in T} \abs{K'''\brac{t-t_j}}\\
& \leq \Imag{\alpha_1}\abs{K''\brac{t}}+ \alpha^\infty F_{2}\brac{\Deltamin,t} + \beta^\infty \brac{\abs{K'''\brac{t}}+ F_{3}\brac{\Deltamin,t}} \leq 1.893 \, \fc^2
\end{align*}
and 
\begin{align*}
  \abs{q'\brac{t}} & = \abs{\sum_{t_j \in T} \alpha_j K'\brac{t-t_j} + \beta_j K''\brac{t-t_j}}\\
  & \leq \normInf{\alpha} \sum_{t_j \in T}\abs{K'\brac{t-t_j}} + \normInf{\beta}\sum_{t_j \in T} \abs{K''\brac{t-t_j}}\\
  & \leq \alpha^\infty \abs{K'\brac{t}}+ \alpha^\infty
  F_{1}\brac{\Deltamin,t} + \beta^\infty
  \brac{\abs{K''\brac{t}}+F_{2}\brac{\Deltamin,t}} \leq 0.7882\, \fc
  \text{.}
\end{align*}
These bounds allow us to conclude that $|q|''$ is negative on $[0, \tC
\lambda_c]$ since
\[
q_R\brac{t}q_R''\brac{t} + \abs{q'\brac{t}}^2+ \abs{q_I\brac{t}}\abs{
  q_I''\brac{t}} \leq -9.291 \; 10^{-2} \fc^2 < 0.
\]
This completes the proof.

\subsubsection{Proof of Lemma \ref{lemma:boundq}}
\label{subsec:boundq}

As before, we assume without loss of generality that $\tau = 0$ and
$q(0) =1$.  We use Lemma \ref{lemma:bound_sum_kernel} again to bound
the absolute value of the dual polynomial on $[\tC \lambda_c,
\Delta/2]$ and write
\begin{align}
  \abs{q\brac{t}} & = \Bigl| \sum_{t_j \in T} \alpha_j K\brac{t-t_j} + \beta_j K'\brac{t-t_j}\Bigr| \notag\\
  & \leq \normInf{\alpha}\Bigl[\abs{K\brac{t}} + \sum_{t_j \in T\setminus\keys{0}} \abs{K\brac{t-t_j}}\Bigr] +\normInf{\beta}\Bigl[\abs{K'\brac{t}} + \sum_{t_j \in T\setminus\keys{0}} \abs{K'\brac{t-t_j}}\Bigr] \notag\\
  & \leq \alpha^\infty \abs{K\brac{t}} + \alpha^\infty
  F_{0}\brac{\Deltamin,t}  +\beta^\infty \abs{K'\brac{t}} +  \beta^\infty
  F_{1}\brac{\Deltamin,t}. \label{boundqabs}
\end{align}
Note that we are assuming adversarial sign patterns and as a result we
are unable to exploit cancellations in the coefficient vectors
$\alpha$ and $\beta$. To control $|K(t)|$ and $|K'(t)|$ between $\tC
\, \lambda_c$ and $\tA \, \lambda_c$, we use series expansions around
the origin which give
\begin{equation}
  \label{KboundUp}
\begin{aligned} 
  K\brac{t} & \leq 1 - \frac{\pi^2 \fc \brac{\fc+4} t^2}{6} +
  \frac{\pi^4\brac{\fc+2}^4 t^4}{72} \\
  \abs{K'\brac{t}} & \leq \frac{\pi^2 \fc\brac{\fc+4}t}{3}, 
\end{aligned} 
\end{equation}
for all $t \in [-1/2,1/2]$. Put 
\[
L_{1}\brac{t} = \alpha^{\infty} \Bigl[ 1 - \frac{\pi^2 \fc
  \brac{\fc+4} t^2}{6} + \frac{\pi^4\brac{\fc+2}^4 t^4}{72}\Bigr] +
\beta^{\infty} \frac{\pi^2 \fc\brac{\fc+4}t}{3},
\]
with derivative equal to
\[
L_{1}'\brac{t} = -\alpha^{\infty} \Bigl[\frac{\pi^2 \fc \brac{\fc+4}
  t}{3} - \frac{\pi^4\brac{\fc+2}^4 t^3}{18}\Bigr] + \beta^{\infty}
\frac{\pi^2 \fc\brac{\fc+4}}{3} \text{.}
\]
This derivative is strictly negative between $\tC \, \lambda_c$ and
$\tA \, \lambda_c$, which implies that $L_{1}\brac{t}$ is decreasing
in this interval. Put
\[
 L_{2}\brac{t}  = \alpha^{\infty} F_{0}\brac{\Deltamin,t} + \beta^{\infty}  F_{1}\brac{\Deltamin,t}.
\]  
By Lemma \ref{lemma:bound_sum_kernel}, this function is
increasing. With \eqref{boundqabs}, this gives the crude bound
\begin{equation}
\abs{q\brac{t}} \leq L_{1}\brac{t}+L_{2}\brac{t} \le
L_{1}\brac{t_1}+L_{2}\brac{t_2} \quad \text{for all } t \in [t_1,
t_2].
\label{t1t2_bound1}
\end{equation}
Table \ref{table:L_bounds} shows that taking $\keys{t_1,t_2} = \keys{
  \tC \, \lambda_c,\tx \, \lambda_c}$ {and then} $\keys{t_1,t_2}
=\keys{ \tx \, \lambda_c,\tA \, \lambda_c}$ proves that $\abs{q(t)} <
1$ on $[\tC \lambda_c, \tA \lambda_c]$.  For $\tA \lambda_c \leq t
\leq \Delta/2$, we apply Lemma \ref{lemma:fejersq_bounds} and obtain
\begin{align*}
  |q(t)| &\leq \alpha^{\infty}\Bigl[ B_0\brac{t} + B_0\brac{\Delta-t}+ \sum_{j=1}^{\infty} B_0\brac{j\Deltamin+\Delta-t}+\sum_{j=1}^{\infty} B_0\brac{j\Deltamin+t}\Bigr] \\
  & \qquad \qquad +\beta^{\infty}\Bigl[B_1\brac{t} + B_1\brac{\Delta-t}+ \sum_{j=1}^{\infty} B_1\brac{j\Deltamin+\Delta-t}+\sum_{j=1}^{\infty} B_1\brac{j\Deltamin+t}\Bigr] \\
  &\leq \alpha^{\infty}\Bigl[B_0\brac{\tA \lambda_c}+ \sum_{j=1}^{\infty} B_0\brac{j\Deltamin-\tA \lambda_c}+ \sum_{j=1}^{\infty} B_0\brac{j\Deltamin+\tA \lambda_c}\Bigr]  \\
  &\qquad \qquad + \beta^{\infty}\Bigl[B_1\brac{\tA \lambda_c} +
  \sum_{j=1}^{\infty} B_1\brac{j\Deltamin-\tA \lambda_c}+
  \sum_{j=1}^{\infty} B_1\brac{j\Deltamin+\tA \lambda_c}\Bigr]  \\
  & \leq 0.758;
  \end{align*}
  here, the second step follows from the monotonicity of $B_0$ and
  $B_1$. Finally, for $\Delta/2 \le t \le t_+/2$, this last inequality
  applies as well.  This completes the proof.
\begin{table}[htbp]
\begin{center}
\begin{tabular}{| c | c | c | c |}
	\hline
	$t_1/\lambda_c$ & $t_2/\lambda_c$ & $L_{1}\brac{t_1}$ & $L_{2}\brac{t_2}$\\
	\hline
	$\tC $ & $ \tx$ & $0.9818$ &  $1.812 \, 10^{-2}$\\
	\hline
	$\tx $ & $ \tA$ & $0.7929$ &  $0.2068$  \\
	\hline
\end{tabular}
\vspace{10pt}
\caption{Numerical quantities used in \eqref{t1t2_bound1}.}
  \label{table:L_bounds}
\end{center}
\end{table}
\subsection{Proof of Lemma \ref{lemma:qbound_strict}}

Replacing $\Delta=1.98\, \lambda_c$ by $\Delta=2.5\, \lambda_c$ and
going through exactly the same calculations as in Sections
\ref{subsec:fejersq_coeffs} and \ref{subsec:concavity} yields that for
any $t$ obeying $0 \le \abs{t - \tau} \le \tC \, \lambda_c$,
\begin{equation*}
\derTwo{t}{\abs{q}}\brac{t} \leq -0.6706 \, \fc^2.
\end{equation*}
For reference, we have computed in Table \ref{table:F_2p5} numerical
upper bounds on $F_{\ell}(2.5 \lambda_c,t)$ at $t = \{0, \tC \,
\lambda\}$. Since $|q(0)| = 1$ and $q'(0) = 0$, it follows that
\begin{equation}
  \abs{q\brac{t}}  \leq \abs{q\brac{0}} - \frac12 \, {0.6706 \,\fc^2} t^2. 
\label{q_strong}
\end{equation}
At a distance of $\tC \, \lambda_c$, the right-hand side is equal to
$0.9909$. The calculations in Section \ref{subsec:boundq} with
$\Delta=2.5\, \lambda_c$ imply that the magnitude of $q(t)$ at
locations at least $\tC \, \lambda_c$ away from an element of $T$ is
bounded by $0.9843$. This concludes the proof.
\begin{table}[htbp ]
\begin{center}
\begin{tabular}{| c | c | c | c | c |}
	\hline
	$t/\lambda_c$ & $F_{0}\brac{2.5 \lambda_c,t}$ & $F_{1}\brac{2.5 \lambda_c,t}$ & $F_{2}\brac{2.5 \lambda_c,t}$ & $F_{3}\brac{2.5 \lambda_c,t}$ \\
	\hline
	$0$ & $ 5.175\; 10^{-3}$ & $6.839\; 10^{-2} \fc$ & $0.8946\, \fc^2$ & $ 7.644\, \fc^3$ \\
	\hline
	$\tC$ & $ 5.182\; 10^{-3}$ & $6.849\; 10^{-2} \fc$ & $0.9459\, \fc^2$ & $ 7.647\, \fc^3$ \\
	\hline
\end{tabular}
\caption{Numerical upper bounds on $F_{\ell}(2.5 \lambda_c, t)$.}
  \label{table:F_2p5}
\end{center}
\end{table}

\subsection{Improvement for real-valued signals}
\label{subsec:real}

The proof for real-valued signals is almost the same as the one we
have discussed for complex-valued signals---only simpler. The only
modification to Lemmas \ref{lemma:fejersq_coeffs} and
\ref{lemma:boundq} is that the minimum distance is reduced to $1.87\,
\lambda_c$, and that the bound in Lemma \ref{lemma:boundq} is shown to
hold starting at $0.17 \, \lambda_c$ instead of $\tC \,
\lambda_c$. For reference, we provide upper bounds on $F_{\ell}(1.87
\lambda_c, t)$ at $t \in \{0, 0.17 \, \lambda_c\}$ in Table
\ref{table:F_1p87}. As to Lemma \ref{lemma:concavity}, the only
difference is that to bound $\abs{q}$ between the origin and $0.17 \,
\lambda_c$, it is sufficient to show that the second derivative of $q$
is negative and make sure that $q>-1$. Computing \eqref{bound_qR2p}
for $\Delta=1.87 \lambda_c$ for $t \in [0, 0.17 \, \lambda_c]$, we
obtain $q''<-0.1181$. Finally, \eqref{bound_qR} yields $q>0.9113$ in
$[0, 0.17 \, \lambda_c]$.
\begin{table}[htbp ]
\begin{center}
\begin{tabular}{| c | c | c | c | c |}
	\hline
	$t/\lambda_c$ & $F_{0}\brac{1.87 \lambda_c,t}$ & $F_{1}\brac{1.87 \lambda_c,t}$ & $F_{2}\brac{1.87 \lambda_c,t}$ & $F_{3}\brac{1.87 \lambda_c,t}$ \\
	\hline
	$0$ & $ 6.708\; 10^{-3}$ & $7.978\; 10^{-2} \fc$ & $1.078\, \fc^2$ & $ 16.01\, \fc^3$ \\
	\hline
	$0.17$ & $ 6.747\; 10^{-3}$ & $0.1053 \fc$ & $1.081\, \fc^2$ & $ 41.74\, \fc^3$ \\
	\hline
\end{tabular}
\caption{Numerical upper bounds on $F_{\ell}(1.87 \lambda_c, t)$.}
  \label{table:F_1p87}
\end{center}
\end{table}

\section{Stability}
\label{sec:stability}

This section proves Theorem \ref{theorem:discrete_noisy} and we begin
by establishing a strong form of the null-space property. In the
remainder of the paper $P_T$ is the orthogonal projector onto the
linear space of vectors supported on $T$, namely, $(P_T x)_i = x_i$ if
$i \in T$ and is zero otherwise.

\begin{lemma}
\label{lemma:nullspace_prop}
Under the assumptions of Theorem \ref{theorem:discrete_noisy}, any
vector $h$ such that $F_n h = 0$ obeys 
\begin{equation}
  \normOne{P_T h} \leq \rho \normOne{P_{T^c} h}, 
\label{nullspace_strict}
\end{equation}
for some numerical constant $\rho$ obeying $0 < \rho < 1$. This
constant is of the form $1 - \rho = \alpha/\text{\em SRF}^2$ for some
positive $\alpha > 0$. If $\srf \ge 3.03$, we can take $\alpha =
0.0883$.
\end{lemma}
\begin{proof}
  Let $P_T h_t =\abs{P_T h_t} e^{i\phi_t}$ be the polar decomposition of $P_T h$,
  and consider the low-frequency polynomial $q(t)$ in Proposition
  \ref{prop:continuous_dualcert} interpolating $v_t = e^{-i\phi_t}$. We
  shall abuse notations and set $q = \{q_t\}_{t = 0}^{N-1}$ where $q_t
  = q(t/N)$.  For $t \notin T$, $|q(t/N)| = |q_t| \le \rho < 1$.  By
  construction $q = P_n q$, and thus $\<q, h\> = \<q, P_n h\> = 0$.
  Also,
\[
\<P_T q, P_T h\> = \|P_T h\|_1. 
\]
The conclusion follows from
\[
0 = \<q,h\> = \<P_T q, P_T h\> + \<P_{T^c} q, P_{T^c} h\> \ge \|P_T
h\|_1 - \|P_{T^c} q\|_\infty \|P_{T^c} h\|_1 \ge \|P_T
h\|_1 - \rho \|P_{T^c} h\|_1. 
\]

For the numerical constant, we use Lemma \ref{lemma:qbound_strict}
which says that if $0 \in T$ and $1/N \le \tC \lambda_c$, which is
about the same as $1/\srf \approx 2f_c/N \le 2 \times \tC$ or $\srf >
3.03$, we have
\[
|q(1/N)| \le 1 - 0.3533 \, (f_c/N)^2 \approx 1 - 0.0883/\srf^2 = \rho. 
\]
This applies directly to any other $t$ such that $\min_{\tau \in T}
\abs{t - \tau} = 1$. Also, for all $t$ at distance at least 2 from
$T$, Lemma \ref{lemma:qbound_strict} implies that $|q(t/N)| \le
\rho$. This completes the proof.
\end{proof}

\subsection{Proof of Theorem \ref{theorem:discrete_noisy}}

The proof is a fairly simple consequence of Lemma
\ref{lemma:nullspace_prop}. Set $h = \hat{x} - x$ and decompose the
error into its low- and high-pass components
\[
h_L = P_n h, \quad h_H = h - h_L. 
\]
The high-frequency part is in the null space of $P_n$ and
\eqref{nullspace_strict} gives
\begin{equation}
  \normOne{P_T h_H} \leq \rho \normOne{P_{T^c} h_H}.
  \label{noiseineq1_l1}
\end{equation}
For the low-frequency component we have
\begin{equation}
  \label{eq:tube}
  \|h_L\|_1 = 
  \|P_n(\hat x - x)\|_1 \le \|P_n \hat x - s\|_1 + \|s - P_n x\|_1 \le 2\delta.
\end{equation}
To bound $\|P_{T^c} h_H\|_1$, we exploit the fact that $\hat{x}$ has
minimum $\ell_1$ norm. We have
\begin{align*}
  \normOne{x} \geq \normOne{x + h} & \ge \normOne{x + h_H}
  - \|h_L\|_1 \\
  & \ge \|x\|_1 - \normOne{P_T h_H} + \normOne{P_{T^c} h_H} - \|h_L\|_1\\
  & \ge \|x\|_1 + (1-\rho) \normOne{P_{T^c} h_H} - \|h_L\|_1,
\end{align*}
where the last inequality follows from \eqref{noiseineq1_l1}. Hence, 
\[
\normOne{P_{T^c} h_H} \le \frac{1}{1-\rho} \|h_L\|_1 \quad \Rightarrow
\quad \|h_H\|_1 \le \frac{1+\rho}{1-\rho} \, \|h_L\|_1. 
\]
To conclude, 
\[
\|h\|_1 \le \|h_L\|_1 + \|h_H\|_1 \le \frac{2}{1-\rho} \|h_L\|_1 \le
\frac{4\delta}{1-\rho}, 
\]
where the last inequality follows from \eqref{eq:tube}.

Since from Lemma \ref{lemma:nullspace_prop}, we have $1-\rho =
\alpha/\srf^2$ for some numerical constant $\alpha$, the upper bound
is of the form $4\alpha^{-1} \, \srf^2 \, \delta$. For $\Delta(T) \ge
2.5\lambda_c$, we have $\alpha^{-1} \approx 11.235$.

\subsection{Sparsity is not enough}
\label{sec:slepian}

Consider the vector space $\C^{48}$ of sparse signals of length
$N=4096$ supported on a certain interval of length 48. Figure
\ref{fig:singval} shows the eigenvalues of the low-pass filter $P_n =
\frac{1}{N} F_n F_n^*$ acting on $\C^{48}$ for different values of the
super-resolution factor. For $\srf=4$, there exists a subspace of
dimension 24 such that any unit-normed signal ($\|x\|_2 = 1$)
belonging to it obeys
\[
\|P_n x\|_2 \le 2.52 \; 10^{-15} \quad \Leftrightarrow \quad
\frac{1}{\sqrt{N}} \|F_n x\|_2 \le 5.02 \; 10^{-8}.
\]
For $\srf=16$ this is true of a subspace of dimension 36, two thirds
of the total dimension. Such signals can be completely canceled out by
perturbations of norm $5.02 \; 10^{-8}$, so that even at
signal-to-noise ratios (SNR) of more than 145 dB, recovery is
impossible \textit{by any method} whatsoever.
\begin{figure}[ht]
\centering
\subfigure[]{
\includegraphics[width=4.5cm]{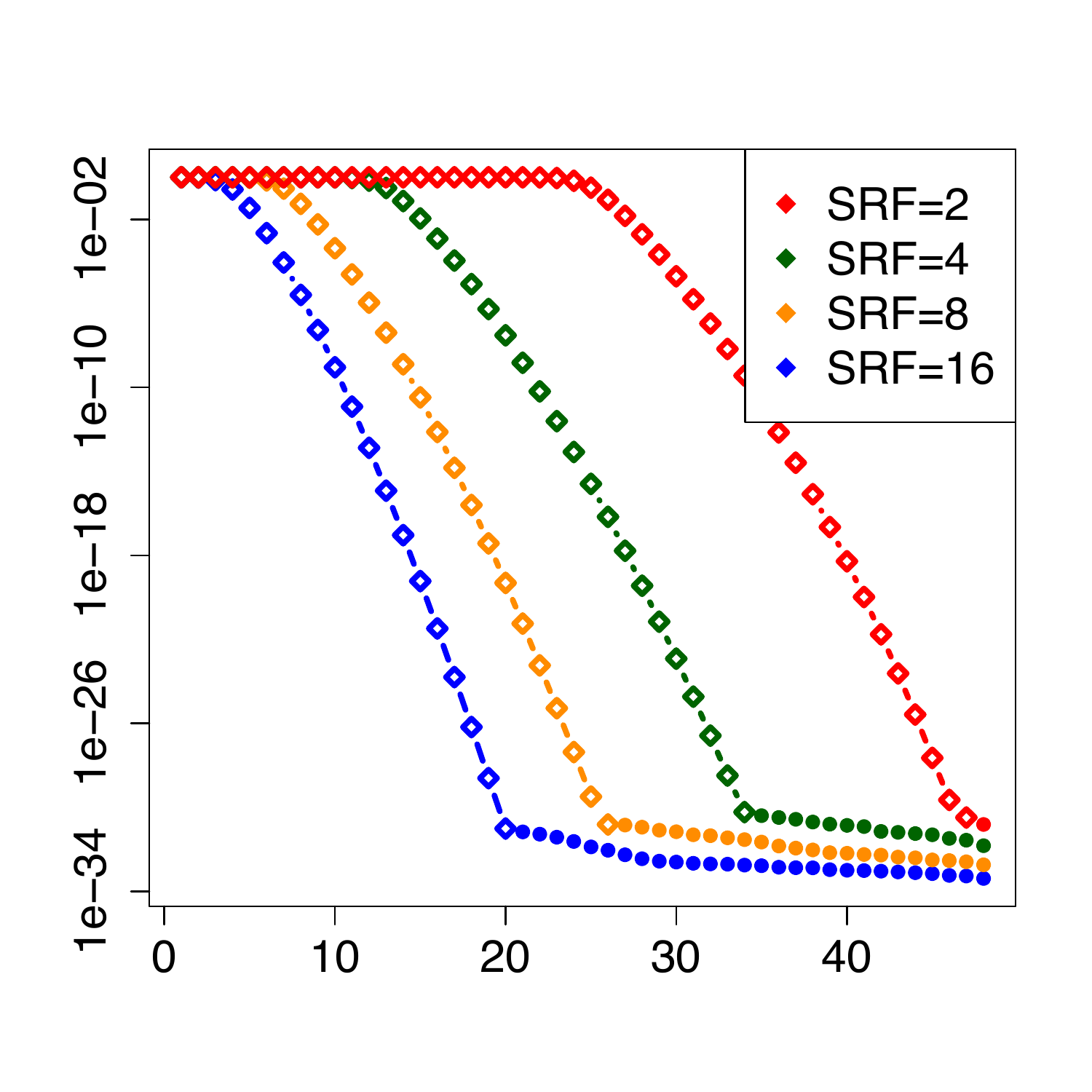}
\label{fig:singval_a}
}
\subfigure[]{
\includegraphics[width=4.5cm]{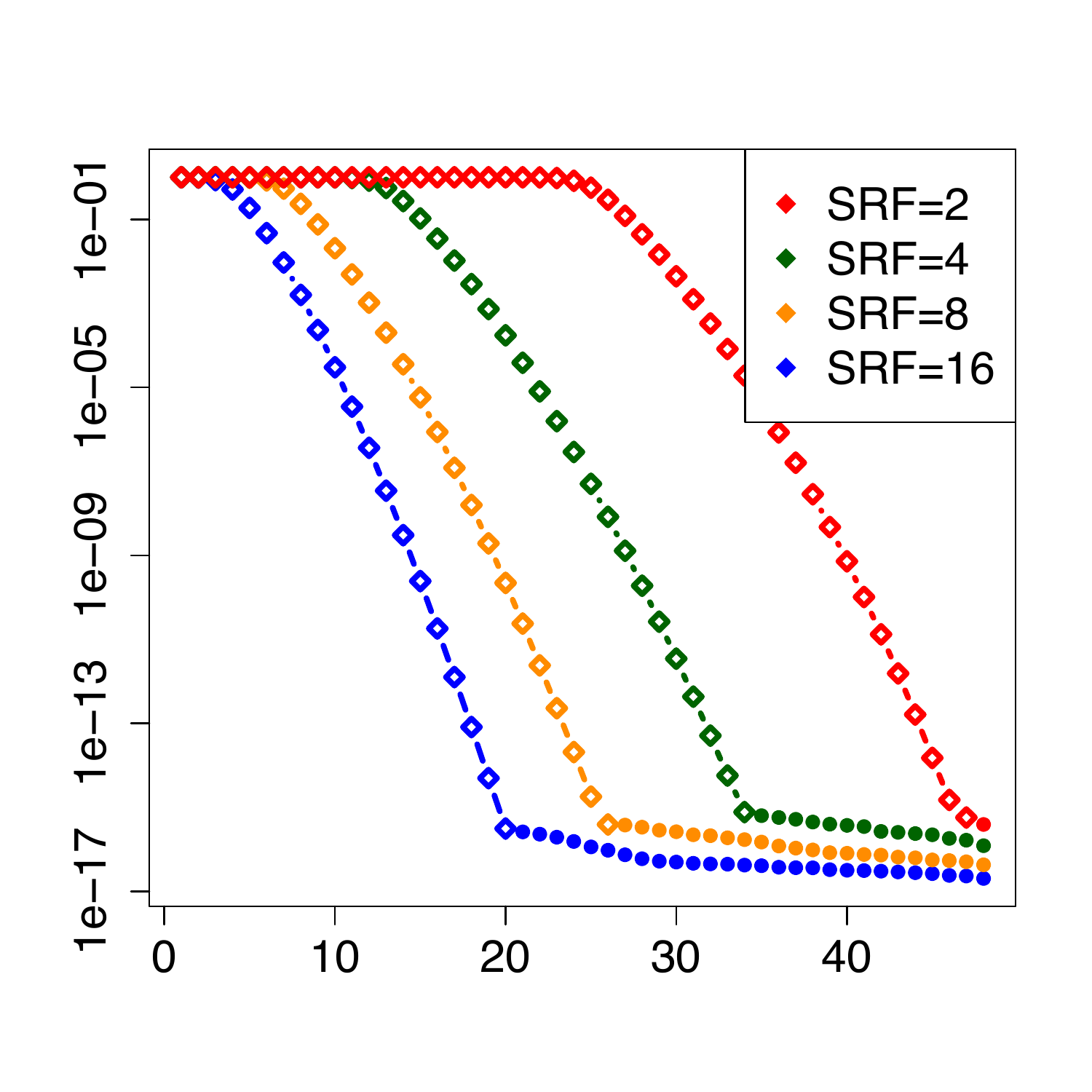}
\label{fig:singval_b}
}
\subfigure[]{
\includegraphics[width=4.5cm]{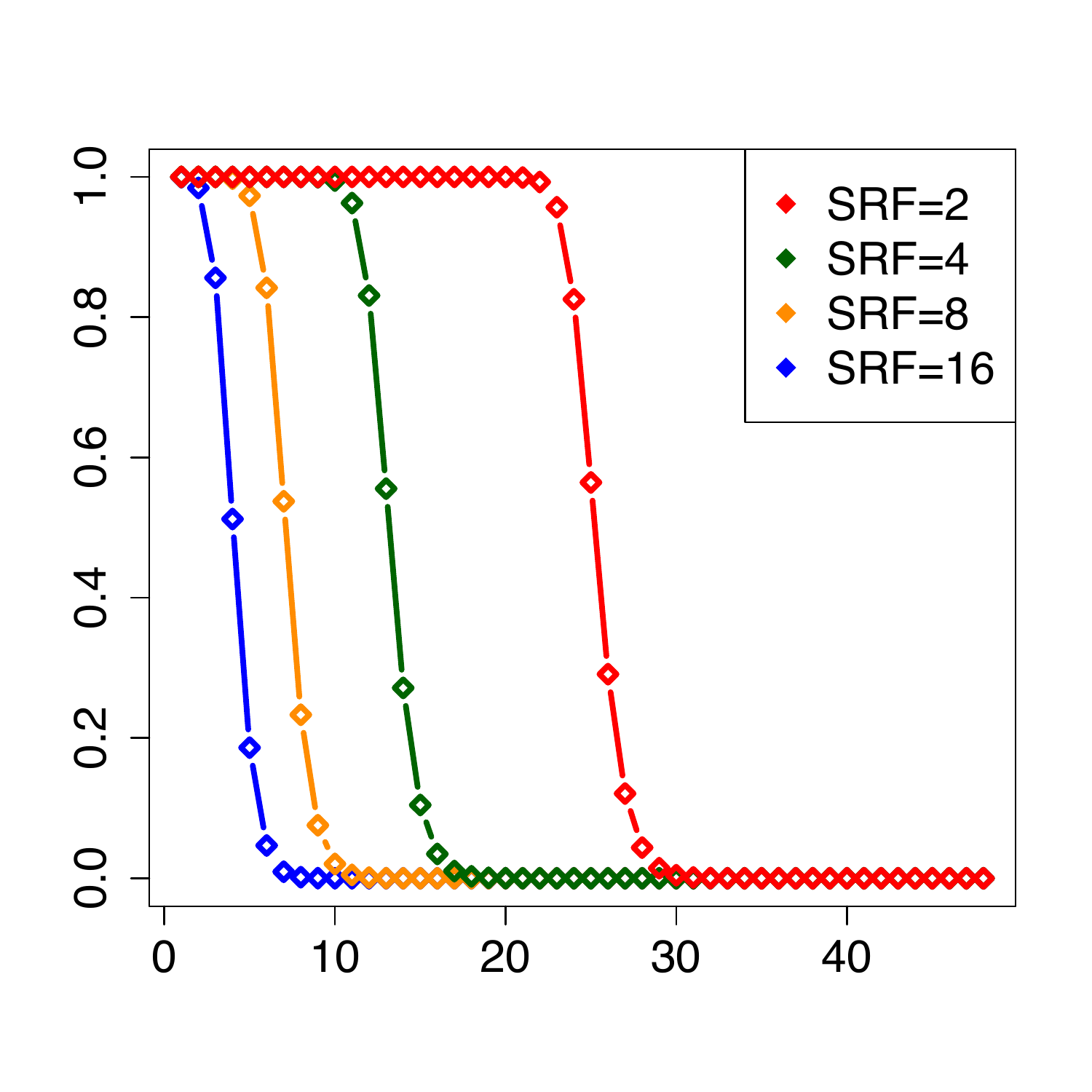}
\label{fig:singval_c}
}\caption{(a) Eigenvalues of $P_n$ acting on signals supported on a
  contiguous interval of length $48$ for super-resolution factors of
  2, 4, 8 and 16 and a signal length of $4096$. (b) Singular values of
  $\frac{1}{\sqrt{N}} F_n$ on a logarithmic scale.  (c) Same as (b)
  but on a linear scale.  Due to limited numerical precision (machine
  precision), the smallest singular values, marked with circular dots
  on the graphs, are of course not computed accurately.}
\label{fig:singval}
\end{figure}

Interestingly, the sharp transition shown in Figure \ref{fig:singval}
between the first singular values almost equal to one and the others,
which rapidly decay to zero, can be characterized asymptotically by
using the work of Slepian on prolate spheroidal sequences
\cite{slepian_discrete}. Introduce the operator $\mathcal{T}_k$, which
sets the value of an infinite sequence to zero on the complement of an
interval $T$ of length $k$.  With the notation of Section
\ref{sec:slepian_intro}, the eigenvectors of the operator $\PWn
\mathcal{T}_k$ are the discrete prolate spheroidal sequences
$\{s_j\}_{j = 1}^k$ introduced in \cite{slepian_discrete},
\begin{equation}
\label{eig_eq}
\PWn \mathcal{T}_k s_j = \lambda_j s_j,
\quad 1 > \lambda_1 \ge \lambda_2 \ge \ldots \ge \lambda_k > 0. 
\end{equation}
Set $v_j=\mathcal{T}_k s_j/ \sqrt{\lambda_j}$, then by \eqref{eig_eq},
it is not hard to see that
\[
\mathcal{T}_k \PWn v_j = \lambda_j v_j, \qquad \normTwo{v_j} = 1.
\]
In fact, the $v_j$'s are also orthogonal to each other
\cite{slepian_discrete}, and so they form an orthonormal basis of $\C^k$
(which can represent any sparse vector supported on $T$). For values
of $j$ near $k$, the value of $\lambda_j$ is about 
\[
A_j e^{-\gamma \brac{k+1}}, \quad \quad A_j =
\frac{\sqrt{\pi}2^{\frac{14\brac{k-j}+9}{4}}\alpha^{\frac{2\brac{k-j}+1}{4}}\brac{k+1}^{k-j+0.5}}{\brac{k-j}!\brac{2-\alpha}^{k-j+0.5}},
\]
where
\[
\alpha = 1 + \cos{2\pi W}, \quad \gamma = \log \brac{1 + \frac{2
    \sqrt{\alpha}}{\sqrt{2} - \sqrt{\alpha}}}.
\]
Therefore, for a fixed value of $\srf = 1/2W$, and $k \ge 20$, the
small eigenvalues are equal to zero for all practical purposes. In
particular, for $\srf=4$ and $\srf=1.05$ we obtain
\eqref{eq:eigenvalue_slepian} and \eqref{eq:eigenvalue_slepian_1p05}
in Section \ref{sec:slepian_intro} respectively. Additionally, a
Taylor series expansion of $\gamma$ for large values of $\srf$ yields
\eqref{eq:slepian_largeSRF}. 

Since $\|\PWn v_j\|_{L_2} = \sqrt{\lambda_j}$, the bound on
$\lambda_j$ for $j$ near $k$ directly implies that some sparse signals
are essentially zeroed out, even for small super-resolution
factors. However, Figure \ref{fig:singval} suggests an even stronger
statement: as the super-resolution factor increases not only
\textit{some}, but \textit{most} signals supported on $T$ seem to be
almost completely suppressed by the low pass filtering. Slepian
provides an asymptotic characterization for this phenomenon. Indeed,
just about the first $2k W$ eigenvalues of $\PWn \mathcal{T}_k$
cluster near one, whereas the rest decay abruptly towards zero. To be
concrete, for any $\epsilon>0$ and $j \geq 2 k W \brac{1+\epsilon}$,
there exist positive constants $C_0$, and $\gamma_0$ (depending on
$\epsilon$ and $W$) such that
\[
\lambda_j \leq C_0 e^{-\gamma_0 k}.
\]
This holds for all $k \ge k_0$, where $k_0$ is some fixed integer.
This implies that for any interval $T$ of length $k$, there exists a
subspace of signals supported on $T$ with dimension asymptotically
equal to $\brac{1-1/\srf}k$, which is obliterated by the measurement
process. This has two interesting consequences. First, even if the
super-resolution factor is just barely above one, asymptotically there
will always exist an irretrievable vector supported on $T$. Second, if
the super-resolution factor is two or more, most of the information
encoded in clustered sparse signals is lost.
Consider for instance a random sparse vector $x$
supported on $T$ with i.i.d.~entries. Its projection onto a fixed
subspace of dimension about $ \brac{1-1/\srf} k $ (corresponding to
the negligible eigenvalues) contains most of the energy of the signal
with high probability. However, this component is practically destroyed by
low-pass filtering. Hence, super-resolving almost any tightly
clustered sparse signal in the presence of noise is hopeless.  This
justifies the need for a minimum separation between nonzero
components.

\section{Minimization via semidefinite programming}
\label{sec:sdp}
At first sight, finding the solution to the total-variation norm
problem \eqref{TVnormMin} might seem quite challenging, as it requires
solving an optimization problem over an infinite dimensional space. It
is of course possible to approximate the solution by discretizing the
support of the signal, but this could lead to an increase in
complexity if the discretization step is reduced to improve
precision. Another possibility is to try approximating the solution by
estimating the support of the signal in an iterative fashion
\cite{radon_measures}.  Here, we take a different route and show that
\eqref{TVnormMin} can be cast as a semidefinite program with just
$\brac{n+1}^2/2$ variables, and that highly accurate solutions can be
found rather easily.  This formulation is similar to that in
\cite{atomic_norm_denoising} which concerns a related infinite
dimensional convex program. Our exposition is less formal here than in the rest of the paper.

The convex program dual to \eqref{TVnormMin} is
\begin{align}
\label{TVnormMin_dual}
\max_{c} \; \operatorname{Re} \<y, c\>  \quad \text{subject to}
\quad \normInf{\mathcal{F}_{n}^{\ast} \, c} \leq 1;
\end{align}
the constraint says that the trigonometric polynomial
$(\mathcal{F}_{n}^{\ast} \, c)(t) = \sum_{|k| \le \fc} c_k e^{i2\pi k
  t}$ has a modulus uniformly bounded by $1$ over the interval
$[0,1]$. The interior of the feasible set contains the origin and is
consequently non empty, so that strong duality holds by a generalized
Slater condition \cite{rockafellar1974conjugate}.  The cost function
involves a finite vector of dimension $n$, but the problem is still
infinite dimensional due to the constraints. A corollary to Theorem
4.24 in \cite{dumitrescu} allows to express this constraint as the
intersection between the cone of positive semidefinite matrices $\{X :
X \succeq 0\}$ and an affine hyperplane.
\begin{corollary}
\label{cor:sdp-charact}
A causal trigonometric polynomial $\sum_{k=0}^{n-1} c_k e^{i 2 \pi k
  t} $ with $c \in \C^{n}$ is bounded by one in magnitude if and only
if there exists a Hermitian matrix $Q \in \C^{n \times n}$ obeying
 \begin{equation}
\label{eq:sdp-charact}
   \MAT{Q & c \\ c^{\ast} & 1} \succeq 0, \qquad \sum_{i=1}^{n-j}Q_{i,i+j} = \begin{cases} 1, & j = 0,\\
     0, & j = 1, 2, \ldots, n-1. 
   \end{cases} 
\end{equation}
\end{corollary}
To see one direction, note that the positive semidefiniteness constraint
is equivalent to $Q-cc^{\ast} \succeq 0$. Hence, $z^{\ast} Q z \geq
\abs{c^{\ast}z}^2$ for all $z \in \C^n$. Fix $t \in [0,1]$ and set
$z_k = e^{i 2 \pi kt}$. The equality constraints in
\eqref{eq:sdp-charact} give $z^{\ast} Q z=1$ and $\abs{c^{\ast}z}^2 =
|(\mathcal{F}_n^* c)(t)|^2$ so we obtain the desired inequality constraint on the magnitude of the trigonometric polynomial $\sum_{k=0}^{n-1} c_k e^{i 2 \pi k t}$. 

Returning to \eqref{TVnormMin_dual}, the polynomial $e^{i2\pi f_c t}
\, (\mathcal{F}_{n}^{\ast} \, c)(t)$ is causal and has the same
magnitude as $(\mathcal{F}_{n}^{\ast} \, c)(t)$. Hence, the dual
problem is equivalent to 
\begin{align}
\label{TVnormMin_sdp}
\max_{c,Q} \; \operatorname{Re} \<y, c \> \quad
\text{subject to} \quad \text{\eqref{eq:sdp-charact}}.
\end{align}
To be complete, the decision variables are an Hermitian matrix $Q \in
\C^{n\times n}$ and a vector of coefficients $c \in \C^n$.  The finite
dimensional semidefinite program can be solved by off-the-shelf convex
programming software.

\begin{figure}
\centering
 \includegraphics[width=10cm]{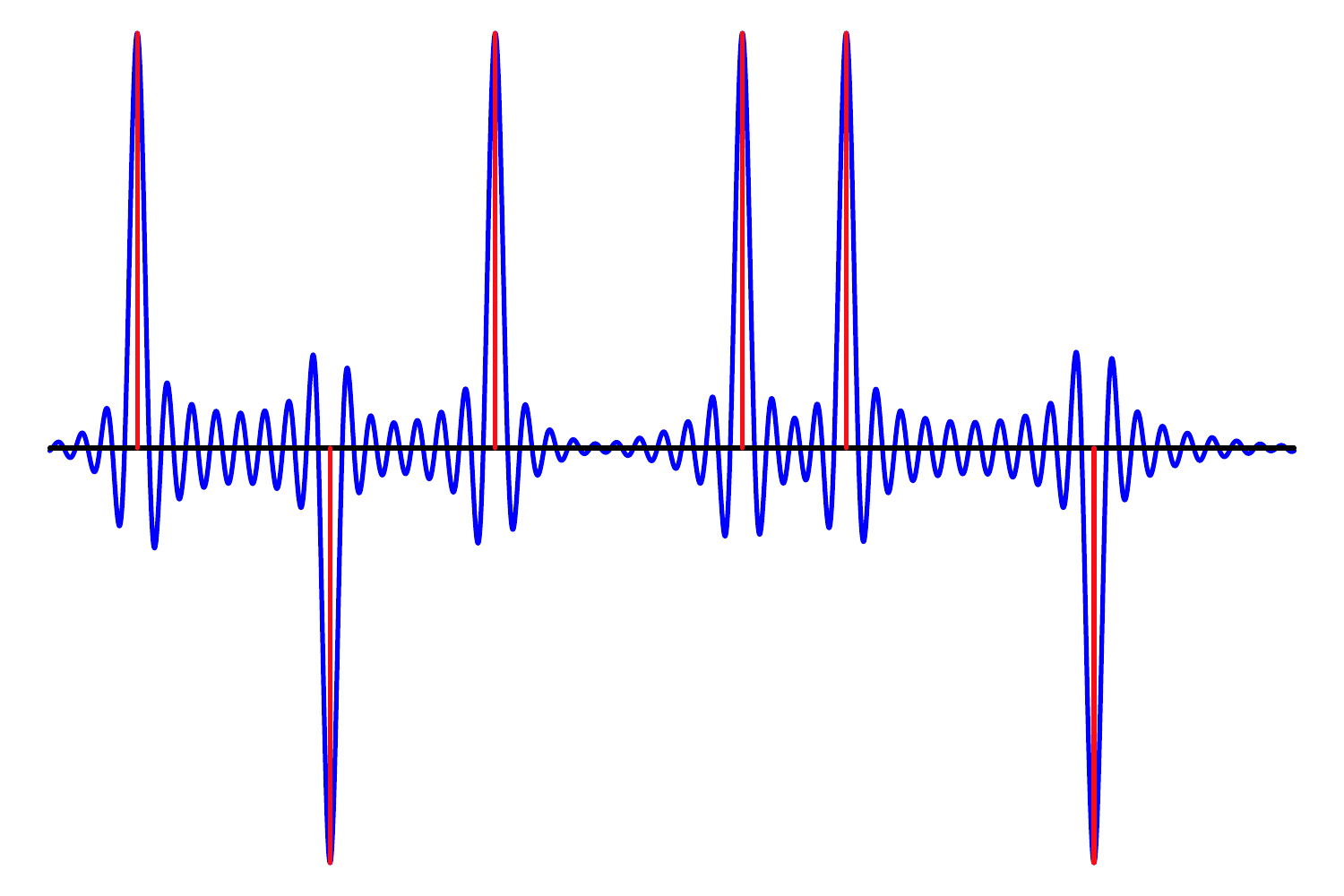}
 \caption{The sign of a real atomic measure $x$ is plotted in red. The
   trigonometric polynomial $\mathcal{F}_{n}^{\ast} c$ where $c$ is a
   solution to the dual problem \eqref{TVnormMin_sdp} is plotted in
   blue.  Note that $\mathcal{F}_{n}^{\ast} c$ interpolates the sign
   of $x$. Here, $\fc=50$ so that we have $n = 101$ low-frequency
   coefficients.} 
\label{fig:dual_pol_sdp} 
\end{figure}
The careful reader will observe that we have just shown how to compute
the optimal value of \eqref{TVnormMin}, but not how we could obtain a
solution. 
To find a primal solution, we abuse notation by letting $c$ be the
solution to \eqref{TVnormMin_sdp} and consider the trigonometric
polynomial
\begin{align}
  p_{2n-2}(e^{i2\pi t}) = 1-\abs{(\mathcal{F}_{n}^{\ast}c)(t)}^2 =
  1-\sum_{k=-2\fc}^{2\fc} u_k e^{i2\pi k t}, \qquad u_k = \sum_j
  c_{j}\bar{c}_{j-k}. \label{polynomial_P}
\end{align}
Note that $z^{2\fc}p_{2n-2}(z)$, where $z \in \C$, is a polynomial of
degree $4\fc = 2(n-1)$ with the same roots as $p_{2n-2}(z)$---besides
the trivial root $z = 0$. Hence, $p_{2n-2}(e^{i2\pi t})$ has at most
$2n-2$ roots. By construction $p_{2n-2}(e^{i2\pi t})$ is a real-valued
and nonnegative trigonometric polynomial; in particular, it cannot
have single roots on the unit circle since the existence of single
roots would imply that $p_{2n-2}(e^{i2\pi t})$ takes on negative
values. Therefore, $p_{2n-2}(e^{i2\pi t})$ is either equal to zero
everywhere or has at most $n-1$ roots on the unit circle.
By strong duality, any solution $\hat{x}$ to \eqref{TVnormMin} obeys
\begin{align*}
  \operatorname{Re}\<y, c\> = \operatorname{Re} \<\mathcal{F}_{n} \,
  \hat{x}, c\> = \operatorname{Re} \<\hat{x}, \mathcal{F}_{n}^{\ast}
  \, c\> = \operatorname{Re}\sqbr{ \int_0^1
    \overline{(\mathcal{F}_{n}^{\ast} c)\brac{t}}\, \hat{x}(\text{d}t)} =
  \normTV{\hat{x}},
\end{align*}
which implies that the trigonometric polynomial
$\mathcal{F}_{n}^{\ast} c$ is exactly equal to the sign of $\hat{x}$
when $\hat x$ is not vanishing. This is illustrated in Figure
\ref{fig:dual_pol_sdp}. Thus, to recover the support of the solution
to the primal problem, we must simply locate the roots of $p_{2n-2}$
on the unit circle, for instance by computing the eigenvalues of its
companion matrix \cite{numerical_recipes}. As shown in Table
\ref{table:sdp}, this scheme allows to recover the support with very
high precision. Having obtained the estimate for the support
$\hat{T}$, the amplitudes of the signal can be reconstructed by
solving the system of equations $\sum_{t \in \hat{T}} e^{-i2\pi k t}
a_t = y_k$, $\abs{k} \le f_c$, using the method of least
squares. There is a unique solution as we have at most $n-1$ columns
which are linearly independent since one can add columns to form a
Vandermonde system.\footnote{The set of roots contains the support of
  a primal optimal solution; if it is a strict superset, then some amplitudes will vanish.}
  Figure \ref{fig:sdp_recovery}
illustrates the accuracy of this procedure; a Matlab script
reproducing this example is available at
\url{http://www-stat.stanford.edu/~candes/superres_sdp.m}.
\begin{table}[htbp]
\begin{center}
\begin{tabular}{| c | c | c | c | c |}
	\hline
	$\fc$ &  25 & 50 &  75 &  100  \\
	\hline
	Average error   & $ 6.66  \, 10^{-9} $ & $1.70  \, 10^{-9} $ &  $5.58  \, 10^{-10} $ & $2.96  \, 10^{-10} $ \\
	\hline
	Maximum error   & $1.83  \, 10^{-7} $ & $8.14  \, 10^{-8} $ &  $2.55 \, 10^{-8}$   & $2.31  \, 10^{-8} $ \\
	\hline
\end{tabular}
\caption{Numerical recovery of the support of a sparse signal obtained
  by solving \eqref{TVnormMin_sdp} via CVX \cite{cvx}. For each value of
  the cut-off frequency $\fc$, 100 signals were generated with random
  complex amplitudes situated at approximately $\fc/4$ random
  locations in the unit interval separated by at least $2/\fc$. The
  table shows the errors in estimating the support locations. }
  \label{table:sdp}
\end{center}
\end{table}
\begin{figure}
\centering
 \includegraphics[width=8cm]{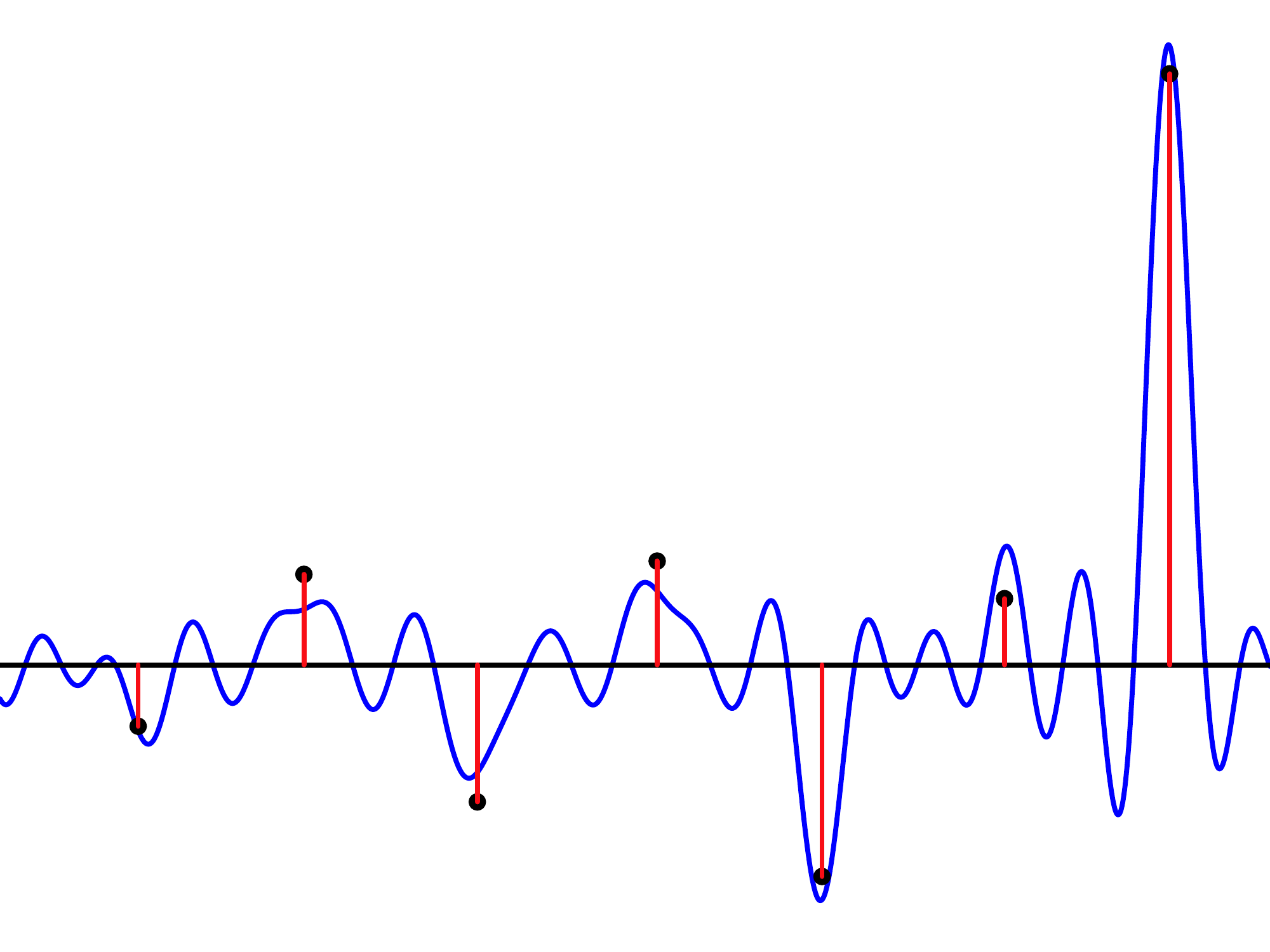}
 \caption{There are 21 spikes situated at arbitrary locations
   separated by at least $2\lambda_c$ and we observe 101 low-frequency
   coefficients ($\fc = 50$).  In the plot, seven of the original
   spikes (black dots) are shown along with the corresponding low
   resolution data (blue line) and the estimated signal (red line). 
 }
\label{fig:sdp_recovery}
\end{figure}

\begin{figure}
\centering
 \includegraphics[width=9cm,height=4cm]{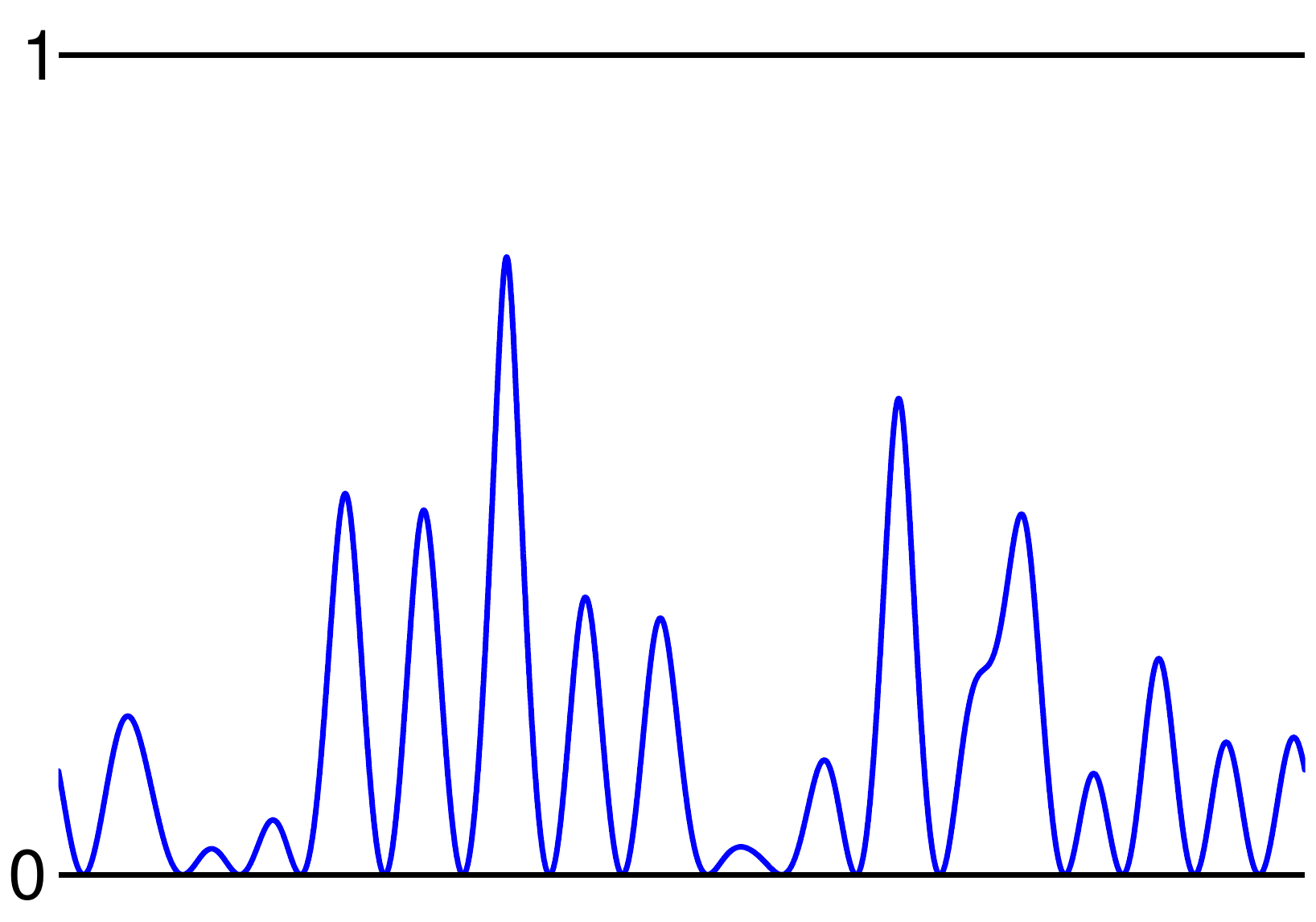}
 \caption{The trigonometric polynomial $p_{2n-2}(e^{i2\pi t})$ with
   random data $y \in \C^{21}$ ($n = 21$ and $\fc = 10$) with
   i.i.d.~complex Gaussian entries. The polynomial has 16 roots.}
\label{fig:random_y}
\end{figure}
In summary, in the usual case when $p_{2n-2}$ has less than $n$ roots
on the unit circle, we have explained how to retrieve the minimum
total-variation norm solution. It remains to address the situation in
which $p_{2n-2}$ vanishes everywhere.  In principle, this could happen
even if a primal solution to \eqref{TVnormMin} is an atomic measure
supported on a set $T$ obeying $|T| < n$. For example, let $x$ be a
positive measure satisfying the conditions of Theorem
\ref{theorem:noiseless}, which implies that it is the unique solution
to \eqref{TVnormMin}. Consider a vector $c \in \C^n$ such that
$(\mathcal{F}_{n}^{\ast} c)(t) = 1$; i.~e., the trigonometric
polynomial is constant. Then
\begin{align*}
  \operatorname{Re}\<y, c\> = \operatorname{Re} \<\mathcal{F}_{n} \,
  x, c\> = \operatorname{Re} \<x, \mathcal{F}_{n}^{\ast} \, c \> =
  \normTV{x},
\end{align*}
which shows that $c$ is a solution to the dual \eqref{TVnormMin_sdp}
that does not carry any information about the support of $x$.
Fortunately, this situation is in practice highly unusual. In fact, it
does not occur as long as
\begin{equation}
  \text{there exists a solution $\tilde c$ to \eqref{TVnormMin_sdp} obeying
  $\abs{(\mathcal{F}_{n}^{\ast} \tilde{c})(t)} < 1$ for some $t \in
  \sqbr{0,1}$}, \label{cond_c}
\end{equation}
and we use interior point methods as in SDPT3 \cite{sdpt3} to solve
\eqref{TVnormMin_sdp}.  (Our simulations use CVX which in turn calls
SDPT3.) This phenomenon is explained below.  At the moment, we would
like to remark that Condition \eqref{cond_c} is sufficient for the
primal problem \eqref{TVnormMin} to have a unique solution, and holds
except in very special cases.  To illustrate this, suppose $y$ is a
random vector, \textit{not} a measurement vector corresponding to a
sparse signal. In this case, we typically observe dual solutions as
shown in Figure \ref{fig:random_y} (non-vanishing polynomials with at
most $n-1$ roots). To be sure, we have solved 400 instances of
\eqref{TVnormMin_sdp} with different values of $\fc$ and random data
$y$. In every single case, condition \eqref{cond_c} held so that we
could construct a primal feasible solution $x$ with a duality gap
below $10^{-8}$, see Figure \ref{fig:sdp_polynomial_histogram}.  In
all instances, the support of $x$ was constructed by determining roots
of $p_{2n-2}(z)$ at a distance at most $10^{-4}$ from the unit circle.
\begin{figure}
\centering
\subfigure[]{
\includegraphics[width=5.5cm,height=3cm]{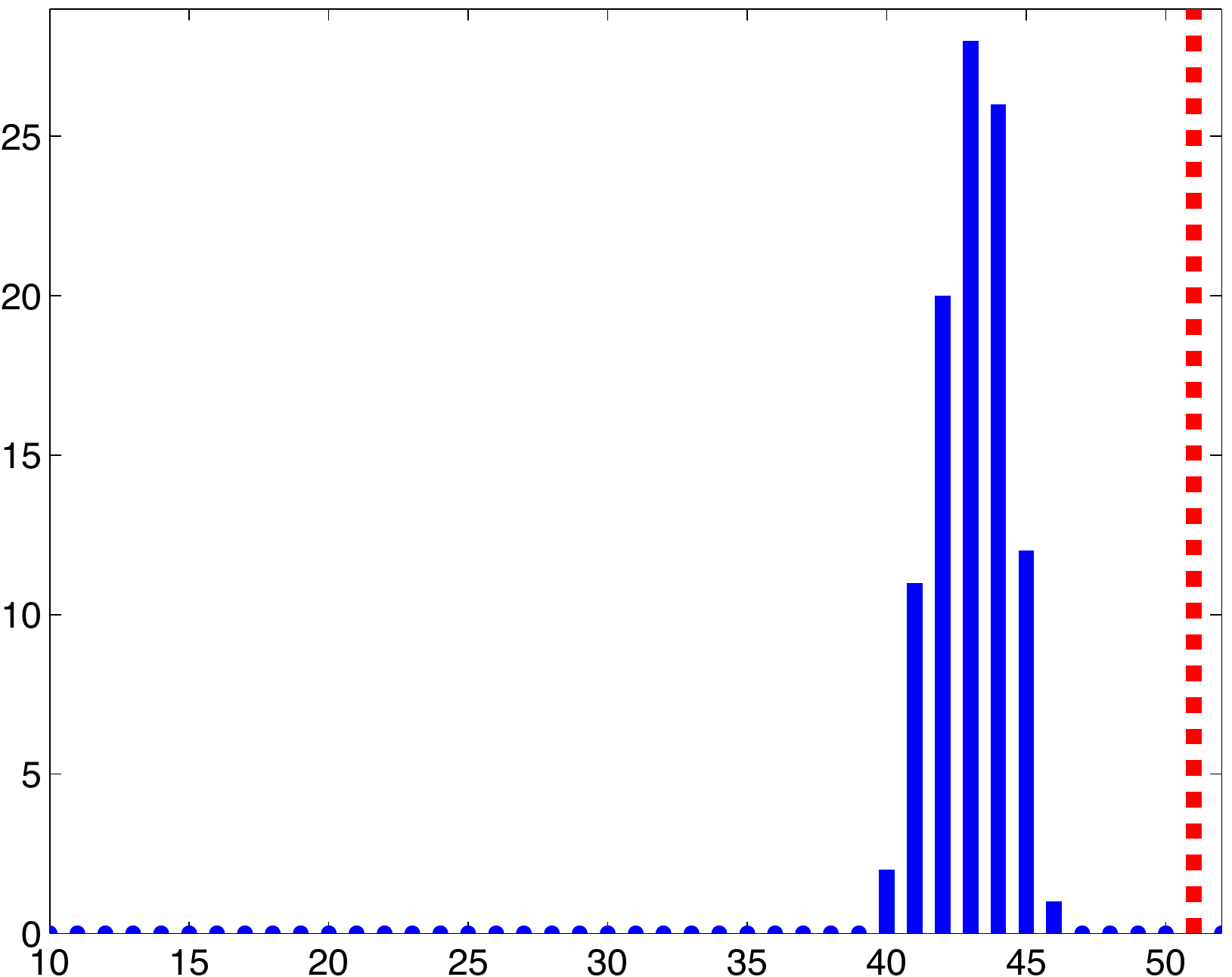}
}
\subfigure[]{
\includegraphics[width=5.5cm,height=3cm]{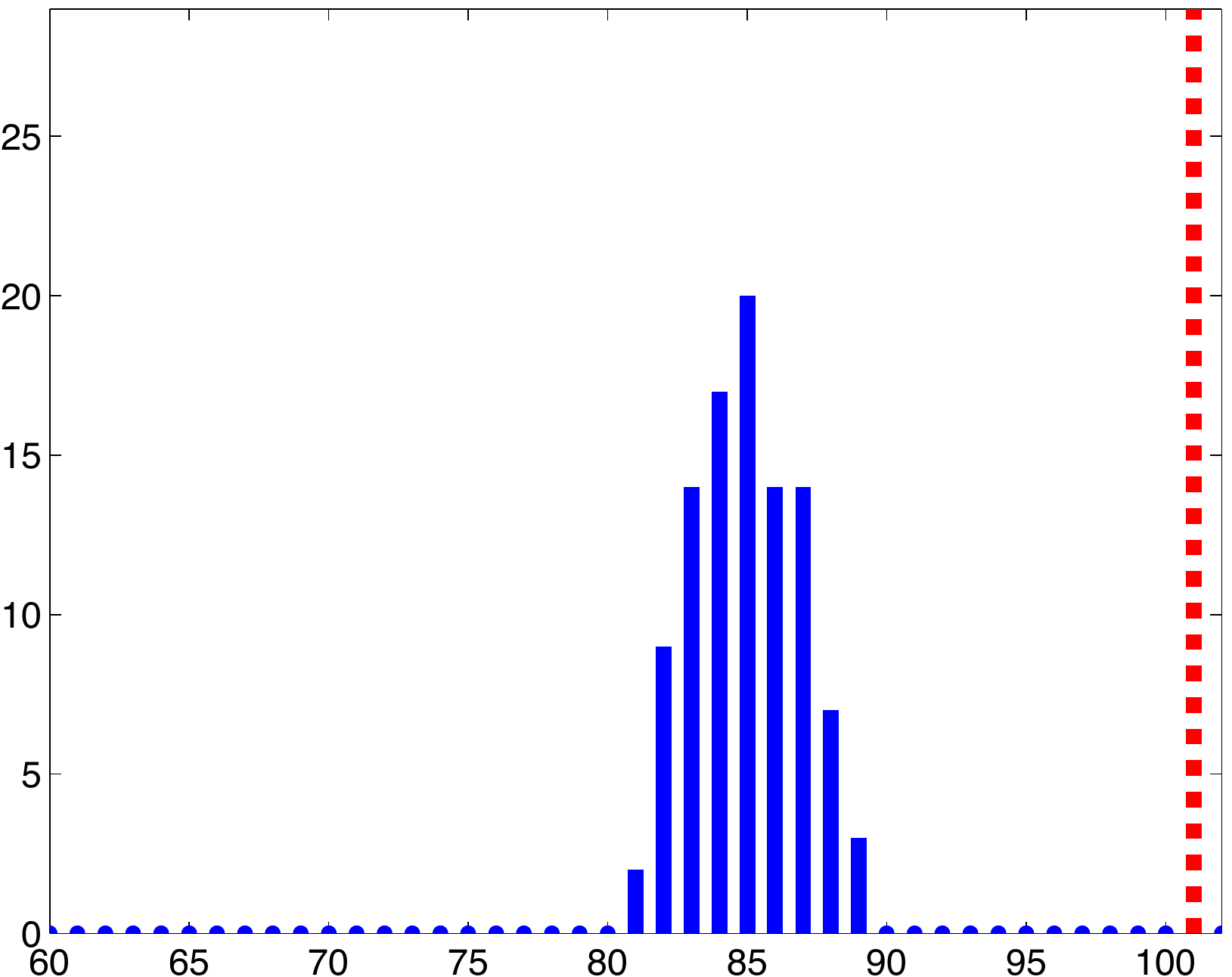}
}
\subfigure[]{
\includegraphics[width=5.5cm,height=3cm]{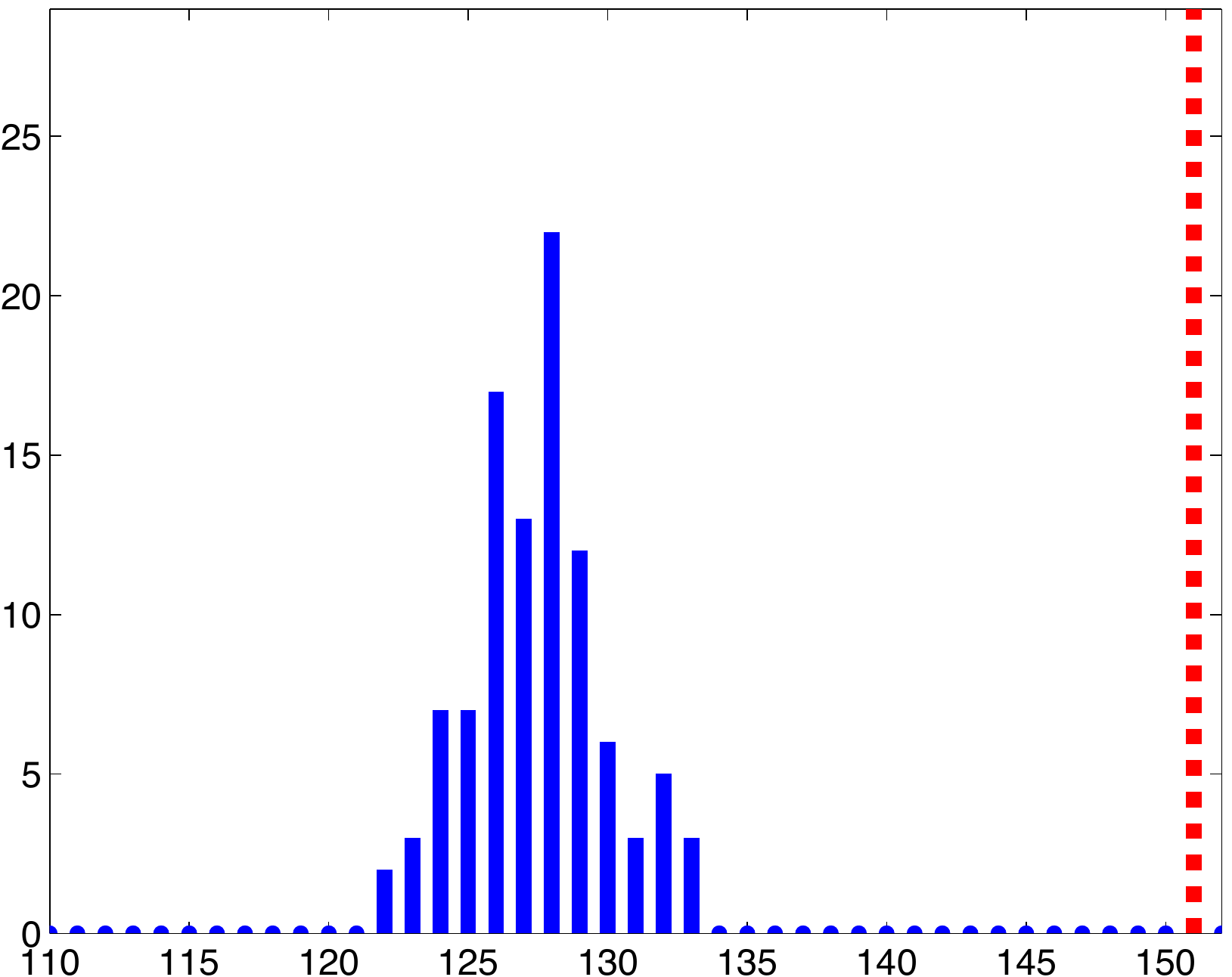}
}
\subfigure[]{
\includegraphics[width=5.5cm,height=3cm]{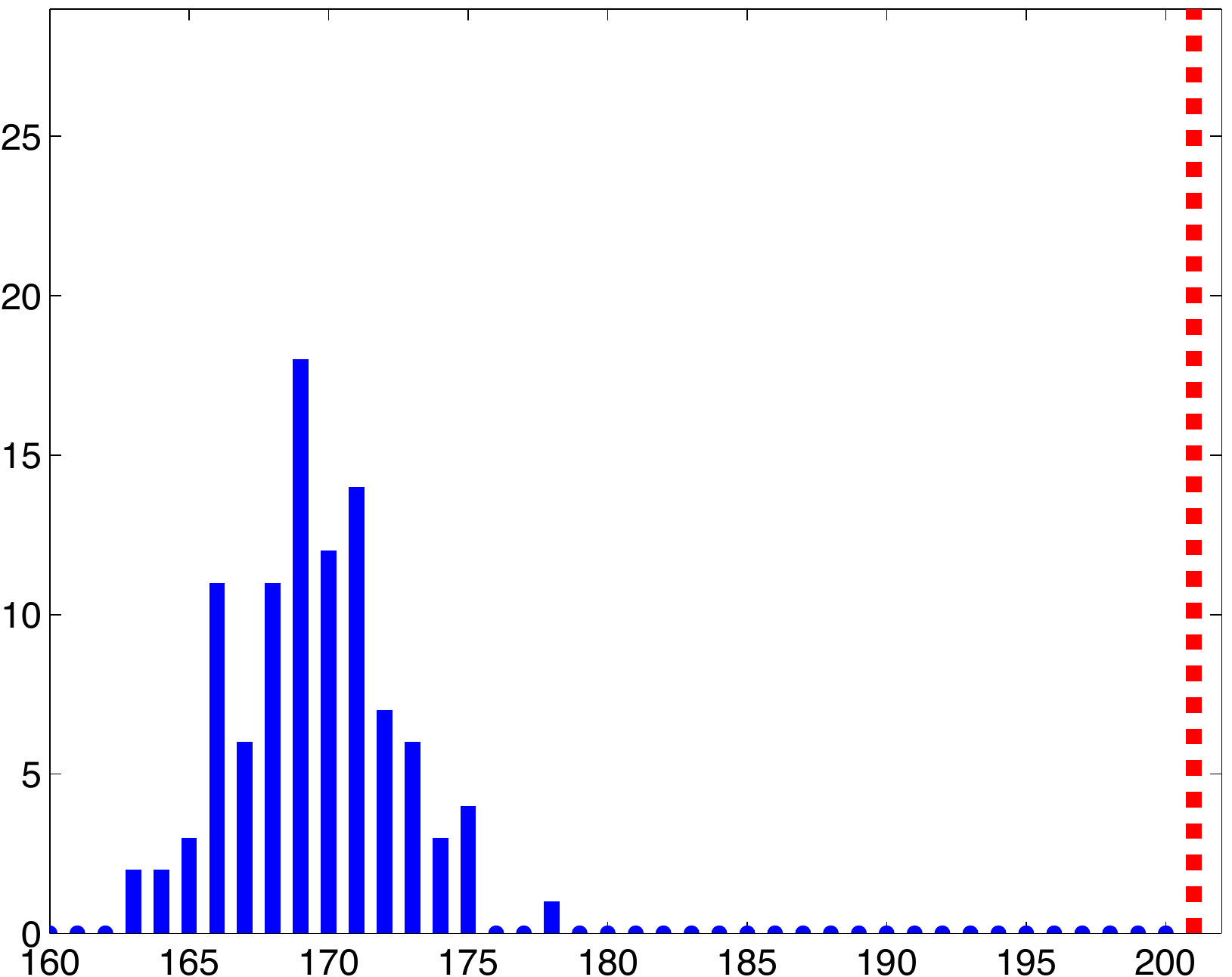}
}
\caption{Primal feasible points $x$ with a duality gap
  below $10^{-8}$ are constructed from random data $y$ sampled 
with i.i.d.~complex
  Gaussian entries. A dual gap below $10^{-8}$ implies that
  $\|x\|_{\text{TV}} - \|\hat x\|_{\text{TV}} \le 10^{-8}$, where 
  $\hat x$ is any primal optimal solution. (For reference, the optimal
  value $\|\hat x\|_{\text{TV}}$ is on the order of 10 in all cases.)
  Each figure plots the frequency of occurrence of support
  cardinalities out of 100 realizations. For example, in (a) we
  obtained a support size equal to 44 in 25 instances out of 100.
  The value of $n$ is the same in each plot and is marked by a dashed red line; 
(a) $n = 51$, (b) $n
  = 101$, (c) $n = 151$, (d) $n = 201$.
}
\label{fig:sdp_polynomial_histogram}
\end{figure}

Interior point methods approach solutions from the interior of the
feasible set by solving a sequence of optimization problems in which
an extra term, a scaled \textit{barrier function}, is added to the
cost function \cite{boyd_convex}. To be more precise, in our case
\eqref{TVnormMin_sdp} would become
 \begin{align}
\max_{c,Q} \; \operatorname{Re}\sqbr{y^{\ast} c} + t \log \det \brac{\MAT{Q & c \\ c^{\ast} & 1}}\quad
\text{subject to} \quad \text{\eqref{eq:sdp-charact}}, \label{interior_point_problem}
\end{align}
where $t$ is a positive parameter that is gradually reduced towards
zero in order to approach a solution to \eqref{TVnormMin_sdp}. Let
$\lambda_k$, $1 \le k \le n$, denote the eigenvalues of $Q -
cc^{\ast}$. By Schur's formula (Theorem 1.1 in \cite{zhangSchur}) we
have
\begin{align*}
 \log \det \brac{\MAT{Q & c \\ c^{\ast} & 1}} & = \log \det \brac{Q - cc^{\ast}} = \sum_{k=1}^{n} \log \lambda_k.
\end{align*}
Suppose condition \eqref{cond_c} holds. Then Corollary
\ref{cor:sdp-charact} states that there exists $\tilde c$ with the
property that at least one eigenvalue of $Q -
\tilde{c}\tilde{c}^{\ast}$ is bounded away from zero. This is the
reason why in the limit $t \rightarrow 0$,
\eqref{interior_point_problem} will construct a non-vanishing
polynomial $p_{2n-2}$ with at most $n-1$ roots on the unit circle
rather than the trivial solution $p_{2n-2} = 0$ since in the latter
case, all the eigenvalues of $Q - {c}{c}^{\ast}$ vanish. Hence, an
interior-point method can be said to solve the primal problem
\eqref{TVnormMin} provided \eqref{cond_c} holds.

To conclude, we have merely presented an informal discussion of a
semidefinite programming approach to the minimum-total variation
problem \eqref{TVnormMin}.  It is beyond the scope of this paper to
rigorously justify this approach---for example, one would need to
argue that the root finding procedure can be made stable, at least
under the conditions of our main theorem---and we leave this to future
work along with extensions to noisy data.

\section{Numerical experiments}
\label{sec:numerical}

To evaluate the minimum distance needed to guarantee exact recovery by
$\ell_1$ minimization of any signal in $\C^N$, for a fixed $N$, we
propose the following heuristic scheme:
\begin{itemize}
\item For a super-resolution factor $\srf = N/n$, we work with a
  partial DFT matrix $F_n$ with frequencies up to $f_c = \lfloor n/2
  \rfloor$. Fix a candidate minimum distance $\Delta$.

\item Using a greedy algorithm, construct an adversarial support with
  elements separated by at least $\Delta$ by sequentially adding
  elements to the support. Each new element is chosen to minimize the
  condition number formed by the columns corresponding to the selected
  elements.

\item Take the signal $x$ to be the singular vector corresponding to
  the smallest singular value of $F_n$ restricted to $T$.

\item Solve the $\ell_1$-minimization problem \eqref{l1problem} and
  declare that exact recovery occurs if the normalized error is below
  a threshold (in our case $10^{-4}$).
\end{itemize}
This construction of an adversarial signal was found to be better
adapted to the structure of our measurement matrix than other methods
proposed in the literature such as \cite{dossal_adversarial}. We used
this scheme and a simple binary search to determine a lower bound for
the minimum distance that guarantees exact recovery for $N=4096$,
super-resolution factors of 8, 16, 32 and 64 and support sizes equal
to 2, 5, 10, 20 and 50. {The simulations were carried out in Matlab,
  using \texttt{CVX} \cite{cvx} to solve the optimization problem.}
Figure \ref{fig:distances} shows the results, which suggest that on
the discrete grid we need at least a minimum distance equal to twice
the super-resolution factor in order to guarantee reconstruction of
the signal (red curve). Translated to the continuous setting, in which
the signal would be supported on a grid with spacing $1/N$, this
implies that $\Delta \gtrsim \lambda_c$ is a necessary condition for
exact recovery.
\begin{figure}[ht]
\centering
\includegraphics[width=10cm]{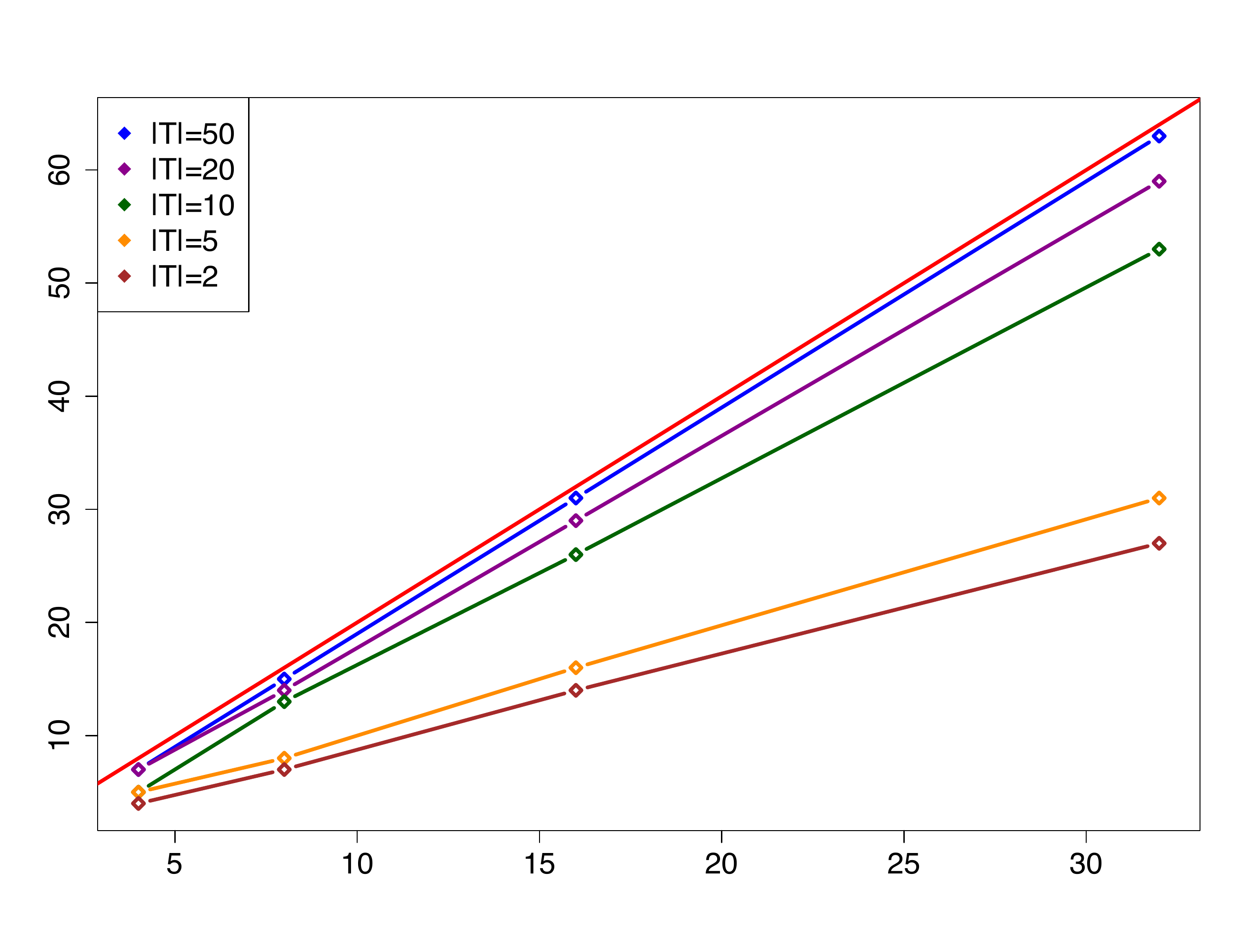}
\caption{
  Minimum distances (vertical axis) at which exact recovery by
  $\ell_1$ minimization occurs for the adversarial complex sign
  patterns against the corresponding super-resolution factors. {At the
    red curve, the minimum distance would be exactly equal to twice
    the super-resolution factor.} Signal length is $N=4096$.}
\label{fig:distances}
\end{figure}

%
\section{Discussion}
\label{sec:discussion}

In this paper, we have developed the beginning of a mathematical
theory of super-resolution. In particular, we have shown that we can
super-resolve `events' such as spikes, discontinuity points, and so on
with infinite precision from just a few low-frequency samples by
solving convenient convex programs. This holds in any dimension
provided that the distance between events is proportional to $1/f_c =
\lambda_c$, where $f_c$ is the highest observed frequency; for
instance, in one dimension, a sufficient condition is that the
distance between events is at least $2\lambda_c$. Furthermore, we have
proved that when such condition holds, stable recovery is possible
whereas super-resolution---by any method whatsoever---is in general
completely hopeless whenever events are at a distance smaller than about
$\lambda_c/2$.

\subsection{Improvement}

In one dimension, Theorem \ref{theorem:noiseless} shows that a
sufficient condition for perfect super-resolution is $\Delta(T) \ge
\optvalue \, \lambda_c$. Furthermore, the authors of this paper have a
proof showing that $\Delta(T) \ge 1.85 \, \lambda_c$ is
sufficient. This proof, however, is longer and more technical and,
therefore, not presented here. In fact, we expect that arguments more
sophisticated than those presented here would allow to lower this
value even further.  On the other hand, our numerical experiments show
that we need at least $\Delta(T) \ge \lambda_c$ as one can otherwise
find sparse signals which cannot be recovered by $\ell_1$
minimization. Hence, the minimum separation needed for success is
somewhere between $\lambda_c$ and $1.85 \, \lambda_c$. It would be
interesting to know where this critical value might be.

\subsection{Extensions}

We have focused in this paper on the super-resolution of point
sources, and by extension of discontinuity points in the function
value, or in the derivative and so on. Clearly, there are many other
models one could consider as well. For instance, we can imagine
collecting low-frequency Fourier coefficients of a function
\[
f(t) = \sum_j x_j \varphi_j(t), 
\]
where $\{\varphi_j(t)\}$ are basis functions. Again, $f$ may have lots
of high-frequency content but we are only able to observe the low-end
of the spectrum. An interesting research question is this: suppose the
coefficient sequence $x$ is sparse, then under what conditions is it
possible to super-resolve $f$ and extrapolate its spectrum accurately?
In a different direction, it would be interesting to extend our
stability results to other noise models and error metrics. We leave
this to further research.

\subsection*{Acknowledgements}
E.~C.~is partially supported by NSF via grant CCF-0963835 and the 2006
Waterman Award, by AFOSR under grant FA9550-09-1-0643 and by ONR under
grant N00014-09-1-0258. C.~F.~is supported by a Caja Madrid Fellowship
and was previously supported by a La Caixa Fellowship. E.~C.~would
like to thank Mikhail Kolobov for fruitful discussions, and Mark
Davenport, Thomas Strohmer and Vladislav Voroninski for useful
comments about an early version of the paper. C.~F.~would like to thank Armin Eftekhari for a remark on Lemma~\ref{lemma:fejersq_bounds}.

\bibliographystyle{abbrv}
\bibliography{refs_SR}

\newpage
\appendix
\section{Background on the recovery of complex measures}
\label{sec:tv}


With $\mathbb{T} = [0,1]$, the total variation of a complex measure
$\nu$ on a set $B \in \mathcal{B}\brac{\mathbb{T}}$ is defined by
\[
\abs{\nu}\brac{B} = \sup \sum_{j=1}^{\infty}\abs{\nu \brac{B_j}},
\]
where the supremum is taken over all partitions of $B$ into a finite
number of disjoint measurable subsets. The total variation $\abs{\nu}$
is a positive measure on $\mathcal{B}\brac{\mathbb{T}}$ and can be
used to define the total-variation norm on the space of complex
measures on $\mathcal{B}\brac{\mathbb{T}}$, 
\[
\normTV{\nu} = \abs{\nu}\brac{\mathbb{T}}.
\]
For further details, we refer the reader to \cite{rudin}.


\begin{proposition}
\label{prop:noiseless_recovery_proof}
Suppose that for any vector $v \in \C^{\abs{T}}$ with unit-magnitude
entries, there exists a low-frequency polynomial $q$ \eqref{eq:trig}
obeying \eqref{eq:cond_q}.  Then $x$ is the unique solution to
\eqref{TVnormMin}.
\end{proposition}
\begin{proof}
  The proof is a variation on the well-known argument for finite
  signals, and we note that a proof for continuous-time signals,
  similar to that below, can be found in \cite{supportPursuit}. Let
  $\hat x$ be a solution to \eqref{TVnormMin} and set $\hat x = x +
  h$. Consider the Lebesgue decomposition of $h$ relative to
  $\abs{x}$,
\[
 h  = h_T + h_{T^c} \text{,}
\]
where (1) $h_T$ and $h_{T^c}$ is a unique pair of complex measures on
$\mathcal{B}\brac{\mathbb{T}}$ such that $h_T$ is absolutely
continuous with respect to $\abs{x}$, and (2) $h_{T^c}$ and $\abs{x}$
are mutually singular. It follows that $h_T$ is concentrated on $T$
while $h_{T^c}$ is concentrated on $T^c$. Invoking a corollary of the
Radon-Nykodim Theorem (see Theorem 6.12 in \cite{rudin}), it is
possible to perform a polar decomposition of $h_T$:
\begin{align*}
 h_T = e^{i 2 \pi \phi \brac{t}}\abs{h_T}\text{,}
\end{align*}
such that $\phi\brac{t}$ is a real function defined on
$\mathbb{T}$. Assume that there exists $q(t) = \sum_{j=-\fc}^{\fc} a_j
e^{-i2\pi jt}$ obeying 
\begin{equation}
\label{eq:qApp}
\begin{cases} q\brac{t_j}  = e^{-i 2 \pi \phi \brac{t_j}}, & \forall t_j \in T \\
  \abs{q\brac{t}}  < 1, & \forall t \in [0,1]\setminus T.
\end{cases}
\end{equation}
The existence of $q$ suffices to establish a valuable inequality
between the total-variation norms of $h_T$ and $h_{T^c}$. Begin with
\[
0 = \int_{\mathbb{T}}q(t) h(\text{d}t) = \int_{\mathbb{T}}q(t) h_T(\text{d}t)
+\int_{\mathbb{T}}q(t) h_{T^c}(\text{d}t) = \normTV{h_T} + \int_{\mathbb{T}}q(t)
h_{T^c}(\text{d}t)
\]
and observe that 
\[
\abs{\int_{\mathbb{T}}q(t) h_{T^c}(\text{d}t)} < \normTV{h_{T^c}}
\]
provided $h_{T^c} \neq 0$. This gives
\[
\normTV{h_T} \le \normTV{h_{T^c}}
\]
with a strict inequality if $h \neq 0$. Assuming $h \neq 0$, we have
\[
\normTV{x} \ge \normTV{x + h} = \normTV{x + h_T} + \normTV{h_{T^c}}
\geq \normTV{x} - \normTV{h_T} + \normTV{h_{T^c}} > \normTV{x}.
\]
This is a contradiction and thus $h = 0$. In other words, $x$ is the
unique minimizer. 
\end{proof}

\section{Proof of Lemma \ref{lemma:fejersq_bounds}}
\label{proof:fejersq_bounds}
The first inequality in the lemma holds due to two lower bounds on the sine function:
\begin{align}
  \abs{\sin \brac{\pi t}} & \geq 2\abs{t}, \text{ for all } t \in
  [-1/2,1/2]\label{bound_sinx}\\
  \sin \brac{\pi t} & \geq \pi t - \frac{\pi^3t^3}{6} = \frac{2
    t}{a(t)}, \text{ for all } t \geq 0\label{sinelowerbound}.
\end{align}
The proof for these expressions, which we omit, is based on concavity of the sine function and on a Taylor expansion around the origin.
Put $f = \fc/2+1$ for short. Some simple calculations give $K'(0) = 0$
and for $t \neq 0$,
\begin{equation}
K'\brac{t} =
    4\pi \brac{\frac{\sin \brac{f \pi t}}{f\sin
        \brac{\pi t}}}^{3}\brac{\frac{\cos \brac{f \pi
          t}}{\sin \brac{\pi t}} - \frac{\sin \brac{f \pi
          t}\cos \brac{\pi t} }{f \sin^2 \brac{\pi t}} }.
    \label{Kprime}
  \end{equation}
Further calculations show that the value of the second derivative of
$K$ at the origin is $-\pi^2 \fc \brac{\fc+4}/3$, and for $t \neq 0$,
\begin{multline}
  K''\brac{t} =  \frac{4\pi^2 \sin ^{2} \brac{f \pi t} }{f^{2} \sin^4 \brac{\pi t} } \left[ 3\brac{\cos \brac{f \pi t} - \frac{\sin \brac{f \pi t}\cos \brac{\pi t} }{f \sin \brac{\pi t}} }^2 \right. \\ \left.
    -\sin^2 \brac{f \pi t} -\frac{\sin \brac{2f \pi t}}{f\tan
      \brac{\pi t}} + \frac{\sin^2 \brac{f \pi t}}{f^2 \tan^2
      \brac{\pi t}} + \frac{\sin^2 \brac{f \pi t}}{f^2 \sin^2
      \brac{\pi t}}\right].\label{Kprime2}
\end{multline}
It is also possible to check that the third derivative of $K$ is zero
at the origin, and for $t \neq 0$,
\begin{equation}
  K'''\brac{t}  =  \frac{ 4 \pi^3 \sin \brac{f \pi t}}{f\sin^{4} 
    \brac{\pi t}} \, \left( 6H_1(t) + 9 \sin\brac{f \pi t} H_2(t) + \sin^2\brac{f \pi t}H_3(t)\right) 
\label{Kprime3}
\end{equation} 
with
\begin{align*}
  H_1(t) & =  \brac{\cos \brac{f \pi t} - \frac{\sin \brac{f \pi t}\cos \brac{\pi t} }{f \sin \brac{\pi t}} }^3 \\
  H_2(t) & = \brac{\cos \brac{f \pi t} - \frac{\sin \brac{f \pi t}\cos
      \brac{\pi t} }{f \sin \brac{\pi t}} } \left( -\sin \brac{f \pi
      t} -\frac{2\cos \brac{f \pi t}}{f\tan \brac{\pi t}}+ \frac{\sin
      \brac{f \pi t}}{f^2 \tan^2 \brac{\pi t}} + \frac{\sin \brac{f
        \pi t}}{f^2 \sin^2 \brac{\pi
        t}}\right)\\
  H_3(t) & = \left(\frac{3\cos \brac{f \pi t}\brac{1+\cos^2 \brac{\pi
          t}}}{f^2\sin^2 \brac{\pi t}} -\cos \brac{f \pi t} + \frac{3
      \sin \brac{f \pi t}}{f\tan \brac{\pi t}} - \frac{\sin \brac{f
        \pi t} \brac{1+5\cos \brac{\pi t}}}{f^3\sin^3 \brac{\pi
        t}}\right).
\label{Kprime3}
\end{align*} 
The remaining inequalities in the lemma are all almost direct consequences of plugging \eqref{bound_sinx} and \eqref{sinelowerbound} into \eqref{Kprime}, \eqref{Kprime2} and \eqref{Kprime3}. The bounds are nonincreasing in $t$ because the derivative of $b(t)$ is negative between zero and $\sqrt{2}/\pi$ and one can check that $H_\ell(\sqrt{2}/\pi) < H_\ell^\infty$ for $\fc \geq 128$.
Additionally, $b^k(t)$ is strictly convex for positive $t$ and $k\in \keys{1,2,3}$, so the derivative with respect to $\tau$ of $b^k(\Delta-\tau)+b^k(\Delta+\tau)$ is positive for $0 \leq \tau <\Delta/2$, which implies that
$\tilde{B}_\ell(\Delta-\tau) + \tilde{B}_\ell(\Delta+\tau)$ is increasing in $\tau$.

\newpage

\section{Proof of Theorem \ref{theorem:2D}}
\label{sec:proof_2D}
Theorem \ref{theorem:2D} follows from Proposition
\ref{prop:2D_dualcert} below, which guarantees the existence of a dual
certificate. In this section, we write $\Delta = \Delta(T) \ge
\Deltamin = \optvaluetwoD \, \lambda_c$. Unless specified otherwise,
$|r - r'|$ is the $\infty$ distance.

\begin{proposition}
 \label{prop:2D_dualcert}
 Let $T=\keys{r_1,r_2,\dots} \subset \mathbb{T}^2$ be any family of
 points obeying the minimum distance condition
\begin{equation*}
  \abs{r_j-r_k} \geq \optvaluetwoD \, \lambda_c, \quad r_j \neq r_k \in T.
\end{equation*}
Assume $\fc\geq \minmTwoD$. Then for any vector $v \in \R^{\abs{T}}$
with $\abs{v_j} = 1$, there exists a trigonometric polynomial $q$ with
Fourier series coefficients supported on $\{-f_c, -\fc + 1, \ldots,
f_c\}^2$ with the property
\begin{equation}
  \label{eq:cond_q_2D}
  \begin{cases} q(r_j) = v_j, & t_j \in T,\\
    |q(r)| < 1, & t \in \mathbb{T}^2 \setminus T. 
\end{cases}
\end{equation}
\end{proposition}
The proof is similar to that of Proposition
\ref{prop:continuous_dualcert} in that we shall construct the dual
polynomial $q$ by interpolation with a low-pass, yet rapidly decaying
two-dimensional kernel. Here, we consider
\begin{align*}
\Ktwo\brac{r} = K\brac{x}K\brac{y} \text{,}
\end{align*}
obtained by tensorizing the square of the Fejer kernel
\eqref{def:kernel}. (For reference, if we had data in which $y(k)$ is
observed if $\|k\|_2 \le f_c$, we would probably use a radial kernel.)
Just as before, we have fixed $K$ somewhat arbitrarily, and it would
probably be possible to optimize this choice to improve the constant
factor in the expression for the minimum distance.  We interpolate the
sign pattern using $\Ktwo$ and its partial derivatives, denoted by
$\Ktwo_{\brac{1,0}}$ and $\Ktwo_{\brac{0,1}}$ respectively, as
follows:
\[
q\brac{r} = \sum_{r_j \in T} \bigl[ \alpha_j \Ktwo\brac{r-r_j} + \beta_{1j}
\Ktwo_{\brac{1,0}}\brac{r-r_j}+\beta_{2j} \Ktwo_{\brac{0,1}}\brac{r-r_j}\bigr], 
\]
and we fit the coefficients so that for all $t_j \in T$,
\begin{equation}
  \label{eq:cond2_q_2D}
\begin{aligned}
  q\brac{t_j} & = v_j, \\
\nabla q\brac{t_j} & = 0.  
\end{aligned}
\end{equation}
The first intermediate result shows that the dual polynomial is well
defined, and also controls the magnitude of the interpolation
coefficients.
\begin{lemma}
\label{lemma:coeff_bounds_2D}
Under the hypotheses of Proposition \ref{prop:2D_dualcert}, there are
vectors $\alpha$, $\beta_1$ and $\beta_2$ obeying \eqref{eq:cond2_q_2D} and 
\begin{equation}
  \label{bound_alpha_2D}
\begin{aligned}
  \normInf{\alpha}  & \leq 1 + 5.577 \; 10^{-2},\\
  \normInf{\beta} & \leq 2.930 \;
  10^{-2}\,\lambda_c, 
\end{aligned}
\end{equation}
where $\beta = (\beta_1, \beta_2)$. Further, if $v_1=1$,
\begin{align}
  \alpha_1 & \geq 1- 5.577 \; 10^{-2} \label{bound_alpha1_2D}\text{.}
\end{align} 
\end{lemma}

Proposition \ref{prop:2D_dualcert} is now a consequence of the two
lemmas below which control the size of $q$ near a point $r_0 \in
T$. Without loss of generality, we can take $r_0 = 0$.
\begin{lemma}
\label{lemma:q_bound_2D_conc}
Assume $0 \in T$. Then under the hypotheses of Proposition
\ref{prop:2D_dualcert}, $\abs{q\brac{r}}<1$ for all $0 < |r| \leq
\rOne \, \lambda_c$.
\end{lemma}
\begin{lemma}
\label{lemma:q_bound_2D}
Assume $0 \in T$.  Then under the conditions of Proposition
\ref{prop:2D_dualcert}, $\abs{q\brac{r}}<1$ for all $r$ obeying $\rOne
\, \lambda_c \le \abs{r} \leq \Delta/2$. This also holds for all $r$
that are closer to $0 \in T$ (in the $\infty$ distance) than to any
other element in $T$.
\end{lemma}

\subsection{Proof of Lemma \ref{lemma:coeff_bounds_2D}}
\label{subsec:coeff_bounds_2D}

To express the interpolation constraints in matrix form, define
\begin{align*}
\brac{D_0}_{jk} & = \Ktwo \brac{r_j-r_k}, & \qquad 
\brac{D_{\brac{1,0}}}_{jk} & = \Ktwo_{\brac{1,0}}\brac{r_j-r_k}, & \qquad 
\brac{D_{\brac{0,1}}}_{jk} & = \Ktwo_{\brac{0,1}}\brac{r_j-r_k},\\
\brac{D_{\brac{1,1}}}_{jk} & =  \Ktwo_{\brac{1,1}}\brac{r_j-r_k}, & \qquad 
\brac{D_{\brac{2,0}}}_{jk} & = \Ktwo_{\brac{2,0}}\brac{r_j-r_k},  & \qquad 
\brac{D_{\brac{0,2}}}_{jk} & = \Ktwo_{\brac{0,2}}\brac{r_j-r_k} \text{.}
\end{align*}
To be clear, $\Ktwo_{(\ell_1,\ell_2)}$ means that we are taking
$\ell_1$ and $\ell_2$ derivatives with respect to the first and second
variables. Note that $D_0$, $D_{\brac{2,0}}$, $D_{\brac{1,1}}$ and
$D_{\brac{0,2}}$ are symmetric, while $D_{\brac{1,0}}$ and
$D_{\brac{0,1}}$ are antisymmetric, because $K$ and $K''$ are even
while $K'$ is odd. The interpolation coefficients are solutions to
\begin{align}
\MAT{D_0 & D_{\brac{1,0}} & D_{\brac{0,1}}\\ D_{\brac{1,0}} & D_{\brac{2,0}} & D_{\brac{1,1}} \\ D_{\brac{0,1}} & D_{\brac{1,1}} & D_{\brac{0,2}}} \MAT{\alpha \\ \beta_1 \\ \beta_2} & = \MAT{v\\0\\0} \quad \Leftrightarrow \quad 
\MAT{D_0 & -\tilde{D}_1^T \\ \tilde{D}_1 & \tilde{D}_2}\MAT{\alpha \\ \beta} =\MAT{v\\0}\text{,}\label{equation2D}
\end{align}
where we have defined two new matrices $\tilde{D}_1$ and
$\tilde{D}_2$. The norm of these matrices can be bounded by leveraging
1D results.
For instance, consider
\[
\normInfInf{\Id- D_0} = \sum_{r_j \in T\setminus \{0\}} \abs{\Ktwo
  \brac{r_j}} .
\]
We split this sum into different regions corresponding to whether
$|x_j|$ or $|y_j| \le \Delta/2$ and to $\min(|x_j|, |y_j|) \ge
\Delta/2$. First,
\[
\sum_{r_j \neq 0 \, : \, |y_j| < \Delta/2} \abs{\Ktwo
  \brac{r_j}} \le \sum_{r_j \neq 0 \, : \, |y_j| < \Delta/2}
B_{0}\brac{x_j} \le 2 \sum_{j \ge 1}
B_0(j\Delta).
\]
This holds because the $x_j$'s must be at least $\Delta$ apart,
$B_0$ is nonincreasing and the absolute value of $\Ktwo$ is bounded by one. The region $\{r_j \neq 0, \, |x_j| <
\Delta/2\}$ yields the same bound. Now observe that Lemma
\ref{lem:simple} below combined with Lemma \ref{lemma:fejersq_bounds} gives
\[
\sum_{r_j \neq 0 \, : \, \min(x_j, y_j) \ge \Delta/2} \abs{\Ktwo
  \brac{r_j}} \le \sum_{r_j \neq 0 \, : \, \min(x_j, y_j) \ge \Delta/2} B_{0}\brac{x_j}
B_0\brac{y_j} \le \Bigl[\sum_{j_1 \ge 0} B_0(\Delta/2 + j_1 \Delta)\Bigr]^2. 
\]
To bound this expression, we apply the exact same technique as for
\eqref{sumB_formula} in Section \ref{subsec:proofs_lemmas}, starting at
$j=0$ and setting $j_0=20$. This gives
\begin{equation}
\label{Id_F_bound_2D}
\normInfInf{\Id- D_0} \le 4 \sum_{j \ge 1} B_0(j\Delta) + 4
\Bigl[\sum_{j \ge 0} B_0(\Delta/2 + j \Delta)\Bigr]^2 \le 4.854 \; 10^{-2}.
\end{equation}

\begin{lemma}
  \label{lem:simple}
  Suppose $x \in \R_+^2$ and $f(x) = f_1(x_1) f_2(x_2)$ where both
  $f_1$ and $f_2$ are nonincreasing. Consider any collection of points
  $\{x_j\} \subset \R^2_+$ for which $|x_i - x_j| \ge 1$. Then 
\[
\sum_j f(x_j) \le \sum_{j_1 \ge 0} f_1(j_1) \sum_{j_2 \ge 0} f_2(j_2).
\]
\end{lemma}
\begin{proof}
  Consider the mapping $x \in \R^2_+ \mapsto (\lfloor x_1 \rfloor,
  \lfloor x_2 \rfloor)$. This mapping is injective over our family of
  points. (Indeed, two points cannot be mapped to the same pair of
  integers $(j_1, j_2)$ as it would otherwise imply that they are both
  inside the square $[j_1 + 1) \times [j_2 + 1)$, hence violating the
  separation condition.) Therefore, the monotonicity assumption gives
\[
\sum_j f(x_j) \le \sum_{j} f_1(\lfloor x_{j,1} \rfloor) f_2(\lfloor
x_{j,2} \rfloor) \le \sum_{j_1, j_2 \ge 0}  f_1(j_1)
f_2(j_2),
\]
which proves the claim.
\end{proof}

Applying the same reasoning, we obtain
\[
\normInfInf{D_{\brac{1,0}}} \le 2 \sum_{j \ge 1} B_1(j\Delta) + 2\|K'\|_\infty
\sum_{j \ge 1} B_0(j\Delta) + 4 \Bigl[\sum_{j \ge 0} B_0(\Delta/2 + j
\Delta)\Bigr] \Bigl[\sum_{j \ge 0} B_1(\Delta/2 + j \Delta)\Bigr]. 
\]
In turn, the same upper-bounding technique yields
\begin{equation}
  \label{Dx_bound}
\normInfInf{D_{\brac{1,0}}} \le 7.723 \; 10^{-2}\, \fc, 
\end{equation}
where we have used the fact that $\|K'\|_\infty \le 2.08 \brac{\fc+2}$, which follows from combining Lemma \ref{lemma:fejersq_bounds}
with \eqref{boundKp_near}. Likewise,
\begin{equation}
 \label{Dxy_bound}
\normInfInf{D_{\brac{1,1}}}  \leq 4 \|K'\|_\infty \sum_{j \ge 1} B_1(j\Delta)
+ 4 \Bigl[\sum_{j \ge 0} B_1(\Delta/2 + j \Delta)\Bigr]^2 \le 0.1576 \, \fc^2, 
\end{equation}
and finally, 
\begin{multline}
\label{Id_Dxx_bound}
  \normInfInf{\abs{\Ktwo_{\brac{2,0}}\brac{0}}\Id-D_{\brac{2,0}}} \leq 2 \sum_{j
    \ge 1} B_2(j\Delta) + 2\|K''\|_\infty \sum_{j \ge 1} B_0(j\Delta)
\\  + 4 \Bigl[\sum_{j \ge 0} B_0(\Delta/2 + j \Delta)\Bigr]
  \Bigl[\sum_{j \ge 0} B_2(\Delta/2 + j \Delta)\Bigr] \le 0.3539 \, \fc^2,
\end{multline}
since $\|K''\|_\infty=\pi^2\fc \brac{\fc+4}/3$, as $\abs{K''}$ reaches its global maximum at the origin.

We use these estimates to show that the system \eqref{equation2D} is
invertible and to show that the coefficient sequences are bounded.  To
ease notation, set
\begin{align*}
S_1 & = D_{\brac{2,0}} - D_{\brac{1,1}}D_{\brac{0,2}}^{-1}D_{\brac{1,1}},\\
S_2 & = D_{\brac{1,0}} - D_{\brac{1,1}}D_{\brac{0,2}}^{-1}D_{y},\\
S_3 & = D_0 + S_2^TS_1^{-1}S_2 - D_{\brac{0,1}}D_{\brac{0,2}}^{-1}D_{\brac{0,1}} \text{.}
\end{align*}
Note that $S_1$ is a Schur complement of $D_{\brac{0,2}}$ and that a standard
linear algebra identity gives
\begin{align*}
\tilde{D}_2^{-1} & = \MAT{S_1^{-1} & -S_1^{-1}D_{\brac{1,1}}D_{\brac{0,2}}^{-1}\\ -D_{\brac{0,2}}^{-1}D_{\brac{1,1}}S_1^{-1} & D_{\brac{0,2}}^{-1} +D_{\brac{0,2}}^{-1}D_{\brac{1,1}}S_1^{-1}D_{\brac{1,1}}D_{\brac{0,2}}^{-1}} \text{.}
\end{align*}
Using this and taking the Schur complement of $\tilde{D}_2$, which is
equal to $S_3$, the solution to \eqref{equation2D} can be written as
\begin{align*}
  \MAT{\alpha \\
    \beta} 
  & =\MAT{\Id \\ -\tilde{D}_2^{-1}\tilde{D}_1}
  \brac{D_0+\tilde{D}_1^T\tilde{D}_2^{-1}\tilde{D}_1}^{-1} v \quad
  \Leftrightarrow \quad  \MAT{\alpha \\
    \beta_1\\\beta_2}  =\MAT{\Id \\ -S_1^{-1}S_2 \\
    D_{\brac{0,2}}^{-1}\brac{D_{\brac{1,1}}S_1^{-1}S_2-D_{\brac{0,1}}}} S_3^{-1} v \text{.}
\end{align*}
Applying \eqref{infinfnorminv} from Section
\ref{subsec:fejersq_coeffs}, we obtain
\begin{align}
\normInfInf{D_{\brac{0,2}}^{-1}} & \leq  \frac{1}{\abs{\Ktwo_{\brac{0,2}}\brac{0}} - \normInfInf{\abs{\Ktwo_{\brac{0,2}}\brac{0}}\Id-D_{\brac{0,2}}}} \leq \frac{0.3399}{\fc^2} \text{,} \label{Dxx_inv_bound}
\end{align}
which together with $\Ktwo_{\brac{2,0}}\brac{0}=-\pi^2\fc \brac{\fc+4}/3$ and
\eqref{Dxy_bound} imply
\begin{align*}
\normInfInf{\abs{\Ktwo_{\brac{2,0}}\brac{0}}\Id-S_1} & \leq \normInfInf{\abs{\Ktwo_{\brac{2,0}}\brac{0}}\Id-D_{\brac{2,0}}} + \normInfInf{D_{\brac{1,1}}}^2\normInfInf{D_{\brac{0,2}}^{-1}} \leq 0.33624 \, \fc^2 \text{.}
\end{align*}
Another application of \eqref{infinfnorminv} then yields
\begin{align}
\normInfInf{S_1^{-1}} & \leq \frac{1}{\abs{\Ktwo_{\brac{2,0}}\brac{0}} - \normInfInf{\abs{\Ktwo_{\brac{2,0}}\brac{0}}\Id-S_1}} \leq \frac{0.3408}{\fc^2}\text{.} \label{S1inv}
\end{align}
Next, \eqref{Dx_bound}, \eqref{Dxy_bound} and \eqref{Dxx_inv_bound}
allow to bound $S_2$,
\begin{align*}
\normInfInf{S_2} & \leq \normInfInf{D_{\brac{1,0}}}+\normInfInf{D_{\brac{1,1}}}\normInfInf{D_{\brac{0,2}}^{-1}}\normInfInf{D_{\brac{0,1}}} \leq 8.142  \; 10^{-2}\, \fc\text{,}
\end{align*}
which combined with \eqref{Id_F_bound_2D}, \eqref{Dx_bound},
\eqref{Dxx_inv_bound} and \eqref{S1inv} implies
\begin{align*}
\normInfInf{ \Id -S_3} & \leq \normInfInf{\Id -D_0}+ \normInfInf{S_2}^2\normInfInf{S_1^{-1}}+\normInfInf{D_{\brac{0,1}}}^2\normInfInf{D_{\brac{0,2}}^{-1}}\leq 5.283 \; 10^{-2} \text{.}
\end{align*}
The results above allow us to derive bounds on the coefficient vectors by applying \eqref{infinfnorminv} one last time, establishing
\begin{align*}
\normInf{\alpha} & \leq  \normInfInf{S_3^{-1}} \leq  1 + 5.577 \; 10^{-2},\\
\normInf{\beta_1} & \leq \normInfInf{S_1^{-1}S_2S_3^{-1}} \leq \normInfInf{S_1^{-1}} \normInfInf{S_2} \normInfInf{S_3^{-1}} \leq  2.930 \; 10^{-2} \, \lambda_c , \\
\alpha_1 & = v_1 - \brac{\brac{\Id - S_3^{-1}}v}_1 \geq 1-\normInfInf{S_3^{-1}}\normInfInf{\Id - S_3} \geq 1- 5.577 \; 10^{-2}
\text{,}
\end{align*}
where the last lower bound holds if $v_1=1$. The derivation for $\normInf{\beta_2}$ is identical and we omit it.

\subsection{Proof of Lemma \ref{lemma:q_bound_2D_conc}}
\label{subsec:q_bound_2D_conc}

Since $v$ is real valued, $\alpha$, $\beta$ and $q$ are all real
valued. For $|r| \leq \rOne \, \lambda_c$, we show that the Hessian
matrix of $q$,
\[
H = \MAT{q_{\brac{2,0}}\brac{r} & q_{\brac{1,1}}\brac{r}  \\ q_{\brac{1,1}}\brac{r} & q_{\brac{0,2}}\brac{r}}
\]
is negative definite.  In what follows, it will also be useful to
establish bounds on the kernel and its derivatives near the
origin. Using \eqref{boundK_near}--\eqref{boundKp3_near}, we obtain
\begin{align*}
 \Ktwo \brac{x,y} & \geq \brac{1-\frac{\pi^2 \fc \brac{\fc+4}x^2}{6}}\brac{1-\frac{\pi^2 \fc \brac{\fc+4}y^2}{6}} \\
\Ktwo_{\brac{2,0}}\brac{x,y} & \leq \brac{-\frac{\pi^2 \fc \brac{\fc+4}}{3} + \frac{\brac{\fc+2}^4 \pi^4 x^2}{6}} \brac{1-\frac{ \pi^2 \fc \brac{\fc+4}y^2}{6}} 
\end{align*}
and
\begin{alignat*}{2}
  \abs{\Ktwo_{\brac{1,0}}\brac{x,y}} & \leq \frac{\pi^2\fc\brac{\fc+4}x}{3}, &
  \qquad \abs{\Ktwo_{\brac{1,1}}\brac{x,y}} & \leq \frac{\pi^4
    \fc^2\brac{\fc+4}^2 xy}{9},\notag\\ \abs{\Ktwo _{\brac{2,1}}\brac{x,y}} &
  \leq \frac{\pi^4 \fc^2\brac{\fc+4}^2 y}{9}, & \qquad \abs{\Ktwo
    _{\brac{3,0}}\brac{x,y}} & \leq \frac{\pi^4\brac{\fc+2}^4x}{3}.
\end{alignat*}
These bounds are all monotone in $x$ and $y$ so we can evaluate them
at $x=\rOne \, \lambda_c$ and $y=\rOne \, \lambda_c$ to show that for
any $|r| \le \rOne \, \lambda_c$, 
\begin{alignat}{3}
\Ktwo\brac{r} & \geq 0.8113 & \qquad \abs{\Ktwo_{\brac{1,0}}\brac{r}} & \leq 0.8113 & \qquad \Ktwo_{\brac{2,0}}\brac{r} & \leq -2.097 \, \fc^2, \notag\\
\abs{\Ktwo_{\brac{1,1}}\brac{r}} & \leq 0.6531 \, \fc, & \qquad \abs{\Ktwo _{\brac{2,1}}\brac{r}} & \leq 2.669 \, \fc^2, & \qquad \abs{\Ktwo _{\brac{3,0}}\brac{r}} & \leq  8.070\, \fc^3 \label{bound_K2D}.
\end{alignat}
The bounds for $\Ktwo_{\brac{1,0}}$, $\Ktwo_{\brac{2,0}}$,
$\Ktwo_{\brac{2,1}}$ and $\Ktwo_{\brac{3,0}}$ of course also hold for
$\Ktwo_{\brac{0,1}}$, $\Ktwo_{\brac{0,2}}$, $\Ktwo_{\brac{1,2}}$ and
$\Ktwo_{\brac{0,3}}$. Additionally, it will be necessary to bound sums
of the form $\sum_{r_j \in T/\keys{0}}\abs{\Ktwo
  _{\brac{\ell_1,\ell_2}}\brac{r-r_j}} $ for $r$ such that $\abs{r}
\leq \Delta/2$ and $\ell_1,\ell_2 = 0,1,2,3$.  Consider the case
$\brac{\ell_1,\ell_2} = \brac{0,0}$. Without loss of generality, let
$r = \brac{x,y}\in \R_{+}^2$. By Lemma \ref{lem:simple}, the
contribution of those ${r_j}$'s belonging to the three quadrants $\{|r
| >\Delta/2\} \setminus \R^2_+$ obeys
\[
\sum_{|r_j| > \Delta/2, r_j \notin \R_+^2} \abs{\Ktwo \brac{r-r_j}}
\le 3 \Bigl[\sum_{j \ge 0} B_0(\Delta/2 + j \Delta)\Bigr]^2.
\]
Similarly, the contribution from the bands where either $|r_{j,1}|$ or
$|r_{j,2}| \le \Delta/2$ obeys
\[
\sum_{|r_{j,1}| \le \Delta/2 \text{ or } |r_{j,2}| \le \Delta/2}
\abs{\Ktwo \brac{r-r_j}} \le 2\sum_{j \ge 1} B_0(j\Delta-\abs{r}).
\]
It remains to bound the sum over $r_j$'s lying in the positive
quadrant $\{|r | >\Delta/2\} \cap \R^2_+$. To do this, let
$f_1\brac{t}$ be equal to one if $\abs{t}\leq\Delta$ and to
$B_0(\Delta t-\abs{r})$ otherwise. Taking $f_2=f_1$, Lemma
\ref{lem:simple} gives 
\[
\sum_{|r_j| > \Delta/2, r_j \in \R_+^2} \abs{\Ktwo \brac{r-r_j}} \le
\sum_{j \ge 1} B_0(j\Delta-\abs{r})+ \Bigl[\sum_{j \ge 1} B_0(j
\Delta-\abs{r})\Bigr]^2.
\]
We can apply exactly the same reasoning to the summation of $\Ktwo
_{\brac{\ell_1,\ell_2}}$ for other values of $\ell_1$ and $\ell_2$,
and obtain that for any $r$ such that $\abs{r} \leq \Delta/2$, 
\begin{equation}
\sum_{r_j \in T\setminus\keys{0}}\abs{\Ktwo _{\brac{\ell_1,\ell_2}}\brac{r-r_j}}  \leq  Z_{\brac{\ell_1,\ell_2}}\brac{\abs{r}};\label{eq_Z}
\end{equation}
here, for $u \ge 0$,
\begin{align*}
  Z_{\brac{\ell_1,\ell_2}}\brac{u} & = 2\sum_{j
    \ge 1} K^{\brac{\ell_1}}_{\infty} B_{\ell_2}(j\Delta-u) + 2K^{\brac{\ell_2}}_{\infty} B_{\ell_1}(j\Delta-u) +K^{\brac{\ell_1}}_{\infty} \sum_{j\ge 1} B_{\ell_2}(j\Delta) \\
  & \qquad  + K^{\brac{\ell_2}}_{\infty} \sum_{j \ge 1} B_{\ell_1}(j\Delta) + 3
  \Bigl[\sum_{j \ge 0} B_{\ell_1}(\Delta/2 + j \Delta)\Bigr]
  \Bigl[\sum_{j \ge 0} B_{\ell_2}(\Delta/2 + j \Delta-u)\Bigr]\\
  & \qquad \qquad + \Bigl[\sum_{j \ge 1} B_{\ell_1}(j \Delta-u)\Bigr]
  \Bigl[\sum_{j \ge 1} B_{\ell_2}(j \Delta)\Bigr]
\end{align*}
in which $K^{\brac{\ell_1}}_{\infty}$ is a bound on the global maximum
of $K^{\brac{\ell_1}}$. The absolute value of the kernel $K$ and its second derivative reach their global maxima at the origin, so $K^{\brac{0}}_{\infty}=1$ and
$K^{\brac{2}}_{\infty}=\pi^2\fc \brac{\fc+4}/3$. Combining the bounds
on $\abs{K'}$ and $\abs{K'''}$ in Lemma \ref{lemma:fejersq_bounds}
with \eqref{boundKp_near} and \eqref{boundKp3_near}, we can show that
$K^{\brac{1}}_{\infty}=2.08 \brac{\fc+2}$ and
$K^{\brac{3}}_{\infty}=25.3\brac{\fc+2}^3$ if $\fc \geq
\minmTwoD$. Since $Z_{\brac{\ell_1,\ell_2}}=Z_{\brac{\ell_2,\ell_1}}$,
we shall replace $Z_{\brac{\ell_1,\ell_2}}$ for which $\ell_1>\ell_2$
by $Z_{\brac{\ell_2,\ell_1}}$.

Since
\[
q_{\brac{2,0}}\brac{r}  = \sum_{r_j \in T} \alpha_j \Ktwo_{\brac{2,0}}\brac{r-r_j} +
\beta_{1j} \Ktwo_{\brac{3,0}}\brac{r-r_j}+\beta_{2j} \Ktwo_{\brac{2,1}}\brac{r-r_j}
\]
it follows from \eqref{bound_K2D} and \eqref{eq_Z} that
\begin{align*}
  q_{\brac{2,0}}\brac{r}
  & \leq \alpha_1 \Ktwo_{\brac{2,0}}\brac{r} + \normInf{\alpha}\sum_{r_j \in T\setminus\keys{0}}\abs{\Ktwo_{\brac{2,0}}\brac{r-r_j}} \\
  & \qquad + \normInf{\beta} \Bigl[ \abs{\Ktwo_{\brac{3,0}}\brac{r}} +
  \sum_{r_j \in T\setminus\keys{0}} \abs{\Ktwo_{\brac{3,0}}
    \brac{r-r_j}} + \abs{\Ktwo_{\brac{2,1}}\brac{r}} + \sum_{r_j \in
    T\setminus\keys{0}}
  \abs{\Ktwo_{\brac{2,1}} \brac{r-r_j}}\Bigr] \\
  & \leq  \alpha_1 \Ktwo_{\brac{2,0}}\brac{r} + \normInf{\alpha}Z_{\brac{0,2}}(|r|) + \normInf{\beta} \brac{ \abs{\Ktwo_{\brac{3,0}}\brac{r}} +Z_{\brac{0,3}}(|r|) + \abs{\Ktwo_{\brac{2,1}}\brac{r}}+ Z_{\brac{1,2}}(|r|)} \\
  & \leq -1.175 \, \fc^2.
\end{align*}
The last inequality uses values of $Z_{\brac{\ell_1,\ell_2}}\brac{u}$
at $u = \rOne \, \lambda_c$ reported in Table \ref{table:Z}. By
symmetry, the same bound holds for $q_{\brac{0,2}}$. Finally, similar
computations yield
\begin{align*}
  \abs{q_{\brac{1,1}}(r)} & = \sum_{r_j \in T} \alpha_j \Ktwo_{\brac{1,1}}\brac{r-r_j} + \beta_{1j} \Ktwo _{\brac{2,1}}\brac{r-r_j}+\beta_{2j} \Ktwo _{(1,2)}\brac{r-r_j} \notag \\
  & \leq \normInf{\alpha}\Bigl[\abs{\Ktwo_{\brac{1,1}}\brac{r}} +
  Z_{(1,1)}(|r|) \Bigr] + \normInf{\beta}\Bigl[
  \abs{\Ktwo _{\brac{2,1}}\brac{r}} + \abs{\Ktwo _{\brac{1,2}}\brac{r}} + 2 Z_{(1,2)}(|r|)\Bigr]\\
   & \leq 1.059 \, \fc^2.
\end{align*}
Since $\tr(H) = q_{\brac{2,0}} + q_{\brac{0,2}} < 0$ and $\det(H) = |q_{\brac{2,0}}||q_{\brac{0,2}}| -
|q_{\brac{1,1}}|^2 > 0$, both eigenvalues of $H$ are strictly negative. 
\begin{table}[htbp]
\begin{center}
\begin{tabular}{| c | c | c | c | c | c |}
  \hline
  $Z_{\brac{0,0}}(u)$ & $Z_{\brac{0,1}}(u)$ & $Z_{\brac{1,1}}(u)$ & $Z_{\brac{0,2}}(u)$ & $Z_{\brac{1,2}}(u)$ & $Z_{\brac{0,3}}(u)$\\
  \hline
  $6.405 \; 10^{-2}$ & $0.1047 \, \fc$ & $0.1642 \, \fc^2$ &  $0.4019\, \fc$ & $0.6751 \, \fc^3$ &  $1.574 \fc^3$\\
  \hline	
\end{tabular}
\caption{
  Upper bounds on $Z_{\brac{\ell_1,\ell_2}}(u)$ at $\rOne \,
  \lambda_c$.}
  \label{table:Z}
\end{center}
\end{table}


We have shown that $q$ decreases along any segment originating at $0$.
To complete the proof, we must establish that $q > -1$ in the
square. Similar computations show
\begin{align*}
  q\brac{r} & = \sum_{r_j \in T} \alpha_j \Ktwo\brac{r-r_j} + \beta_{1j} \Ktwo_{\brac{1,0}}\brac{r-r_j}+\beta_{2j} \Ktwo_{\brac{0,1}} \brac{r-r_j} \notag \\
  &\geq \alpha_1 \Ktwo\brac{r} - \normInf{\alpha}Z_{\brac{0,0}}(|r|) -  \normInf{\beta}\brac{\abs{\Ktwo_{\brac{0,1}}\brac{r}} + \abs{\Ktwo_{\brac{1,0}}\brac{r}} + 2Z_{\brac{0,1}}(|r|)}\\
  & \geq 0.6447.
\end{align*}

\subsection{Proof of Lemma \ref{lemma:q_bound_2D}}
\label{subsec:q_bound_2D}

For $\rOne \, \lambda_c \le \abs{r} \le \Delta/2$,
\begin{align*}
  \abs{q} & \leq \Bigl| \sum_{r_j \in T} \alpha_j \Ktwo \brac{r-r_j} + \beta_{1j} \Ktwo_{\brac{1,0}}\brac{r-r_j}+\beta_{2j} \Ktwo_{\brac{0,1}}\brac{r-r_j}\Bigr| \\
  & \leq \normInf{\alpha}\Bigl[\abs{\Ktwo \brac{r}} +
  Z_{(0,0)}(|r|)\Bigr] +
  \normInf{\beta}\Bigl[\abs{\Ktwo_{\brac{1,0}}\brac{r}} +
  \abs{\Ktwo_{\brac{0,1}}\brac{r}} + 2 Z_{(0,1)}(|r|)\Bigr]. 
\end{align*}
Using the series expansion around the origin of $K$ and $K'$
\eqref{KboundUp}, we obtain that for $t_1 \le |r| \le t_2$,
\begin{align*} 
\abs{\Ktwo \brac{r}} & \leq \brac{1 - \frac{\pi^2 \fc \brac{\fc+4} x^2}{6} + \frac{\pi^4\brac{\fc+2}^4 x^4}{72}}\brac{1 - \frac{\pi^2 \fc \brac{\fc+4} y^2}{6} + \frac{\pi^4\brac{\fc+2}^4 y^4}{72}}\\
& \leq \brac{1 - \frac{\pi^2 \brac{1+\frac{2}{\fc}}^2 t_1^2}{6}\brac{1 - \frac{\pi^2\brac{1+\frac{2}{\fc}}^2 t_2^2}{12}}}^2,\\
\abs{\Ktwo_{\brac{1,0}}\brac{r}} & \leq \brac{\frac{\pi^2\fc\brac{\fc+4}t_2}{3}}^2 \text{.}
\end{align*} 
The same bound holds for $\Ktwo_{\brac{0,1}}$.  Now set
\[
W\brac{r} = \alpha^\infty \abs{\Ktwo \brac{r}} + 2\beta^\infty
\abs{\Ktwo_{\brac{1,0}}\brac{r}}
\]
where $\alpha^\infty$ and $\beta^\infty$ are the upper bounds from
Lemma \ref{lemma:coeff_bounds_2D}. The quantities reported in Table \ref{table:2D_bounds} imply that setting $\keys{t_1,t_2}$ to $\keys{\tC\, \lambda_c,0.27\, \lambda_c}$, $\keys{0.27\, \lambda_c,0.36 \, \lambda_c}$, $\keys{0.36 \, \lambda_c,0.56 \, \lambda_c}$ and $\keys{0.56 \, \lambda_c,\rTwo \, \lambda_c}$ yields $\abs{q}<0.9958$, $\abs{q}<0.9929$, $\abs{q}<0.9617$ and $\abs{q}<0.9841$ respectively in the corresponding intervals.  
Finally, for $\abs{r}$ between $\rTwo \,
\lambda_c$ and $\Delta/2$, applying Lemma \eqref{lemma:fejersq_bounds}
yields $W\brac{r} \leq 0.5619$, $Z_{\brac{0,0}}\brac{\rTwo \, \lambda_c} \le 0.3646$ and
$Z_{\brac{0,1}}\brac{\rTwo \, \lambda_c} \le 0.6502 \, \fc$ , so that $\abs{q} \leq 0.9850$. These
last bounds also apply to any location beyond $\Delta/2$ closer to $0$
than to any other element of $T$ because of the monotonicity
of $B_0$ and $B_1$. This concludes the proof.
\begin{table}[htbp]
\begin{center}
\begin{tabular}{| c | c | c | c | c |}
	\hline
	$t_1/\lambda_c$ & $t_2/\lambda_c$ & Upper bound on $W\brac{r}$ & $Z_{\brac{0,0}}\brac{t_2}$ & $Z_{\brac{0,1}}\brac{t_2}$\\
	\hline
	$\rOne $ & $ 0.27$ & $0.9203$ &  $6.561 \, 10^{-2}$ & $0.1074$ \\
	\hline
	$0.27 $ & $ 0.36$ & $0.9099$ &  $7.196 \, 10^{-2}$   & $0.1184$ \\
	\hline
	$0.36 $ & $ 0.56$ & $0.8551$ &  $9.239 \, 10^{-2}$   & $0.1540$ \\
	\hline
	$0.56 $ & $ \rTwo$ & $0.8118$ &  $0.1490$   & $0.2547$ \\
	\hline
\end{tabular}
\vspace{10pt}
\caption{Numerical quantities used to bound $\abs{q}$ between $\rOne$ and $\rTwo$.}
  \label{table:2D_bounds}
\end{center}
\end{table}

\end{document}